\documentclass[sn-mathphys]{sn-jnl}
\jyear{2021}

\theoremstyle{thmstyleone}%
\newtheorem{theorem}{Theorem}
\newtheorem{proposition}[theorem]{Proposition}%
\newtheorem{lemma}[theorem]{Lemma}
\newtheorem{corollary}[theorem]{Corollary}

\theoremstyle{thmstyletwo}%
\newtheorem{example}{Example}%
\newtheorem*{remark}{Remark}%

\theoremstyle{thmstylethree}%
\newtheorem{definition}{Definition}%
\newtheorem{conjecture}[theorem]{Conjecture} 

\usepackage{hyperref}
\hypersetup{colorlinks=true}
\usepackage{caption}
\usepackage{subcaption}
\usepackage[normalem]{ulem}

\usepackage{amsmath,amssymb,amsthm,bm}
\usepackage{thmtools}
\usepackage{mathtools}
\usepackage{enumerate}

\newcommand{\com}[1]{} 

\newcommand{\RR}{\mathbb{R}}
\newcommand{\PP}{\mathbb{P}}
\newcommand{\EE}{\mathbb{E}}

\DeclarePairedDelimiter\abs{\lvert}{\rvert}

\newcommand{\lsa}{\text{LSA}}
\newcommand{\bt}{\text{BT}}
\newcommand{\udp}[1]{\overset{#1}{\longleftrightarrow}}

\usepackage{xcolor}

\definecolor{changecolor}{RGB}{192,64,0}

\begin{document}
\title[Identifiability of phylogenetic networks from average distances]{Identifiability of local and global features of phylogenetic networks from average distances}

\author*[1]{\fnm{Jingcheng} \sur{Xu}}

\author[1,2]{\fnm{C\'ecile} \sur{An\'e}}

\affil*[1]{\orgdiv{Department of Statistics}, \orgname{University of
    Wisconsin - Madison}, \orgaddress{\city{Madison},
    \postcode{53706}, \state{WI}, \country{United States}}}

\affil[2]{\orgdiv{Department of Botany}, \orgname{University of
    Wisconsin - Madison}, \orgaddress{\city{Madison},
    \postcode{53706}, \state{WI}, \country{United States}}}

\abstract{ 
Phylogenetic networks extend phylogenetic trees to 
model non-vertical inheritance, by which a lineage inherits
material from multiple parents. The computational complexity
of estimating phylogenetic networks from genome-wide data
with likelihood-based methods limits the size of networks
that can be handled. Methods based on pairwise distances
could offer faster alternatives.
We study here the information that average pairwise distances
contain on the underlying phylogenetic network, by characterizing
local and global features that can or cannot be identified.
For general networks, we clarify that the root
and edge lengths adjacent to reticulations are not identifiable,
and then focus on the class of  zipped-up semidirected networks.
We provide a criterion to swap subgraphs locally, such
as 3-cycles, resulting in indistinguishable networks.
We propose the ``distance split tree", which can be constructed
from pairwise distances, and prove that it is a refinement
of the network's tree of blobs, capturing the tree-like features
of the network.
For level-1 networks, this distance split tree is equal to the
tree of blobs refined to separate polytomies from blobs,
and we prove that the mixed representation of the network
is identifiable. The information loss is localized around 4-cycles,
for which the placement of the reticulation is unidentifiable.
The mixed representation combines split edges for 4-cycles,
regular tree and hybrid edges from the semidirected
network, and edge parameters that encode all information
identifiable from average pairwise distances.
}


\maketitle

\section{Introduction}
Phylogenetic trees represent the past history of a set of
organisms and are central to the field of evolutionary biology.
Phylogenetic networks offer a convenient framework to
extend phylogenetic trees, in which extra edges explicitly represent
the various biological processes by which an ancestral organism
or population inherits genetic material from several parents.
With the advent of genome-wide data that can be collected across
many organisms, there is robust evidence for hybridization and
gene flow in many groups, and rooted phylogenetic networks are now
widely used \cite{2018folk-hybridization,2020blair-review}.

Inferring phylogenetic networks is hard, however.
Computing times are prohibitive with more
than a handful of taxa for likelihood-based approaches,
such as full likelihood or Bayesian methods in
PhyloNet or SnappNet
\cite{Sol_s_Lemus_2016_infer,cao2019-phylonet-practical,rabier2021-snappnet}.
Methods based on pairwise distances have the potential to be
much faster \cite{2004bryant-neighbonet}.
For inferring phylogenetic trees, Neighbor-Joining
and other distance-based methods
\cite{1987satounei-NJ,2004despergascuel-fastme}
are orders of magnitude faster than likelihood-based methods
and can handle data with many more taxa,
even if their speed might be at the cost of accuracy
\cite[but see][]{2017rusinko}.

We study here the information carried by average pairwise
distances about the underlying phylogenetic network.
In other words, we ask whether phylogenetic networks are
identifiable and what can be known about the network, theoretically,
from distances between pairs of taxa, averaged across the trees
displayed in the network. Trees and their branch lengths are identifiable
from distances \cite{charles03_phylog}. 
Trees form the simplest class of networks.
How much sparseness must be imposed on networks to
maintain identifiability?

Much previous work has focused on using shortest distances
\cite{Bordewich_2018,2017chang-shortestdistance-ultrametric},
sets or multisets of distances
\cite{2016bordewichsemple-intertaxadistances,bordewich16_algor_recon_ultram_tree_child,2018bordewich-treechild-multisetdistances}
or the logdet distance \cite{allman2022-identifiability-logdet-net}.
Average distances were used for network inference
but without theoretical guarantees \cite{willems14_new_effic_algor_infer_explic}.
Willson studied the identifiability of parameters from average distances
when the network topology is known \cite{willson12_tree_averag_distan_certain_phylog}.
Other previous work has focused on the full identifiability
of the network, thereby imposing strong constraints,
such as a single reticulation
\cite{willson13_recon_certain_phylog_networ_from,francis-steel-2015}.
To obtain general results, we focus on the identifiability
(or lack thereof) of local features and of global features,
without necessarily asking for the full identifiability of the network.
We also study the identifiability of branch lengths and inheritance values,
often understudied in previous work.

We highlight here some of our results.
Notably, we show that the root of the network is not generally identifiable
from average distances. This is well-known for trees but has not been
clarified by prior work on networks, which assumed data
available at the root or a known outgroup
\cite{willson13_recon_certain_phylog_networ_from,2018bordewich-treechild-multisetdistances}
or the network being ultrametric
\cite[e.g.][]{2005chan-ultrametric-galled-distancematrix,bordewich16_algor_recon_ultram_tree_child,Bordewich_2018,allman2022-identifiability-logdet-net},
or equal edge lengths (without any degree-2 nodes
except perhaps for the root) \cite[e.g.][]{2016bordewichsemple-intertaxadistances}.
Therefore, we focus our study on semidirected networks, in which the root
is suppressed and edges are undirected except for hybrid edges
\cite{Sol_s_Lemus_2016_infer}.

Without any restriction on the network complexity,
we prove that we may swap a local subgraph with another
without altering average distances, provided that the swapped subgraphs
have the same pairwise inheritance and distance matrices at their boundary.
We apply this swap result to subgraphs with a small boundary, showing
that degree-2 blobs, degree-3 blobs and 3-cycles are not identifiable;
and showing that level-2 networks are not identifiable from
average distances, not even generically.
This result provides a simple explanation for the reticulate exceptions
that are permitted in a network whose average distances fit on a tree
in \cite{francis-steel-2015}.
We anticipate that the application of our swap lemma
will lead to other applications, using larger subgraphs,
finding local structures that prevent the identifiability of the network
from average distances.

For the global structure of the network, we prove that a
refinement of the network's tree of blobs is identifiable (under mild
assumptions) which we call the ``distance split tree''.
Informally, any cycle in the network is condensed into a single
node of the tree of blobs, which encodes the tree-like parts of the
network. While the tree of blobs provides limited
knowledge about the network, it could be leveraged to develop
divide-and-conquer approaches. Namely, once a blob is identified
from the tree of blobs using average distances, accurate estimation
methods could be applied to a subset of taxa that cover a given blob,
that may be computationally feasible on the subsample.
Combining different types of methods to estimate
different features of the networks (such as the global tree of blobs
and small subnetworks) may lead to efficient strategies for accurate
and computationally efficient network estimation methods.

Beyond the topology, we prove that only one composite parameter
can be identified from average distances, out of the lengths of
all the parent edges and the child edge adjacent to a hybrid node.
This means that average distances lose extra ``degrees of freedom"
compared to information from displayed trees, for example,
because ``sliding" a reticulation along two parent edges affects edge lengths
in displayed trees \cite{Pardi_2015} and
affects distance \emph{sets} \cite{bordewich16_algor_recon_ultram_tree_child},
but does not affect average distances.
We show that the ``zipped up" version of a network, in which
all hybrid edges have length $0$, does not depend
on the order in which reticulations are zipped up.
Prior work has already constrained
hybrid edges to have length $0$, but arguing that this assumption
is biologically motivated \cite{willson13_recon_certain_phylog_networ_from}.
The zipped up network can be thought of as a canonical version
to be inferred by estimation methods. Such methods
will need to communicate to users that hybrid edge lengths are not
\emph{assumed} to be $0$ ---because many biological
scenarios can lead to positive lengths on hybrid edges,
but are instead \emph{constrained} to be $0$
(or solely influenced by a prior distribution) because they lack
identifiability from average distances.
Future work could consider interactive visualizations that
allow users to zip and slide each reticulation,
to explore the full equivalence class of networks represented by
their zipped-up version.

Finally, we study level-1 networks, in which distinct cycles
don't share nodes and each blob is a single cycle. The topology of
level-1 networks has been shown to be identifiable
(up to some aspects of small cycles)
from quartet concordance factors \cite{2019banos},
logdet distances \cite{allman2022-identifiability-logdet-net}
or some Markov models \cite{gross20_distin_level_phylog_networ_basis}.
We show here that, if internal tree edges have positive lengths
(which can be achieved by creating potential polytomies),
level-1 networks are identifiable from average distances,
except for local features around small cycles.
Namely, neither the direction of hybrid edges within 4-cycles, nor the
parameters (length and inheritance) of edges in and adjacent to
4-cycles are identifiable.
We introduce the ``mixed representation" of a level-1 network
in which 4-cycles are represented by split subgraphs,
whose parallel edges are split edges, with identical
edge lengths and no inheritance values.
These mixed networks formalize the class of network topologies
used in \cite{2019banos,Allman_2019} and in
\cite{allman2022-identifiability-logdet-net}.
We show that the mixed representation of a level-1 network
is identifiable from average distances, including its edge parameters.
Here again, future work on interactive visualizations could let
users re-assign a hybrid node within a 4-cycle, to help explore the
class of phylogenetic networks with a given mixed representation.

We conjecture that the tree of blobs is identifiable from many
other data types, such as distance sets (multiple distances for
each pair of taxa), the logdet distance and other distances.
It would be interesting to characterize the general properties
that a distance function needs to satisfy, for the distance split tree
derived from this distance to identify a relevant refinement of
the network's tree of blobs.
Given the complexity of inferring phylogenetic networks,
we hope that our study of global and local features of the network
will spur the development of new divide-and-conquer approaches.

Notations, main results and implications
are presented in section~\ref{sec:notations-main}.
The proofs and more formal definitions
are presented in section~\ref{sec:nonid-struct} for
non-identifiable features,
section~\ref{sec:ident-blob-tree} for the identifiability
of the tree of blobs, section~\ref{sec:k-sunlet} for the study of sunlets, and
section~\ref{sec:level-1} for 
level-1 networks.
More technical proofs are in the appendix.

\section{Notation and main results}
\label{sec:notations-main}

\subsection{Phylogenetic networks}
\label{sec:network-def}

We use standard definitions for graphs and phylogenetic networks
as in \cite{steel16_phylogeny}, with slight modifications,
and notations mostly following \cite{2019banos}.

\begin{definition}[rooted network]
  \label{def:rooted-network}
  A \emph{topological rooted phylogenetic network} (``rooted network''
  for short) on taxon set $X$ is a tuple $(N^+, f)$.  $N^+$ is a
  rooted directed acyclic graph with vertices
  $V = \{r\} \sqcup V_L \sqcup V_H \sqcup V_T$ and $f: X \to V_L$ a
  labelling function, where
  \begin{itemize}
  \item $r = \rho(N^+)$ is the root, the unique vertex in $N^+$ with
    in-degree $0$;
  \item $V_L$ are the leaves (or ``tips''), the vertices with
    out-degree $0$. We also require that leaves all have in-degree
    $1$;
  \item $V_T$ are the tree nodes, the vertices with in-degree $1$ that
    are not leaves;
  \item $V_H$ are the hybrid nodes, the vertices with in-degree larger
    than $1$;
  \item $f$ is a bijection between $X$ and $V_L$.
  \end{itemize}
  An edge $(a, b)$ is a \emph{tree edge} if its child $b$ is a tree
  node or a leaf node, and a \emph{hybrid edge} otherwise.  We denote
  the set of tree edges by $E_T$, and the set of hybrid edges by
  $E_H$.  We will also write $ab$ for the edge $(a, b)$ when no
  confusion is likely.
  A \emph{polytomy} is a non-root node of degree $4$ or higher, or
  the root $r$ if $r$ is of degree $3$ or higher.
  An \emph{internal} edge is an edge that is not incident to a
  leaf.
  A \emph{partner} edge of a hybrid edge $e$ is a hybrid edge $\tilde e \neq e$
  having the same child as $e$.
\end{definition}

\noindent
Definition~\ref{def:rooted-network} differs from \cite{steel16_phylogeny} in
that we allow for degree-$2$ nodes in a network, and also
require the leaves to have in-degree exactly $1$.  The reason for this
requirement is technical: when a leaf is incident to a
pendant edge, it forms a standalone ``blob'', which is
defined later.  When no confusion is likely, we refer to the
rooted network $(N^+, f)$ as $N^+$. Note that parallel edges are allowed.

For two nodes $a, b$ in a rooted network $N^+$, we write $a \leq b$
and say that $a$ is \emph{above} $b$ if
there is a directed path from $a$ to $b$. We write $a < b$ if
$a \leq b$ and $a \neq b$.
For a set of nodes $W$ in a rooted network $N^+$,
let $D$ be the set of nodes that lie on all paths from the root to
the elements of $W$.
The greatest element of $D$ (i.e.\ the node $s \in D$
such that $s \geq t$ for all $t \in D$) is called the \emph{lowest stable
ancestor} of $W$, or $\lsa(W)$ \cite[p.263]{steel16_phylogeny}.

As in the case with phylogenetic trees, we can unroot a rooted
phylogenetic network to obtain a \emph{semidirected phylogenetic
  network}, or ``semidirected network'' for short (see Fig.~\ref{fig:lsa_semidirected_nets}).

\begin{figure}[h]
  \centering
  \includegraphics[scale=1.3]{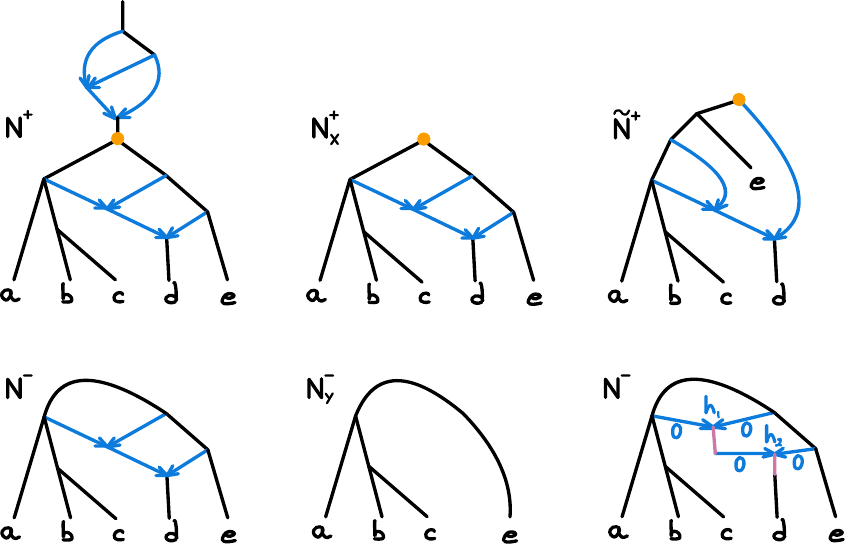}
  \caption{Example rooted network $N^+$ on $X=\{a,b,c,d,e\}$;
  LSA network $N^+_X$; semidirected network $N^-$,
  subnetwork $N^-_Y$ on $Y=\{a,b,c,e\}$; network $\tilde N^+$
  obtained by rerooting $N^-$ (on one of the hybrid edges);
  and a display of $N^-$ illustrating the modification of edge lengths
  to zip-up $N^-$. Hybrid edges are in blue with arrows.
  Lowest common ancestor nodes are large (orange) dots.
  }
  \label{fig:lsa_semidirected_nets}
\end{figure}

\begin{definition}[semidirected network]
  \label{def:sdn}

  A \emph{semidirected graph} $G^- = (V, E)$ is a tuple where $V$ is
  the set of nodes, and $E = E_D \sqcup E_U$ with a set $E_D$
  of \emph{directed edges} and a set $E_U$ of \emph{undirected
    edges}.  $E_D$ consists of ordered pairs $(a, b)$ where
  $a, b \in V$. In contrast, $E_U$ consists of unordered pairs
  $\{a, b\}$, such that if $\{a, b\} \in E_U$, then
  $(a, b)\not\in E_D$, i.e.\ an edge cannot be both directed and
  undirected.

  Let $(N^+, f)$ be a rooted network on $X$. The \emph{topological
    semidirected phylogenetic network induced from $(N^+, f)$} is a
  tuple $(N^-, f)$, where $N^-$ is the semidirected graph obtained by:
  \begin{enumerate}
  \item removing all the edges and nodes above $\lsa(X)$;
  \item undirecting all tree edges $e \in E_T$, but keeping the
    direction of hybrid edges;
  \item suppressing $s = \lsa(X)$ if it has degree $2$: if $s$ is
    incident to two tree edges, then remove $s$ and replace the two
    edges with a single tree edge; if $s$ is incident to one tree edge
    and one hybrid edge, then remove $s$, and replace the two edges by
    a hybrid edge with the same direction as the original hybrid edge.
    (Note that $s$ may not be incident to two hybrid edges if it has
    degree 2 by Lemma 1 in \cite{2019banos}).
  \end{enumerate}
  For a semidirected graph $M^-$ with vertex set $V$ and labelling function
  $g: X \to V$,
  $(M^-, g)$ is a \emph{topological semidirected phylogenetic network}
  if it is the semidirected network induced from some rooted network.
\end{definition}

\begin{remark}
  An alternate definition may consider skipping step 1, that is,
  retain nodes and edges above the lowest stable ancestor,
  for a more general class of semi-directed networks.  We consider
  step 1 because the subgraph above the LSA is not identifiable
  from 
  pairwise distances. Other authors make assumptions that are similar
  to performing step 1, such as assuming the network is ``proper"
  (every cut-edge and cut-vertex induces a non-trivial split of $X$
  \cite{francis-moulton-2018,
    fischer18_unroot_non_binar_tree_based_phylog_networ}, or
  ``recoverable", i.e.\ $\lsa(N)=\rho(N)$
  \cite{huber14_how_much_infor_is_needed})
\end{remark}

  For a semidirected network $N^-$ induced from $N^+$, the
  sets $V_L$, $V_H$ and $V_T$ are still well defined:
  $V_L(N^-)=V_L(N^+)$ is the set of nodes with degree 1, thanks to our
  requirement that leaves must be of degree 1 in a rooted network, and
  because a root of degree-1 would be above the $\lsa(X)$ in the
  rooted network. Hybrid nodes $V_H(N^-)=V_H(N^+)$ remain well-defined in a
  semidirected network, because hybrid edges are directed and point to
  hybrid nodes.  $V_T(N^-)$ is the set of all the other nodes, and may
  include the original root.
  The notion of child (node or edge) is also well defined for hybrid nodes
  in semidirected networks. Indeed, the child edges of a hybrid
  node are all the incident tree edges and outgoing hybrid edges.
  Consequently, the notion of \emph{tree-child}
  network \cite{steel16_phylogeny} also carries over.

For a rooted network $N^+$ on $X$, the \emph{LSA  network} $N^+_X$ of $N$ is the rooted network
obtained from $N^+$ by removing everything above $\lsa(X)$ in $N$ \cite{2019banos}.
If $N$ has the property that $\rho(N) = \lsa(X)$, then we call $N$ an \emph{LSA network}.
One immediate consequence of these definitions is that the semidirected
network induced from $N^+$ and $N^+_X$ are the same.  Furthermore,
every semidirected network can be induced from an LSA network.

The \emph{unrooted} graph $U(N)$ induced from a directed or semidirected
graph $N$ is the undirected graph obtained from $N$ by undirecting all edges in $N$.
Because rooted networks are DAGs, there cannot be directed cycles in
rooted or semidirected networks.
A \emph{cycle} in a rooted or semidirected network $N$ is defined to be a
subgraph $C$ of $N$, such that $U(C)$ is a cycle.

One may also consider \emph{rerooting} a semidirected network $N^-$: either
at a node or on an edge
\cite{gambette2012_quartets_unrooted}. Specifically,
\emph{rerooting at node $s$} refers to designating a node $s$ in
$N^-$ as root and directing all undirected (tree) edges away from $s$,
if this leads to a valid rooted network.
\emph{Rerooting on edge $uv$} refers to adding a new node $s$,
replacing $uv$ by two edges $us$ and $sv$, and finally rerooting at
node $s$.  It follows from Definition~\ref{def:sdn} for semidirected
network $N^-$ that there exists either a node $u$ or an edge $e$ such
that rerooting at $u$ or rerooting on $e$ gives an LSA network $N$
which induces $N^-$: we can reroot at $\lsa(X)$ if it is not
suppressed, or otherwise reroot at the edge $e$ where $\lsa(X)$ is
suppressed.
Note that while there is always a rerooting of $N^-$
that gives a rooted LSA network, not all rerootings give an LSA
network.

Because this work focuses on semidirected networks, in the
later sections for notational convenience we will usually denote
a semidirected network without the superscript, i.e.\ $N$ instead of
$N^-$, and use $N^+$ for an LSA network that induces $N$, which is
obtained from rerooting at a node or on an edge.

A rooted network is \emph{binary} if its root has degree $2$ and all
the other nodes, except for leaves, have degree $3$.  A semidirected
network is \emph{binary} if all its nodes, except for leaves, have
degree $3$.  The semidirected network induced by a binary rooted
network $N^+$ is binary.  On topological phylogenetic networks, we can
further assign edge lengths and hybridization parameters, also called
inheritance probabilities, to obtain \emph{metric phylogenetic
  networks}.

\begin{definition}[metric]\label{def:net-metric}
  A metric on a rooted or semidirected network $N$ is a pair of
  functions $(\ell, \gamma)$, with $\ell: E \to \RR_{\geq 0}$
  assigning lengths to edges, and $\gamma: E_H \to (0, 1)$ assigning
  hybridization parameters to hybrid edges.  The hybridization
  parameter $\gamma(e)$ for a hybrid edge $e$ represents the
  proportion of genetic material that the child inherits through the
  edge.  As a result, we require that for a hybrid node $v$,
  $\sum_{e \in E_H(v)}\gamma(e) = 1$, where $E_H(v)$ denotes the set
  of incoming hybrid edges for $v$.  We define $\gamma(e) = 1$ for any
  tree edge $e$, to extend the function $\gamma: E \to [0, 1]$ to all
  edges of $N$.  A rooted/semidirected network with a metric is called
  a \emph{metric rooted/semidirected network}.
\end{definition}

In a metric semidirected network, when a node is suppressed (see
step~3 in Def.~\ref{def:sdn}), the length of the new edge is the sum
of the original two edges. The hybridization parameter is
unchanged for hybrid edges.

Two metric and/or semidirected networks are isomorphic if the
(semi)directed graphs are isomorphic with an isomorphism that preserves the labelling and
the metric.  We regard isomorphic networks as identical, as we only
identify networks and their properties up to isomorphism.

\begin{definition}[blob, level, tree of blobs]\label{def:blob}
  A \emph{blob} $B$ in a rooted or semidirected network $(N, f)$ is a subgraph of $N$ such
  that $U(B)$ is a $2$-edge-connected component of $U(N)$.
  A blob is \emph{trivial} if it has a single node.
  The \emph{edge-level} (or simply \emph{level})
  of a blob $B$ is the number of edges in $B$ one needs
  to remove in order to obtain a tree (i.e. $\lvert E_B \rvert -
  \lvert V_B \rvert + 1$, where
  $E_B, V_B$ are the edge set and node set of the blob $B$).  The
  \emph{level} of a network is the maximum level of all its blobs.
  The \emph{tree of blobs} $\bt(N)$ of a network $N$ is an undirected graph
  where each vertex is a blob of $N$, and where two vertices $B_1$ and
  $B_2$ are adjacent if there is an edge $b_1b_2$ or
  $b_2b_1$ in $N$ such that $b_1 \in B_1$ and $b_2 \in B_2$.
  The \emph{degree of a blob} is the degree of the corresponding vertex
  in the tree of blobs.
\end{definition}

\begin{remark}
If $N^+$ is an LSA network and $N^-$ is induced from it, then
$N^+$  and $N^-$  have the same blobs and the same tree of blobs,
because they have the same undirected graphs.
\end{remark}

Recall that a graph is \emph{$2$-edge-connected} if the removal of one edge
does not disconnect the graph. A \emph{$2$-edge-connected component} is a
maximal $2$-edge-connected subgraph.
Our definition of level follows \cite{gambette2012_quartets_unrooted} and
is nonstandard in using $2$-edge-connected
components rather than biconnected components.
A graph is \emph{biconnected} if the removal of one vertex does
not disconnect the graph. A \emph{biconnected component} of a graph,
or \emph{block}, is a maximal biconnected subgraph.
Any block of $3$ or more nodes
is $2$-edge-connected,
so each non-trivial block maps to a single blob
and each blob may be formed by one or more block(s).
Therefore, the traditional level based on biconnected components
is lower than or equal to the level used here.
However, the two definitions agree on binary networks.
For binary networks, the level of a blob $B$ is the same as the number
of hybrid nodes in $B$ \cite{gambette2012_quartets_unrooted}.
If hybrid nodes may have more than two parents, the level
of a blob could be greater than its number of hybrid nodes.

  The ``tree of blobs'' was first defined by
  \cite{gusfield07_decom_theor_phylog_networ_incom_charac},
  using blocks and after modifying the network with edges to
  separate overlapping blocks. It is easy to verify that non-trivial blocks and blobs
  are identical after these modifications.
  Despite the similar name and construction,
  the tree of blobs is different from the ``blob tree" defined in
  \cite{murakami19_recon_tree_child_networ_from}.

\com{ 
It would be interesting to know if the traditional level is the minimum level
of all ``resolutions" of a semidirected network, whereas our level
is the maximum.
And then note that each resolution can be made to have identical
average distances (when the added edges are given 0).
}

Unlike blocks,
blobs partition the nodes in $N$
and provide a convenient mapping of edges from the tree of blobs
to the network, as we will show later.
Fig.~\ref{fig:level-example} (left) shows a non-binary network with one blob
of level $2$, 
but with two level-$1$ blocks. There are $3$ ways to refine
this network into a binary network, one of which is of level $1$ with $2$ blobs
(Fig.~\ref{fig:level-example} right), and the other two are of level $2$ with
a single non-trivial blob (and block).
Fig.~\ref{fig:blobtree-example} shows a level-$1$ network with $3$ blobs (left)
and its tree of blobs (top right).
Note that both networks on the top row have the same block-cut tree
(derived from blocks and cut nodes, see \cite{diestel-graphtheory})
after suppressing its degree-2 nodes,
and both networks at the bottom have a block-cut tree
reduced to a star, after suppressing degree-2 nodes.

\begin{figure}
  \centering
  \includegraphics[scale=1.5]{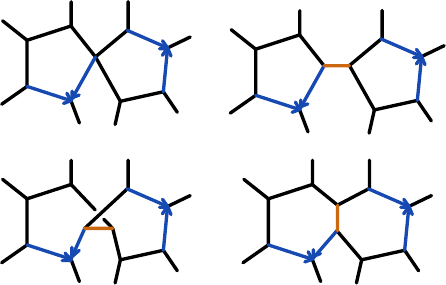}
  \caption{Example networks and their levels.
  Top left: network $N$ with 1 blob but 2 biconnected components.
  Hybrid edges are shown in blue with arrows.
  Top right: one possible resolution of $N$,
  with identical average distances if the added edge (orange, cut edge)
  is assigned length 0. It has 2 blobs and is of level 1.
  Bottom: the other 2 resolutions of $N$, with identical average distances
  if the added edge (orange) is assigned length 0.
  Both have $1$ blob and are of level $2$.
  }
  \label{fig:level-example}
\end{figure}

\begin{figure}
  \centering
  \includegraphics[scale=1.5]{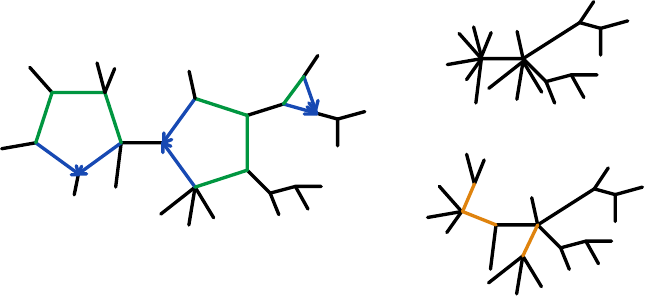}
  \caption{Example tree of blobs.
  Left:  level-1 network $N$, with leaf labels omitted to avoid clutter.
  Hybrid edges are shown in blue; cut edges in black.
  The two 5-cycles are identifiable provided that their tree (green) edges
  have positive length
  (e.g. Corollary~\ref{thm:level1-5up}).
  Top~right: tree of blobs $T$ for $N$.
  One of its degree-3 nodes corresponds to a degree-3 blob in $N$,
  undetectable from average distances.
  Bottom right: distance split tree
  reconstructed from average distances.
  It is a refinement of $T$
  (and is $N$'s block-cut tree after suppressing degree-2 nodes).
  The extra edges (orange)
  correspond to polytomies in $N$.}
  \label{fig:blobtree-example}
\end{figure}

\subsection{Average distances}

\begin{definition}[up-down path, rooted network, from \cite{Bordewich_2018}]
\label{def:up-down-path-rooted}
In a rooted network $N^+$, an \emph{up-down path}
between two nodes $u$ and $v$ is a sequence of distinct nodes
$u = u_1u_2\dots u_n = v$ with a special node $s = u_i$ such that
$u_i\dots u_2u_1$ and $u_i\dots u_{n-1}u_n$ are directed paths in
$N^+$.
\end{definition}
If $u_1u_2\dots u_n$ is an up-down path, we may write $u \leftarrow s \rightarrow v$
or $u \leftrightarrow v$ as a shorthand.
Particularly, a directed path between $u$ and $v$ is also an up-down
path, and we will simply write $u \rightarrow v$ or
$u \leftarrow v$, depending on the direction of the path.  Note that
the up-down paths $u_1u_2\dots u_n$ and $u_nu_{n-1}\dots u_1$ are
considered to be the same.  Formally, we can define up-down paths as
the equivalence classes of these sequence of nodes, with reversal of
the sequence being an equivalence relation.

It is not obvious whether the notion of up-down paths is still valid in
semidirected networks: given an up-down path
$p = u_0 u_1\dots u_n$ in a rooted network $N^+$, is it possible to tell if $p$ is an
up-down path by looking at the induced semidirected network $N^-$ alone?
It turns out the answer is yes: the notion of an up-down path
only has to do with the semidirected structure of a network.
An alternative definition that also applies to the semidirected
networks is the following:

\begin{definition}[up-down path, semidirected network]
  \label{def:up-down-path}
  Let $N$ be a rooted or semidirected network.  An \emph{up-down path}
  is a path of distinct nodes with no v-structure. More formally, 
  $u_0 u_1\dots u_n$ is an \emph{up-down path} if the $u_i$'s are distinct;
  for each $i$, either $u_iu_{i+1}$ or
  $u_{i+1}u_i$ is an edge in $N$ (for a tree edge $uv$ in semidirected
  network $N$, both $uv$ and $vu$ are valid edges in $N$);
  and there is no segment $u_{i-1}u_iu_{i+1}$ such that $u_i$
  is a hybrid node and $u_{i-1}u_i$ and $u_{i+1}u_i$ are hybrid
  (directed) edges in $N$.
  An up-down path with no hybrid nodes is a \emph{tree path}.
\end{definition}

\noindent This following equivalence is proved in appendix~\ref{sec:path-dist-proofs}.

\begin{proposition}\label{prop:up-down-path}
Let $N^+$ be a rooted network with induced semidirected network $N^-$.
$p$ is an up-down path in $N^+$ according to Definition~\ref{def:up-down-path-rooted}
if and only if it is an up-down path according to Definition~\ref{def:up-down-path}.
\end{proposition}

Given a metric $(\ell, \gamma)$ on a network $N$ and a up-down path $p$,
we can define the path length $\ell(p) = \sum_{e \in p} \ell(e)$
where $e$ ranges over the edges in path $p$.
We also define the
path probability $\gamma(p) = \prod_{e \in p} \gamma(e)$ as
the product of all the hybridization parameters of the component
edges.  Here we use the convention of $\gamma(e) = 1$ when $e$ is a
tree edge.
$\gamma(p)$ is the probability of path $p$ being present
in a random tree extracted from $N$, where tree ``extraction" proceeds as follows:
at a hybrid node $h$, we pick one of $h$'s parent hybrid edge according to
the edges' hybridization parameters, and delete all other parent edges of $h$.
If we do this independently for all hybrid nodes, then the result is a random
tree $T$ with the same nodes as $N$.
In $T$, there is a unique path
between $u$ and $v$, which equals $p$ with probability exactly $\gamma(p)$.

In the special case that there is a tree path $p$ between nodes $u$
and $v$ in $N$, then we immediately have $\gamma(p) = 1$.  In fact, $p$ must
be the unique up-down path between $u, v$: because $p$ does not
contain any hybrid nodes, $p$ is the unique path between $u, v$ on any
displayed tree $T$.

\begin{definition}[average distance]
Let $N$ be a rooted or semidirected network.
The \emph{average distance} between two nodes $u$ and $v$ in $N$
is defined as
$$d(u, v) = \sum_{p \in P_{uv}} \gamma(p)\ell(p)$$
where $P_{uv}$ denotes the set of up-down paths between $u$ and $v$.
Equivalently, this is the expected distance between $u$ and $v$ on a
random tree $T$ extracted from $N$ (described above).
As a result, $d$ satisfies the triangle inequality.
We may write $d_N$ to emphasize the dependence on $N$.
\end{definition}

  The same definition was used for rooted networks by
  \cite{willson12_tree_averag_distan_certain_phylog}.  By considering
  our extended definition of up-down paths, our definition clarifies
  that average distances are well-defined on semidirected networks.

\begin{remark}
In a network $N$, contracting a tree edge of length 0 creates a network
$\widetilde N$ that has a polytomy, but whose up-down paths are in bijection
with those of $N$ 
and such that $N$ and $\tilde N$ have identical average distances.
Consequently, a polytomy at a tree node of degree 4 has 3 distinct resolutions
with identical average distances.
This is not true for hybrid edges of length 0: hybrid edges may not be
contracted without modifying the set of up-down paths.
Moreover, if there is a polytomy at a hybrid
node with 2 incoming hybrid edges and 2 other (outgoing) edges (Fig.~\ref{fig:polytomy} left)
then there is a single resolution of this polytomy with identical
set of up-down paths and identical
pairwise distances: with the addition of a tree edge (Fig.~\ref{fig:polytomy} center).
The resolutions shown in Fig.~\ref{fig:polytomy} (right) are not equivalent:
there exist up-down paths $a\rightarrow c$, $a\rightarrow d$,
$b\rightarrow c$ and $b\rightarrow d$ in the network on the left,
but each network on the right is missing one of these paths.

\end{remark}

\begin{figure}
  \centering
  \includegraphics[scale=1.2]{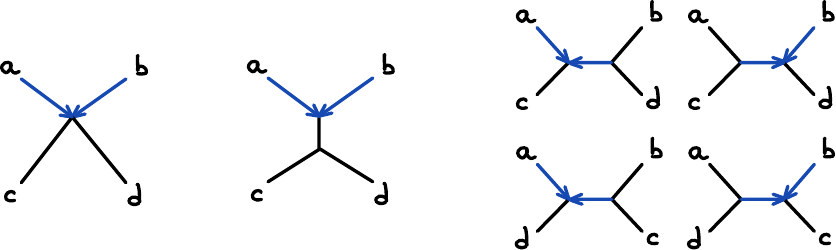}
  \caption{A polytomy below a hybrid node (left) can be resolved
  by adding a new edge of length $0$. The only resolution
  with identical up-down path lengths and average distances is
  by adding a new tree edge (middle).
  Contracting the horizontal hybrid edge in any network on the right
  to match the topology on the left would affect up-down paths and average distances.
  }
  \label{fig:polytomy}
\end{figure}

\begin{definition}[subnetwork, from \cite{2019banos}]
  Let $N^-$ be a semidirected network on $X$, and $Y \subset X$.  Then
  the induced network $N_Y^-$ on $Y$ is obtained by taking the union of all
  up-down paths in $N^-$ between pairs of tips in $Y$.
\end{definition}

  If $N^+$ is a rooted version of $N^-$, then it is possible to reroot $N_Y^-$ at
  $\lsa(Y)$ in $N^+$, which belongs in $N_Y^-$ as shown in
  \cite{2019banos}.
  $N_Y^-$ naturally inherits the metric from $N^-$:
  the distance between any pair of taxa $x, y \in Y$ is the same in $N_Y^-$
  and in $N^-$ because the up-down paths between $x, y$ are preserved,
  together with the edge lengths and hybridization parameters on these paths.

\subsection{Main results}

Since distances are defined on up-down paths, and up-down paths are identical
on a rooted network and its induced semidirected network, it follows
from Proposition~\ref{prop:up-down-path} that
average distances are independent of the root location on a rooted network,
so that the root is not identifiable from average distances:

\begin{proposition}
If rooted networks $N_1^+$ and $N_2^+$
induce the same semidirected network $N^-$, then pairwise distances on
$N_1^+$ and $N_2^+$ are identical.
\end{proposition}
What may be identifiable from average distances, at best,
is the semidirected network $N^-$ induced from $N^+$,
unless further assumptions are made.

Several papers have considered average distances on networks before,
with different assumptions on the networks.
\cite{willson13_recon_certain_phylog_networ_from} worked with binary
networks and assumed the knowledge of the root, that is, the root was
one of the labelled leaves and pairwise distance data was given
between the root and the other leaves.  \cite{francis-steel-2015} also worked with
binary networks and assumed that hybrid edges have length $0$,
along with other assumptions.

The remainder of the work focuses on the following
problem:  Given the average distances between tips,
what can we identify about the semidirected network:
what topological structures, and what continuous parameters?

\subsubsection{Non-identifiable features}

We first cover negative results, on features
that are not identifiable from average distances.
The simplest such feature is  the ``hybrid
zipper'' (Fig.~\ref{fig:hybrid-zipper}).
We will show that a network is not distinguishable
from its zipped-up version defined below.

\begin{definition}[zipped-up network]
\label{def:zip}
In a network, a hybrid node is \emph{zipped up} if all its parent
edges have length $0$. A network is \emph{zipped up} if all its hybrid
nodes are zipped up.  If a hybrid node $h$ is not zipped up in a
network $N$, the version of $N$ zipped up at $h$ is the network
obtained by modifying the edges adjacent to $h$ as follows (we refer
to this operation as a \emph{zipping-up}):
\begin{enumerate}
\item If $h$ has $k\geq 2$ children $c_1,\ldots,c_k$
  and is not zipped up,
  add a tree node $w$ and insert a tree edge
  $hw$ of length $0$ as unique child edge of $h$, then delete
  each edge $hc_i$ and replace by $wc_i$ with identical
  length (see Fig.~\ref{fig:hybrid-zipper}).
\item Set the length of the unique child edge $hw$ of $h$ to
  \begin{equation}\label{eq:zip-up}
  l_h = \ell(hw) +\sum_{u\text{\,parent of\,}h}\gamma(uh) \ell(uh)
  \end{equation}
  then set the length
  of all its parent hybrid edges to $0$.
\end{enumerate}
The \emph{zipped-up version} of a network $N$ is the network
obtained by zipping-up $N$ at all its unzipped hybrid nodes repeatedly.
\end{definition}

\noindent
In appendix~\ref{sec:zip-proofs}, we prove that the zipped-up version is unique.
Note that 
a hybrid node may need to be zipped multiple times before the
network is fully zipped up. In network $N^-$ from Fig.~\ref{fig:lsa_semidirected_nets}
(bottom right) for example, $h_2$ may need to be zipped-up twice
if it is considered before $h_1$.

\begin{figure}[h]
  \centering
  \includegraphics[scale=1.5]{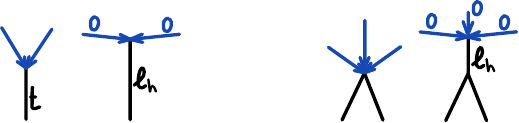}
  \caption{Zipping up: network transformation of branch lengths,
  setting hybrid edge lengths to $0$. The new length $l_h$ is
  given in \eqref{eq:zip-up}.
  In case of a polytomy
  below the hybrid node (right), a new tree edge of length $0$ needs to be
  added before zipping up.
  }
  \label{fig:hybrid-zipper}
\end{figure}

\begin{proposition}[hybrid zip-up]
  \label{prop:zipped-up}
  Let $N$ be a semidirected network, $h$ be a hybrid node in $N$ with
  parents $u_1,\dots,u_n$.
  If $h$ has more than one child, then step 1 in Definition~\ref{def:zip}
  does not affect average distances.
  If $h$ has one child $w$,
  then the average distances between the tips
  depend on
  $\ell(u_1h),\ldots,\ell(u_nh), \ell(hw)$ only through
  $l_h$ given by \eqref{eq:zip-up}.
  Therefore, zipping up $N$ at $h$ does not change average distances.
\end{proposition}

The proof is in
section~\ref{sec:zipped-up}.
This unidentifiability problem was mentioned in
\cite{Pardi_2015}, where it is referred to as ``unzipping'', as well
as in \cite{willson13_recon_certain_phylog_networ_from}.  It is
important to note that because we use average distances instead of
displayed trees in \cite{Pardi_2015}, we have an extra degree of
freedom: As in unzipping, we can subtract an equal amount $\epsilon$
in lengths from both $uh$ and $vh$, and add $\epsilon$ to $hw$.  This
leaves the average distances unchanged.  What is new with average
distances is that we can also ``slide" the hybrid node along the v-structure, that is:
subtract $(1-\gamma)\epsilon$ from $uh$ and add $\gamma \epsilon$ to $vh$.
This has no impact on the average distances either.

Because of the extra degree of freedom, instead of ``fully unzipping''
each reticulation as in \cite{Pardi_2015} and working with networks
where outgoing edges from hybrid nodes have length $0$, we shall
restrict our attention to \emph{zipped-up} 
networks, which are networks where all the hybrid edges have length
$0$ (defined rigorously below).  This requirement is also present in
\cite{willson12_tree_averag_distan_certain_phylog} and
\cite{willson13_recon_certain_phylog_networ_from},
but the non-identifiability underlying this requirement was not clarified.
The requirement was motivated by the fact that hybridizing populations
must be contemporary with each other.
However, hybrid edges of positive length appear naturally when hybridizations
involve ``ghost" populations that went extinct or with no sampled descendants,
or when two populations fuse, such as if their habitat becomes less fragmented
\cite{degnan2018}.

\begin{proposition}[shrinking blobs of degree 2 or 3]
  \label{cor:deg-2-3-blob-unid}
  Let $B$ be a blob of degree 2 or 3 in a semidirected network $N$.
  Then $N'$, the network obtained by shrinking $B$
  (i.e.\ identifying the nodes in $\partial B$ and deleting
  the other nodes in $B$) and modifying the lengths of the cut edges
  adjacent to $B$, induces the same pairwise distances. 
\end{proposition}

\begin{figure}[h]
  \centering
  \includegraphics[scale=1.5]{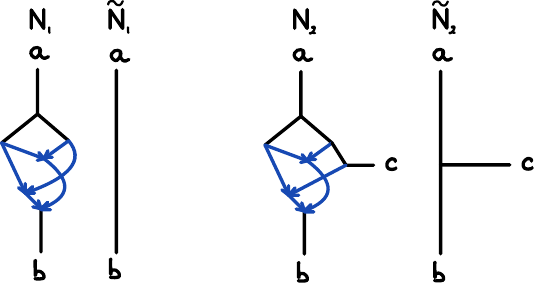}
  \caption{Any blob of degree $2$ (left) or $3$ (right) can be shrunk
  while preserving average distances.
  }
  \label{fig:shrink2d_3d_blobs}
\end{figure}

Section~\ref{sec:swap} presents a more general
  Lemma~\ref{lem:swap} to swap a subgraph with another within
  a semidirected network while keeping the average distances, from
  which Proposition~\ref{cor:deg-2-3-blob-unid} follows as a
  corollary.
Proposition~\ref{cor:deg-2-3-blob-unid} is used in many
proofs 
when considering subnetworks. If a blob 
reduces to a degree-2 or degree-3 blobs after subsampling leaves, then
it can be ignored, up to the lengths assigned to the edges replacing the
blob.
A consequence of
  Proposition~\ref{cor:deg-2-3-blob-unid} is that we require networks
  to not have degree-2 or degree-3 blobs in many of our results.
  Equivalently, this requirement
  can be interpreted as considering networks after these blobs
  have been shrunk and edge lengths modified appropriately.
Parallel edges may form a degree-2 blob. Even if they are
part of a larger blob, they can be swapped with a single tree edge:

\begin{proposition}[merging parallel edges]
  \label{cor:parallel}
  In a network $N$, let $h$ be a hybrid node such that
  all its parent edges $e_1,\dots,e_n$ are
  incident to the same nodes, $v$ and $h$.
  Consider the network $N'$ obtained by replacing $e_1,\dots,e_n$
  by a single tree edge $e = (vh)$ of length $\sum_i \gamma(e_i)\ell(e_i)$.
  Then $d_N = d_{N'}$.
\end{proposition}
This proposition, proved in section~\ref{sec:swap},
gives a rationale for a traditional assumption that
rooted phylogenetic networks do not have parallel edges \cite{steel16_phylogeny},
despite the biological realism of parallel edges.
First, parallel edges can arise from extinction or unsampled taxa:
hybridization between distant species would appear as a pair of parallel edges
if all the descendants of the two parental species are extinct
or not sampled.
Second, a species may split into 2 populations
and then merge back into a single population
due to evolving geographic barriers, such as glaciations.
Therefore, we allowed for parallel edges in our network definition.
Also, parallel edges may be identifiable from models and data
other than average distances \cite{degnan2018}.

Similarly, $3$-cycles are not identifiable:
any $3$-cycle (which may be part of a larger blob) can be shrunk to a single
node with the loss of one reticulation, without affecting average distances
(Fig.~\ref{fig:shrink3cycles}).

\begin{figure}[h]
  \centering
  \includegraphics[scale=1.5]{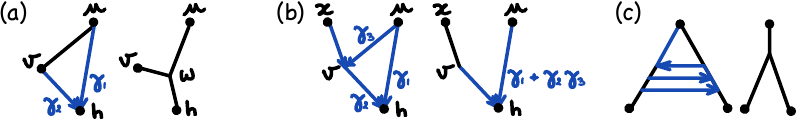}
  \caption{Any $3$-cycle can be swapped 
  while preserving average distances.
  (a) A 3-cycle with a single reticulation (left)
  can be swapped by a subgraph without reticulation (right).
  (b) A 3-cycle with 2 reticulations (left) can be swapped by a single
  reticulation (right). See Proposition~\ref{prop:3-cycle} for edge length adjustments.
  (c) A ladder of reticulations between sister lineages (left) can be eliminated (right)
  by repeated swaps.
  }
  \label{fig:shrink3cycles}
\end{figure}

\begin{proposition}
  \label{prop:3-cycle}
  Let $N$ be a semidirected network on $X$, $h$ a hybrid node in $N$
  with exactly two parents $u$ and $v$.
  Let $\gamma_1 = \gamma(uh)$ and $\gamma_2 = \gamma(vh) = 1 - \gamma_1$.
  \begin{enumerate}
  \item
  If $uv$ is a tree edge (Fig.~\ref{fig:shrink3cycles}a),
  then let $N'$ be the semidirected graph
  obtained by shrinking the 3-cycle $u,v,h$ as follows:
  remove edges $uh$, $vh$, and $uv$;
  add tree node $w$;
  add tree edges $uw$, $vw$, and $wh$ with lengths
  $$\ell(uw) = \gamma_2\ell(uv),\quad\ell(vw)=\gamma_1 \ell(uv), \quad
  \ell(wh)=\gamma_1 \ell(uh) + \gamma_2 \ell(vh)\,.$$
  \item
  If $uv$ is a hybrid edge (Fig.~\ref{fig:shrink3cycles}b) and if
  $v$ has exactly two parents $u$ and $x$ that are not
  adjacent,
  then let $\gamma_3 = \gamma(uv)$ and let $N'$ be the semidirected graph obtained by
  shrinking the 3-cycle as follows:
  remove $uv$; make the other parent edge of $v$ a tree edge;
  and set
  \begin{align*}
  \gamma(uh) &= \gamma_1 + \gamma_2\gamma_3\;;\quad
  \gamma(vh) = \gamma_2(1-\gamma_3)\\
  \ell(uh) &= \big(\gamma_2\gamma_3(\ell(uv) + \ell(vh)) + \gamma_1\ell(uh)\big) /
      (\gamma_2\gamma_3 + \gamma_1) \,.
  \end{align*}
  \end{enumerate}
  Then $N'$ is a semidirected network on $X$ with one fewer reticulation than $N$, and
  $N'$ induces the same average distances as $N$ on $X$.  We may also
  suppress the degree-$2$ nodes in $N'$, which does not affect
  pairwise distances.
\end{proposition}

\noindent
Proposition~\ref{prop:3-cycle} also follows from the swap lemma
and is proved in section~\ref{sec:swap}.
Note that in case 2, if $u$ and $x$ are adjacent, then we may
first shrink the 3-cycle $xvu$ before proceeding and shrinking $uvh$,
possibly recursively (Fig.~\ref{fig:shrink3cycles}c).
The lack of identifiability of 3-cycles and blobs of degree 2 or 3
explains the special cases found by \cite{francis-steel-2015} when
characterizing networks whose average distances fit on a tree.
For some classes, these networks must be trees
except for some local non-tree-like structures.
Namely, the class of ``primitive 1-hybridization" networks was defined
to allow for a short cycle near the root.
When the root is suppressed, this cycle becomes a 3-cycle.
Also, distances from ``HGT networks"
may fit a tree
despite a series of gene exchange
between sister species (Fig.~\ref{fig:shrink3cycles}c),
which form a degree-3 blob.
Our general characterization explains why these local
structures are invisible from average distances,
found by \cite{francis-steel-2015}.


\medskip
The hybrid zippers and $3$-cycles are not the end of identifiability
problems.  
Here we give an example of a level-2 network that is not identifiable
(with generic parameters), showing that in general, it is not possible
to identify the topology of a network with average distances, even
when requiring no degree-2 or 3 blobs and zipped-up reticulations.  As
a result, with average distances, we can only aim to identify networks
given restrictions, or identify only certain features of networks.

\begin{theorem}\label{thm:nonidentifiable}
  Let $k \geq 2$.  Consider the space $\cal S$ of
  zipped-up binary
  semidirected networks of level at most $k$ on $n\geq 4$
  taxa, with no $2$- or $3$-cycles.
  Networks in $\cal S$ are not generically identifiable from
  average pairwise distances, in the sense that
  there exists network topologies $N_1\neq N_2$ in $\cal S$
  and sets of parameters $\Omega_1$ and $\Omega_2$ with
  positive Lesbegue measure satisfying the following:
  for any $(\ell_1,\gamma_1)\in\Omega_1$,
  there exists $(\ell_2,\gamma_2)\in\Omega_2$ such that
  the average distances defined by $(N_1,\ell_1,\gamma_1)$ and
  $(N_2,\ell_2,\gamma_2)$ are identical.
\end{theorem}

The proof is presented in section~\ref{sec:4-sunlet}. 
In short, the main idea is to find examples on $4$ taxa, and then embed
these examples in larger networks for any $n$.
Fig.~\ref{fig:nonidentifiable}
provides examples of topologies that can serve the role of
$N_1$ and $N_2$ in Theorem~\ref{thm:nonidentifiable} for $n=4$.
The network on the left is of level 1, showing that level-2 networks
are not distinguishable from level-1 networks in general.
Also, the network on the right is tree-child, implying that
Theorem~\ref{thm:nonidentifiable} also holds for the smaller class
of tree-child networks of level at most $k$ (zipped and without any 2- or 3-cycles),
thus providing a stronger statement.

\begin{figure}[h]
  \centering
  \includegraphics[scale=1.5]{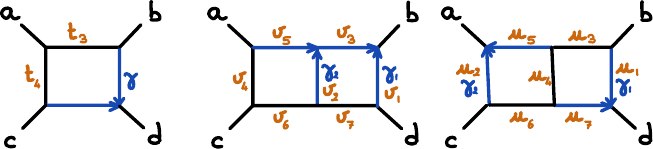}
  \caption{Semidirected networks with identical average distances on a
  parameter set of positive Lebesgue measure. The network on the right is tree-child.
  }
  \label{fig:nonidentifiable}
\end{figure}

\subsubsection{Identifiable features}

Now we move on to positive results, i.e.\ what we can identify
of a network from average distances, subject to certain constraints.

\begin{theorem}[identifying the tree of blobs]
  \label{ebt-refine-identify-thm}
  For a semidirected network $N$ with no degree-2 blob and no internal
  cut edge of length $0$, a refinement of the tree of blobs
  can be
  constructed from pairwise average distances, and which we call the
  \emph{distance split tree}.
\end{theorem}

\noindent
Theorem~\ref{ebt-refine-identify-thm} is proved in section~\ref{sec:ident-blob-tree}.
The distance split tree is defined rigorously in section~\ref{sec:ident-blob-tree},
Definition~\ref{def:dst}. Its construction is based on average distances alone.

Next, we provide examples
showing that the tree of blobs cannot be reconstructed without further assumptions,
and that the restriction of reconstructing a \emph{refinement} is necessary.
A refinement of tree $T$ is a tree $T'$ such that we can obtain
$T$ by contracting edges of $T'$.

\begin{example}\label{ex:blobtree-level2}
Consider the network $N$ in Fig.~\ref{fig:blobtree-level2} (left).
Let $t_i$ and $\gamma_i$ denote edge $i$'s length and hybridization parameter.
$N$ is a binary network of level 2.
Its tree of blobs is a star (Fig.~\ref{fig:blobtree-level2} center).
Its average distances are
equal to those obtained from a tree (Fig.~\ref{fig:blobtree-level2}
right) for specific parameter values,
namely when $\gamma_3 = t_2/(t_1+t_2)$ and $\gamma_7$ is small enough,
and the same topology is obtained by the distance split tree from
Theorem~\ref{ebt-refine-identify-thm}.  In this case, the distance
split tree is a strict refinement of the tree of blobs
$\bt(N)$, but is an exact and parsimonious explanation of the
distances.  We also note that under generic
parameters, 
the distance split tree is equal to the star tree of blobs
(see section~\ref{sec:ident-blob-tree} for the proofs).
\end{example}

\begin{figure}
  \centering
  \includegraphics[scale=1.5]{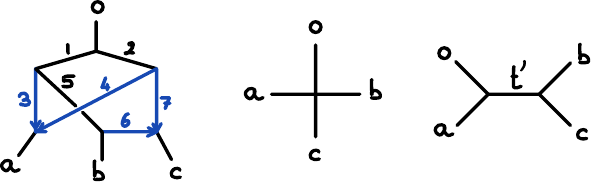}
  \caption{Example of a binary network $N$ of level 2 (left) whose tree of blobs is a star (middle).
   The distance split tree $T$ is also a star for
   generic parameters. But if
  $\gamma_3 = t_2/(t_1+t_2)$ and $\gamma_7$ is small,
  then $T$ is a strict refinement of the tree of blobs (right).
  In this case, $N$ and $T$ have identical average distances.}
  \label{fig:blobtree-level2}
\end{figure}

\begin{example}
Consider the networks $N_1$ and $N_2$ in Fig.~\ref{fig:blobtree-notbinary}.
$N_1$ is not binary. It has one blob, made of two biconnected components, and its tree of blobs is a star.
$N_2$ is a binary resolution of $N_1$ with two blobs: one for each block of $N_1$.
Since the extra cut-edge in $N_2$ has length 0, $N_1$ and $N_2$ have the same average distances
and the same distance split tree $T$
(Fig.~\ref{fig:blobtree-notbinary} right).
$T$ is a strict refinement of $N_1$'s tree of blobs, but it is the tree of blobs of $N_2$:
it recovers the separate blocks with the extra cut edge.
In this case again, the distance split tree represents a true feature of the network.
\end{example}
These examples and our results on non-binary level-1 networks (below)
lead us to state the following conjecture.

\begin{conjecture}[distance split tree as tree of blobs of equivalent network]
Let $N$ be a metric semidirected network on taxon set $X$,
$d_N$ its average distances on $X$, and
let $T$ be the distance split tree reconstructed from $d_N$.
Then there exists a semidirected $N'$ of level equal or less than that of $N$
with $d_{N'}=d_N$ and such that
$T$ is the tree of blobs of $N'$.
\end{conjecture}

\begin{figure}
  \centering
  \includegraphics[scale=1.5]{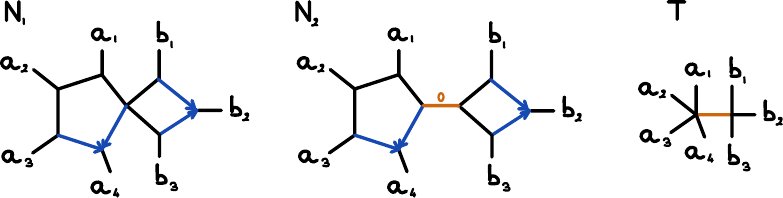}
  \caption{Left: example of a non-binary level-2 network $N_1$.
  Middle: $N_2$ is one possible resolution of $N_1$.
  $N_2$  is of level 1, with an extra cut edge of length $0$
  (in orange). Right: tree of blobs $T$ of $N_2$.
  If, for example, $\ell(e)=1$ for all edges in $N_1$ and $\gamma(e)=0.5$
  for all hybrid edges, then the
  distance split tree for $N_1$ is $T$, which is a
  strict refinement of $N_1$'s tree of blobs (a star).
  Note that the other resolutions of $N_1$ have a star as their tree of blobs.
  }
  \label{fig:blobtree-notbinary}
\end{figure}

In Theorem~\ref{thm:level-1-blockcut-tree} below (proved in section~\ref{sec:ident-blob-tree}),
we add assumptions to identify the tree of blobs exactly.
When we limit the network to be of level 1,
we characterize the distance split tree exactly:
it is the tree of blobs refined by extra edges to partially resolve polytomies
adjacent to blobs.

\begin{theorem}
\label{thm:level-1-blockcut-tree}
  Let $N$ be a level-$1$ network with internal tree edges of positive
  length and with no degree-$2$ blob.
  Then the distance split tree of $N$ is the tree of blobs of $N^R$,
  where $N^R$ is the network obtained as follows.
  For each non-trivial blob $B$ in $N$ and each node $u$ in $B$,
  let $\{a_1u,\ldots,a_ku\}$ be the set of edges of $u$ that
  are not in $B$.
  If $k\geq 2$, then refine the polytomy at $u$ by
  creating a new node $u'$, adding tree edge $uu'$ of length $0$;
  and disconnecting each $a_i$ from $u$ and connecting it to $u'$.
  That is, for each $i = 1, \dots, k$, remove $a_iu$ and create
  tree edge $a_iu'$.
\end{theorem}

\begin{remark}
  The assumption that internal tree edges have positive length is a
  weak requirement, because a tree edge of length $0$ can be
  contracted to create a polytomy.
\end{remark}

In the proof, we show that $N^R$ is indeed a valid semidirected
network, with identical average distances as $N$.
The distance split tree from $N^R$ is in fact the
block-cut tree of $N$
(see ``block-cutpoint trees'' \cite[][p.36]{harary1971graph}),
after suppressing its degree-2 nodes. For
  example, the network in Fig.~\ref{fig:blobtree-example} (left) has a
  distance split tree with $2$ extra edges (in orange, bottom right)
  compared to its tree of blobs (top right).
If we further assume that the network is binary, then $N=N^R$ and the
distance split tree equals the tree of blobs:

\begin{corollary}
  \label{thm:level-1-blob-tree}
  For a binary level-1 semidirected network with
  internal tree edges of positive length and
  no degree-2 blob, the
  tree of blobs can be constructed from average distances.
\end{corollary}

\medskip
In a binary level-$1$ network, each blob is a cycle and we can isolate
a blob by sampling an appropriate subset of tips. On this subset, the
induced subnetwork, after removing degree-$2$ blobs, is a ``sunlet''
\cite{gross18_distin_phylog_networ}, that is, a semidirected network
with a single cycle and a single pendant edge attached to each node in
the cycle.
In section~\ref{sec:k-sunlet}, we tackle the identifiability of sunlets.
Some of these results are special cases of those in
\cite{willson13_recon_certain_phylog_networ_from}
(with a simpler proof strategy).  One exception is
the case of $4$-sunlets, which is excluded by the assumptions in
\cite{willson13_recon_certain_phylog_networ_from}: In general, the
$4$-sunlet and its metric $(\ell,\gamma)$ are not identifiable, but the
unrooted $4$-sunlet is a structure that can be detected in the tree of blobs.

In section~\ref{sec:k-sunlet}, we can characterize all the $4$-sunlets
(and their parameters)
that give rise to a given average distance matrix.
Ideally one would choose a ``canonical'' $4$-sunlet as the
representative of all these distance-equivalent sunlets.  However, we
did not find a single sensible choice for such canonical $4$-sunlet.
Consequently we opt to use a separate split-network type of
representation for these $4$-sunlets.

Section~\ref{sec:level-1} introduces \emph{mixed networks},
in which some parts are semidirected and 4-cycles are undirected.
In short, a mixed network encodes the reticulation node and edges in $k$-cycles for $k\geq 5$,
and the unrooted topology in 4-cycles, without identification of the exact
placement of the hybrid node in a 4-cycle.
In section~\ref{sec:level-1}, we define this representation rigorously
and combine results from sections~\ref{sec:ident-blob-tree} and \ref{sec:k-sunlet}
to prove the following main result. 

\begin{theorem}[identifiability of mixed network representation for level-$1$ networks]
\label{thm:level1-4up}
  Let $N_1$ and $N_2$ be level-$1$ semidirected networks on $X$ with no $2$ or
  $3$-cycles or degree-$2$ nodes, and with internal tree edges of positive lengths.
  Let $N_i^*$ be the mixed network representation of $N_i$ after
  zipping up its reticulations for $i = 1, 2$.
  Then $d_{N_1} = d_{N_2}$ implies that $N_1^* = N_2^*$.
\end{theorem}

\noindent
Note that for a network $N$ satisfying the conditions above
  and its mixed network representation $N^*$, $N$ and $N^*$ have the
same unrooted topology, except that polytomies adjacent to $4$-cycles
in $N$ may be partially resolved in $N^*$.
If we consider the space of level-1 networks with no 4-cycles, we
obtain the following result as a special case.

\begin{corollary}[identifiability of zipped-up version of level-1
  networks]
  \label{thm:level1-5up}
    Let $N_1$ and $N_2$ be zipped-up level-$1$ semidirected networks on
    $X$ with tree edges of positive lengths, without any $2$, $3$ or
    $4$-cycles, and without degree-$2$ nodes.
    If $d_{N_1} = d_{N_2}$ then $N_1 = N_2$.
\end{corollary}

\subsection{Biological relevance}

In practice, average distances between pairs of taxa need to be estimated
from data. \cite{allman2022-identifiability-logdet-net} studied the identifiability
of the network topology using the log-det distance,  for level-1 networks and
under a coalescent model.
Future work could study identifiability of the network and its parameters
under various models and for various methods to estimate evolutionary distances,
such as the average coalescence time between pairs of taxa
\cite{liu2009_STAR_STEAC},
average internode distance \cite{liu2011_NJst},
or the $f_2$ statistic when many genomes are available from each species
\cite{peter2016}.

The most frequent reticulations are expected
between incipient species, or sister species that just split from each other
and have yet to achieve reproductive isolation. Our work shows
that these most frequent reticulations are not identifiable from
average distances.
Only the less frequent events between more distant species can be
detected using average distances.

Our work also shows a strong effect of taxon sampling, as observed with
real data \cite{conover2019_malvaceae,karimi2020_baobab}.
Dense taxon sampling is critical to avoid blobs of degree 2 or 3.
For example, if a hybridization forms a cycle of degree 5 in a full level-1 network,
then it is necessary to sample at least one taxon from across each of the 5 cut edges
adjacent to the cycle, for the reticulation to be identifiable from average distances.
Conversely, it may be useful to reduce taxon sampling strategically. Reducing the degree
of some blobs to be 3 or less in the subnetwork could be a strategy to obtain a more
resolved tree of blobs on the reduced taxon set.
When the true network is of level greater than 1, different taxon subsamples may lead
to different trees of blobs, and to the detection of different reticulation events by methods
that assume a level-1 network. While this sensitivity to taxon sampling may be
disconcerting, subsampling can decrease the level
and bring strength to methods that require low-level networks like
SNaQ \cite{Sol_s_Lemus_2016_infer} and NANUQ \cite{Allman_2019}.

Pairwise distances are not unique in causing some features to lack identifiability.
From quartet concordance factors for example, some 3-cycles cannot be identified,
and the hybrid node position is not always identifiable in a 4-cycle
\cite{Sol_s_Lemus_2016,2019banos,solislemus2020_identifiability}.
Software for network inference should provide information on the class of equivalent networks
with identical optimal likelihood, e.g. list the multiple ways to place the reticulation in a 4-cycle.
Bayesian approaches could report on sets of networks that cannot be distinguished
from the data, and whose relative posterior probabilities are solely influenced by the prior.
Interactive visualization tools could facilitate the exploration of networks with equivalent scores,
so practitioners could avoid interpretations that hinge on a strict subset of these networks.
If software is misleadingly presenting a single network as being optimal without presenting
the whole class of networks with equivalent fit, then undue confidence could be
placed on some interpretation.

\section{Proofs related to non-identifiable features}
\label{sec:nonid-struct}

\subsection{Hybrid zip-up}
\label{sec:zipped-up}

To prove Proposition~\ref{prop:zipped-up}, we first need
the following definition and proposition.

\begin{definition}[displayed tree]
\label{def:displayedtree}
  Let $N$ be a directed or semidirected network. For
  hybrid node $h$, let $E_H(h)$ be its parent hybrid edges.
  Let $T$ be the graph obtained by keeping one hybrid edge
  $e\in E_H(h)$ and deleting the remaining edges in $E_H(h)$,
  for each hybrid node $h$ in $N$.
  Then $T$ is a tree, and is called a \emph{displayed tree}.
  The distribution on displayed trees generated by $N$ is
  the distribution obtained by keeping $e\in E_H(h)$
  with probability $\gamma(e)$, independently across $h$.
\end{definition}
\noindent
Note that $T$ is a tree because it is a DAG (considering $N$ as rooted),
all the nodes are still reachable from the root, and all nodes have in-degree at
most $1$.

\begin{proposition}
  \label{prop:expectation-distance}
  Let $N$ be a (directed or semidirected) network.
  For two nodes $u$ and $v$ and tree $T$,
  let $q_{uv}(T)$ be
  the unique path between $u$ and $v$ in $T$.
  Then, for a given up-down
  path $p$ between $u$ and $v$ in $N$,
  \[ \PP(q_{uv}(T) = p) = \gamma(p)\]
  where the probability is taken over a random tree $T$ displayed in $N$.
  As a result,
  \[d(u, v) = \EE \ell(q_{uv}(T))\] where $d(u,v)$ is the distance
  between $u, v$, and $\ell(q)$ is the length of path $q$.
\end{proposition}

\begin{proof}
  Since rooting a semidirected network does not change the process of
  generating displayed trees, nor up-down paths in the network, it
  suffices to consider the case when $N$ is a directed network.
  Let $p$ be an up-down path in $N$.
  Let $E = \{e_1, \dots, e_n\}$ be the set of hybrid edges in $p$,
  and let $H = \{h_1, \dots, h_n\}$ where $h_i$ is the child of $e_i$.
  All the $h_i$'s are distinct hybrid nodes because the up-down path
  $p$ may not go through partner hybrid edges (no v-structure).
  It suffices to show that $p = q_{uv}(T)$ if and only if for
  each $h_i \in H$, $e_i$ is kept.

  The ``only if'' part is evident since $e_i$ has to be in $T$
  for $p$ to be in $T$.  Now consider a displayed tree $S$ where for
  all $h_i \in H$, $e_i$ is kept.  Because all of the edges of $p$ are
  in $S$, $p$ is a path between $u$ and $v$ in $S$; and since $S$ is
  a tree, $p$ is the unique path between $u$ and $v$ in $S$.
\end{proof}

\begin{proof}[Proof of Proposition~\ref{prop:zipped-up} (hybrid zip-up)]
  If $h$ has more than one child, then step 1 in Definition~\ref{def:zip}
  does not modify the set of up-down paths, other than inserting
  $hw$ to the paths containing $u_ih,hc_j$ for some parent $u_i$ and
  child $c_j$ of $h$. Since the length of $hw$ is set to 0, step 1 does not
  change the length of up-down paths, nor average distances.
  If $h$ has a single child $w$,
  we first assume that $h$ is the only hybrid node in
  $N$.
  Let $a_i = \ell(u_ih)$ and $\gamma_i = \gamma(u_ih)$ for $i = 1, \dots, n$
  and let $t = \ell(hw)$.
  Let $x, y$ be two tips of $N$.  There are two cases:

  \begin{enumerate}
  \item If some up-down path $p$ between $x, y$ does not contain $h$, then
    $p$ must be the unique up-down path between $x$ and $y$.  This is
    because $p$ is a path on the $n$ displayed trees of $N$.
    Consequently the distance between $x$ and $y$ does not depend
    on $a_1,\ldots,a_n$ or $t$, and the statement is vacuously true.
  \item If all up-down paths between $x, y$ contain $h$, then
    there are exactly $n$ up-down paths $q_i$, with $q_i$ containing
    the edge $u_ih$.  In this case the average distance is
    \[d(x,y) = \sum_{i=1}^n\gamma_i(a_i + l_i + t) = l_h +
      \sum_{i=1}^n\gamma_i l_i ,\]
    where $l_i = \sum_{e \in q_i, e \notin \{u_1h,\ldots,u_nh, hw\}}\ell(e)$
    does not depend on $(a_1, \dots, a_n, t)$.
  \end{enumerate}
  For the general case, we use Proposition~\ref{prop:expectation-distance}.
  Let $T$ be a random displayed tree in $N$ and $D_{xy} = D_{xy}(T)$
  be the distance between $x$ and $y$ in $T$ (a random quantity).
  Let $R$ be the set of hybrid edges kept in $T$ at all hybrid nodes except for $h$.
  By conditioning on $R$, we reduce the problem to a network with a single
  reticulation: $\EE[D_{xy}  \mid  R]$ is the average distance on a (random) network
  with a single reticulation at $h$. By the above argument, $\EE[D_{xy}  \mid  R]$ only
  depend on $(a_1, \dots, a_n, t)$ through $l_h$.  Taking expectation
  again gives the result.
\end{proof}

\subsection{Swap lemma and related results}
\label{sec:swap}

In this section we present a general swap lemma and apply it to prove the
non-identifiability of specific features.
We first
introduce some necessary definitions and notations.

Let $A$ be a subgraph of a semidirected network $N$ on $X$.
$A$ is \emph{hybrid closed} if for any hybrid
edge in $A$, all of its partner edges are in $A$.  We use $\partial A$
to denote the \emph{boundary} of $A$ in $N$, defined as
the set of nodes in $A$ that are either leaves, or are
incident to an edge not in $A$.
For two nodes $a, b \in \partial A$, we define
\[
  \gamma_A(a, b) = \sum_{p: \; a \leftrightarrow b \text{ in } A}\gamma(p),
\]
where $p$ ranges over the set of up-down paths from $a$ to $b$ that
lie entirely in $A$.  Then we define the \emph{conditional distance in $A$}
as
\[ d(a, b \mid A) =
  \begin{cases}\displaystyle
    \frac{1}{\gamma_A(a, b)}\;\sum_{p: \; a \leftrightarrow b \text{ in } A} \gamma(p) \ell(p) &\text{if }\gamma_A(a, b) >
    0 \\
    0 & \text{if }\gamma_A(a, b) = 0
    \end{cases}
 \]
where in the sum $p$ again ranges over up-down paths
between $a, b$ that lie in $A$.
The following lemma says that average distances are unchanged if we swap
$A$ with another subgraph of identical boundary, provided that
$\gamma_A$ and $d(.,. \mid A)$ are preserved on $\partial A$.

\begin{lemma}[subgraph swap]
  \label{lem:swap}
  Let $N_1$ and $N_2$ be metric semidirected networks on the same leaf set $X$,
  with node sets $V(N_i) = V_{A_i}\sqcup V_B$ and edge set
  $E(N_i) = E_{A_i} \sqcup E_B$ for $i=1,2$, such that
  $A_1$ and $A_2$ are hybrid-closed subgraphs with identical boundary 
  in $N_1$ and $N_2$ respectively: $\partial A_1 = \partial A_2$, denoted as $\partial A$.
  If
  $\gamma_{A_1}(a, b) = \gamma_{A_2}(a, b)$ and
  $d_{N_1}(a, b \mid A_1) = d_{N_2}(a, b \mid A_2)$
  for every $a, b \in \partial A$, then
  $d_{N_1} = d_{N_2}$ on $X$.
\end{lemma}

\noindent

\medskip\noindent

\begin{proof}
  Given an up-down path $p$ and
  two nodes $a, b$ on the path, we write $a \udp{p} b$ for the segment
  of $p$ between $a, b$, which is an up-down path as well.

  We first prove the lemma when there are no hybrid edges in
  $E_B$. 
  For now we consider the distances in $N_1$.  To simplify notations,
  we shall write $N = N_1$ and $A = A_1$.  Let $x, y \in X$ be two
  tips, and $p$ an up-down path between them.  Then $p$ can be subdivided
  into segments that consists of consecutive edges in $E_B$, and segments
  of consecutive edges in $E_A$, the set of edges in $A$.
  Traversing $p$ from $x$ to $y$, let the $j$th
  segment in $E_A$ be $a_j \udp{p} a_j'$, where $a_j,a_j'\in\partial A$.
  Note $a_1 = x$ and $a_k' = y$ are possible for some $k$,
  if $x$ or $y$ is in $A$. Given a
  $2k$-tuple $\bm{a} = (a_1, a_1', \dots, a_k, a_k') \in \partial A^{2k}$, let
  $\mathcal{P}_k(\bm{a})$ denote the set of
  up-down paths $p$ between $x$ and $y$
  with exactly $k$ segments in $A$, and such that segment $j$
  enters and exits $A$ at $a_j,a_j'$: $a_j \udp{p} a_j'$.
  The set of up-down paths between $x$ and $y$ can then be written as
  \[\bigsqcup_{k\geq 0}\bigsqcup_{\bm{a} \in \partial A^{2k}} \mathcal{P}_k(\bm{a}).\]
  Consequently, we have
  \begin{equation}\label{eqn:d-k-decomp}
    d(x, y) = \sum_{k \geq 0} \sum_{\bm{a} \in \partial A^{2k}} \sum_{p \in
      \mathcal{P}_k(\bm{a})} \gamma(p) \ell(p)
  \end{equation}
  For a given $\bm{a} \in\partial A^{2k}$ and $p \in \mathcal{P}_k(\bm{a})$, consider
  the segments in $E_B$, which are the segments
  $a_j' \udp{p} a_{j+1}, j = 1, \dots, k-1$ and possibly also
  $x \udp{p} a_1$ or $a_k' \udp{p} y$ when $x$ or $y$ is in $V_B$.
  These segments are uniquely determined by $k$ and $\bm{a}$
  because we assumed no hybrid edges in $E_B$. There are no undirected cycles in
  the subgraph formed by nodes $V_B \cup \partial A$ and edges $E_B$, so
  for a given $a_j',a_{j+1}\in\partial A$, there is either no path or a single (tree) path
  $a_j' \leftrightarrow a_{j+1}$ in $E_B$.
  Let $E_B(k, \bm{a})$ be the set of edges
  $e \in E_B$ such that $e \in p$ for any (and every) $p \in \mathcal{P}_k(\bm{a})$.

  The segments of an up-down paths must be up-down paths.
  Conversely, the concatenation of contiguous up-down paths
  alternating from $E_A$ and $E_B$
  is still an up-down path, because $E_B$ contains tree edges only.
  Therefore, if $E_B(k, \bm{a})$ is non-empty
  or not needed ($k=1$ and $x,y\in\partial A$),
  then $\mathcal{P}_k(\bm{a})$ is non-empty and there is a bijection between
  $\mathcal{P}_k(\bm{a})$ and $\prod_{j=1}^k \mathcal{P}_1((a_j, a_j'))$:
  each $p \in \mathcal{P}_k(a)$ is mapped to the
  segments in $E_A$, with interleaving segments
  $E_B(k, \bm{a})$. 
  Consequently, if $\mathcal{P}_k(\bm{a})$ is not empty then we have
  \begin{eqnarray*}
  \sum_{p \in \mathcal{P}_k(\bm{a})} \gamma(p)
  &=&
  \sum_{p \in \mathcal{P}_k(\bm{a})}\prod_{i=1}^k\gamma(a_i \udp{p} a_i')
  = \sum_{p_1 \in \mathcal{P}(a_1, a_1')}\dots\sum_{p_k \in \mathcal{P}(a_k, a_k')}
    \prod_{i=1}^k\gamma(p_i)\\
  &=& \prod_{i=1}^k\gamma_A(a_i, a_i') \;.
  \end{eqnarray*}
  For the first summation in \eqref{eqn:d-k-decomp}, we may then
  write
  \begin{eqnarray}
    \sum_{p \in\mathcal{P}_k(\bm{a})} \gamma(p)\ell(p) &=& \nonumber\!\!
    \sum_{p \in\mathcal{P}_k(\bm{a})}\left(\prod_{i=1}^k\gamma(a_i \udp{p} a_i')\right)
      \left( \sum_{e \in E_B(k, \bm{a})}\ell(e) + \sum_{i=1}^k\ell(a_i \udp{p} a_i') \right) \\
    &=& \nonumber
    \prod_{i=1}^k\gamma_A(a_i, a_i')\left( \sum_{e \in E_B(k,\bm{a})}\ell(e)
    \right) + \\
    && \nonumber
    \sum_{p \in \mathcal{P}_k(\bm{a})}\sum_{i=1}^k
      \left(\prod_{j \neq i}\gamma(a_j \udp{p} a_j')\right)
       \gamma(a_i \udp{p} a_i') \ell(a_i \udp{p} a_i') \\
    &=& \nonumber
    \prod_{i=1}^k\gamma_A(a_i, a_i')\left(
      \sum_{e \in E_B(k,\bm{a})}\ell(e) \right) +
      \sum_{i=1}^k d(a_i, a_i'  \mid  A) \prod_{i=1}^k\gamma_A(a_i,a_i') \\
    &=& \prod_{i=1}^k\gamma_A(a_i, a_i') \left(
      \sum_{e \in E_B(k,\bm{a})}\ell(e) +
      \sum_{i=1}^kd(a_i, a_i'  \mid  A) \right) \;.\label{eqn:cond-dist}
  \end{eqnarray}
  From \eqref{eqn:cond-dist}, it follows that as long as
  $d(\cdot, \cdot \mid  A)$ and $\gamma_A$ remain the same on $\partial A$,
  then $d(x, y)$ does not change.

  The general case follows by first conditioning on a choice of hybrid
  edges in $E_B$ (which must be hybrid closed because $A_1$ is) 
  and then using Proposition~\ref{prop:expectation-distance}.
\end{proof}

\noindent

\begin{proof}[Proof of Proposition~\ref{cor:deg-2-3-blob-unid}]
The non-identifiability of blobs of degree 2 or 3 follows as an immediate
corollary of Lemma~\ref{lem:swap}:
If $A$ is a blob of degree 2 or 3, then $\partial A$ has 2
or 3 nodes, and $\gamma_A \equiv 1$
(Fig.~\ref{fig:shrink2d_3d_blobs}).  Since any metric on a set of 2 or
3 elements can be represented by a tree metric, we may replace $A$ by
one tree edge or by three tree edges, to match $d(.,. \mid A)$
exactly.
Specifically, using Fig.~\ref{fig:shrink2d_3d_blobs}, on the left the degree-2 blob
can be swapped by a single edge $(a,b)$ in ${\tilde N}_1$
of length set to $d_{N_1}(a,b)$.
On the right, a degree-3 blob can be swapped by a single tree node.
Edge lengths in ${\tilde N}_2$ are determined
by the average distances between $a,b,c$ in $N_2$. For example,
the length of the edge to $a$ is $(d_{N_2}(a,b) + d_{N_2}(a,c) -
d_{N_2}(b,c))/2 \geq 0$.
\end{proof}

\begin{proof}[Proof of Proposition~\ref{cor:parallel}]
  The subgraph $A_1$ induced by $\{v,h\}$ contains the parallel edges $e_1,\dots,e_n$
  exactly, has $\partial A_1=\{v,h\}$ and is hybrid closed because all parent edges of $h$
  are assumed to be parallel. $A_2$ is the subgraph on $\partial A_2=\{v,h\}$
  with a single tree edge $e=(vh)$. Trivially,
  $\gamma_{A_1}(v,h) = 1 = \gamma_{A_2}(v,h)$, and the length of $e$ was
  defined to ensure that $d_{N}(v,h \mid  A_1) = d_{N'}(v,h \mid A_2)$.
\end{proof}

\begin{proof}[Proof of Proposition~\ref{prop:3-cycle}]
  We consider here a subgraph that contains a triangle,
  and swap it with a simpler subgraph in which the triangle is shrunk.
  First note that $N'$ is a valid semidirected network, because the
  swapping operation can be made on a rooted network to obtain
  a valid rooted network (with the same root).
  In case 1 we apply Lemma~\ref{lem:swap} to swap
  the subgraph $A_1$ induced by $\{u,v,h\}$ on the left of Fig.~\ref{fig:shrink3cycles}(a),
  with $A_2$ induced by $\{u,v,h,w\}$ on the right.
  $A_1$ is hybrid closed in $N$ and $A_2$ is hybrid closed in $N'$;
  $\partial A = \partial A_1 = \partial A_2 = \{u,v,h\}$ and
  $\gamma_{A_1}\equiv \gamma_{A_2}\equiv 1$ on $\partial A$.
  The branch lengths in Proposition~\ref{prop:3-cycle} ensure that
  $d_{N}(.,. \mid A_1) \equiv d_{N'}(.,. \mid A_2)$ on $\partial A$.

  In case 2, let $x$ be the parent node
  of $v$ other than $u$. We apply Lemma~\ref{lem:swap} to swap
  $A_1$ with $A_2$ in Fig.~\ref{fig:shrink3cycles}(b), with $V(A_1)=V(A_2)=\{u,v,h,x\}$.
  $A_1$ and $A_2$ are hybrid closed in
  $N$ and in $N'$, both with boundary $\partial A  = \{x,u,h\}$.
  In $N$ and $N'$, we have $\gamma_A(x,u)=0$,
  $\gamma_A(x,h)=\gamma_2(1-\gamma_3)$ and
  $\gamma_A(u,h)=\gamma_1 + \gamma_2\gamma_3$.
  The branch lengths in Proposition~\ref{prop:3-cycle} ensure that
  $d_{N}(.,. \mid A_1) \equiv d_{N'}(.,. \mid A_2)$ on $\partial A$.
\end{proof}

\section{Identifying the tree of blobs}
\label{sec:ident-blob-tree}

In this section, we prove Theorems~\ref{ebt-refine-identify-thm} and \ref{thm:level-1-blockcut-tree}.
The key arguments are as follows.
First, edges in the tree of blobs $\bt(N)$ define the same splits of leaves as cut-edges in $N$.
Second, pairwise distances satisfy the ``4-point condition" for any set of four taxa that
spans one of these cut-edge splits.
These terms and statements are made rigorous below.
\com{
First, we justify that we can label the leaves of $\bt(N)$ by the
taxon labels in $X$, and also identify the edges in $T$ with
the cut edges of $N$.
}

\begin{proposition}
  \label{prop:net-blobtree-same-cutedges-tips}
  For a semidirected network $N$, there is a bijection between the
  edges of $\bt(N)$ and the cut edges of $N$, and a bijection between
  the leaves of $\bt(N)$ and the tips of $N$.
\end{proposition}

The proof is in appendix~\ref{sec:blob-tree-proofs}
since it is simply technical.
Recall that a \emph{split} $A \mid B$ of a set $X$ is a partition of $X$
into two disjoint nonempty subsets $A$ and $B$.  For a phylogenetic
$X$-tree $T$ and an edge $e$ of $T$, the split $\sigma(e)$ induced by
$e$ is the partition on $X$ induced from the two connected components
of $T$ when $e$ is removed.  We denote the set of edge-induced splits
of a phylogenetic $X$-tree by $\Sigma(T)$. Two splits $A \mid B$ and
$C \mid D$
of $X$ are \emph{compatible} if at least one of $A \cap C$,
$A \cap D$, $B \cap C$, and $B \cap D$ is empty.  By the
Splits-Equivalence theorem \cite{charles03_phylog}, all the splits in
$\Sigma(T)$ are compatible.  Furthermore, two sets of splits
$\Sigma_1$ and $\Sigma_2$ are \emph{pairwise compatible} if for all
$\sigma_1 \in \Sigma_1, \sigma_2 \in \Sigma_2$, $\sigma_1$ and
$\sigma_2$ are compatible.  A single split $\sigma$ and a set of
splits $\Sigma$ are pairwise compatible if $\{\sigma\}$ and $\Sigma$
are pairwise compatible.

If $N$ has no degree-$2$ blob, then its tree of blobs $T$ can be viewed as
a phylogenetic $X$-tree.  Different cut-edges in $N$, and therefore
different edges in $T$, correspond to different splits $A \mid B$ in $X$.

\begin{definition}[4-point condition]
  Given a metric $d$ on $X$, the tuple $(x, y, u, v)$ of leaves in $X$
  satisfies \emph{the 4-point condition} if
  \begin{equation}\label{4point}
    d(x, y) + d(u, v) \leq d(x, u) + d(y, v) = d(x, v) + d(y, u).
  \end{equation}
  Because \eqref{4point} is the same if we switch $x, y$ or $u, v$, we
  can define the above condition as the 4-point condition on the
  quartet $xy \mid uv$ (short for $\{x, y\} \mid \{u, v\}$).
We also say the 4-point condition is satisfied for $\{x, y, u, v\}$
if it holds for some permutation of $(x, y, u, v)$.
  We say that $xy \mid uv$ satisfies the 4-point condition
  \emph{strictly} if the inequality in~\eqref{4point} is strict.

  A split $A \mid B$ on $X$ is said to satisfy the
  4-point condition (strictly) if for any $x, y \in A$ and $u, v \in B$, the
  4-point condition on $xy \mid uv$ is satisfied (strictly).
\end{definition}

On a tree, the 4-point condition is
satisfied for any choice of four nodes.  In the example below,
the 4-point condition is not satisfied.

\begin{example}[4-point condition on a 4-cycle]
  \label{ex:4pc-4sunlet}
  Let $N$ be the leftmost network in Figure~\ref{fig:nonidentifiable}
  with $t_3 > 0$ and $t_4 > 0$.  A quick calculation shows that
  $d(a, d) + d(b, c) - d(a, c) - d(b, d) = 2 \gamma t_3 > 0$ and
  $d(a, d) + d(b, c) - d(a, b) - d(c, d) = 2 (1-\gamma) t_4 > 0$.
  Therefore the 4-point condition is not satisfied on the tips
  $\{a, b, c, d\}$.
\end{example}

\begin{lemma}
  \label{ebt-edge-distance-compatible-lemma}
  Let $N$ be a semidirected network and $T$ its tree of blobs.
  All splits $\sigma \in \Sigma(T)$ satisfy the 4-point condition for $d_N$,
  and $\sigma(e)$ satisfies the 4-point condition strictly if $\ell(e)>0$.
  Furthermore, if all internal cut-edges in $N$ have positive length,
  then any split $\sigma'$ on $X$
  that satisfies the 4-point condition is pairwise compatible with
  $\Sigma(T)$.
\end{lemma}

\begin{proof}
  As in the previous discussion, we identify edge $e$ in $T$ with the
  corresponding cut edge in $N$.

  Let $\sigma = \sigma(e) = A \mid B$.  Take $a, b \in A$, $u, v \in B$.
  Since $e$ is a cut edge in $N$, removing $e$ results in two
  connected components, such that $a, b$ are in the
  same component and $u, v$ are in the other.  Let $c, w$ be the
  vertices of edge $e$, with $c$ in the same connected component as
  $a, b$, and $w$ in the same one as $u, v$.

  Let $D(p,q)$ be the random up-path length between nodes $p$ and $q$,
  that is, the length of the up-down path between $p$ and $q$ induced by
  a randomly sampled displayed tree.
  Since all up-down paths from $a$ to $u$ must contain $e$,
  we have
  \[ D(a,u) = D(a,c) + \ell(e) + D(w,u). \] Taking expectations,
  \[ d(a,u) = d(a,c) + \ell(e) + d(w,u). \] Similar equations hold
  for the pairs $(a, v)$, $(b, u)$, and $(b, v)$,
  from which we get
  \begin{align*}
    d(a,b) + d(u,v) &\leq d(a,c) + d(b,c) + d(w,u) + d(w,v) \\
    &= d(a,u) + d(b,v) - 2\ell(e) = d(a,v) + d(b,u) - 2\ell(e)\,.
  \end{align*}
  Hence $A \mid B$ satisfies the 4-point condition, strictly if $\ell(e)>0$.

  To prove the second claim, assume that there exists a split $\sigma' = U \mid V$
  satisfying the 4-point condition, but that is not compatible with
  a split $\sigma = A \mid B$ induced by some edge $e$ in $T$.  Since
  $\sigma$ is nontrivial, $e$ is an internal edge and $\ell(e) > 0$.
  Then we can find $a, b, u, v$ such that
  $a, b \in A$, $u, v \in B$, and $a, u \in U$, $b, v \in V$.
  Consequently the 4-point condition holds both on $ab \mid uv$ and
  $au \mid bv$.  It then follows that the three sums
  $d(a, b) + d(u, v)$, $d(a, u) + d(b, v)$ and $d(a, v) + d(b, u)$
  are all equal.
  Then the 4-point condition on $ab \mid uv$ cannot be strict,
  implying $\ell(e)=0$: a contradiction.
\end{proof}

\begin{definition}[distance split tree]
\label{def:dst}
  Let $d$ be a metric on $X$.  Let $\Sigma(d)$ be the
  set of splits on $X$ that satisfy the 4-point condition, and
  $\Sigma'(d)$ the set of splits in $\Sigma(d)$ that are pairwise
  compatible with $\Sigma(d)$.  Note that by construction,
  $\Sigma'(d)$ is pairwise compatible.
  The \emph{distance split tree} is defined as the $X$-tree $\tau(d)$ that
  induces $\Sigma'(d)$.
\end{definition}

\noindent
By the splits-equivalence theorem, $\tau(d)$ exists and is unique.
Also, $\tau(d_T)=T$ if $T$ is a tree \cite{charles03_phylog}.

\begin{proof}[Proof of Theorem~\ref{ebt-refine-identify-thm}]
  Let $N$ be a semidirected network on $X$ satisfying the requirements
  in Theorem~\ref{ebt-refine-identify-thm} and $d=d_N$.
  For a tree $T$, let $\Sigma_T$ the set of splits induced by $T$.
  Using the notations in
  Definition~\ref{def:dst}, we have
  \begin{equation*}
    \Sigma_{\bt(N)} \subset \Sigma_{\tau(d)} = \Sigma'(d) \subset \Sigma(d).
  \end{equation*}
  Because $\Sigma_{\bt(N)} \subset \Sigma_{\tau(d)}$, $\tau(d)$ is a refinement of $\bt(N)$.
\end{proof}

\begin{proof}[Proof of Example~\ref{ex:blobtree-level2}]
For the network $N$ in Fig.~\ref{fig:blobtree-level2} (left),
we prove here that the distance split tree $\tau(d_N)$
is a star for generic parameters,
and is the tree $oa \mid bc$ when $\gamma_3 = t_2/(t_1+t_2)$ for
$\gamma_7$ small enough. It is easy to write the expressions
\begin{align*}
S_a &= d(o,a)+d(b,c) =
  S_0 + \gamma_6 u +
  \gamma_7(2(\gamma_3t_1+\gamma_4t_2) + v)\\
S_b &= d(o,b)+d(a,c) =
  S_0 + \gamma_6(u + 2(\gamma_4t_1+t_5)) +
  \gamma_7(2\gamma_3t_1 + v)\\
S_c &= d(o,c)+d(a,b) =
  S_0 + \gamma_6(u + 2(\gamma_4t_1+t_5)) +
  \gamma_7(2\gamma_4t_2 + v)
\end{align*}
where $S_0$ is the sum of the external edge lengths after zipping-up
the network,
$u=\gamma_3t_1+\gamma_4t_2$ and
$v=\gamma_3t_2+\gamma_4t_1  + t_5$.
Consequently,
\begin{align*}
S_b &= S_c &\Longleftrightarrow&&
\gamma_3&=t_2/(t_1+t_2)\\
S_a &= S_b &\Longleftrightarrow&&
\gamma_6&=\gamma_4t_2/(\gamma_4(t_1+t_2)+t_5)\\
S_a &= S_c &\Longleftrightarrow&&
\gamma_6&=\gamma_3t_1/(t_1+t_5)\,.
\end{align*}
Therefore, except on a subspace of Lebesgue measure $0$,
the pairwise sums $S_a$, $S_b$ and $S_c$ take distinct values,
all non-trivial splits violate the 4-point condition, and the
distance split tree is a star.
Furthermore, we see that
\[S_a<S_b=S_c \quad\Longleftrightarrow\quad \gamma_3=t_2/(t_1+t_2)\mbox{ and }
1-\gamma_7=\gamma_6 >  \frac{t_1t_2}{(t_1+t_2)(t_1+t_5)}
\]
in which case $oa \mid bc$ is the only non-trivial split satisfying
the 4-point condition, and forms the distance split tree.
\end{proof}

Turning to the proof of Theorem~\ref{thm:level-1-blockcut-tree}, we
introduce a few more definitions.
We first define network refinements that preserve
up-down paths and distances (Fig.~\ref{net-refinement-fig}).
They are defined
for networks of any level and at any polytomy,
so they are more general than the refinements described in
Theorem~\ref{thm:level-1-blockcut-tree}.

\begin{definition}[network refinements]
  Let $N$ be a semidirected network on $X$, $u$ a non-binary node
  (i.e.\ of degree $4$ or more) in $N$, and
  let $E(u)$ be the set of edges adjacent to $u$.
  Let $\{E_1, E_2\}$ be a partition of $E(u)$ such that
  $\abs{E_1}, \abs{E_2} \geq 2$ and all the
  incoming hybrid edges (into $u$), if any, are in $E_1$.
  Then the network $N'$ obtained
  by the following steps is called a
  \emph{refinement of $N$ at $u$ by 
  $E_1 \mid E_2$}:
  \begin{enumerate}
  \item Add a new node $u'$, and add a tree edge $uu'$ of length $0$;
  \item replace each edge $uv \in E_2$ by a new edge $u'v$.
  \end{enumerate}
  Further, if $B$ is a blob 
  and $u$ is a node in $B$,
  let $E_B(u)$ denote the set of edges in $B$ incident to $u$.
  If $E_1$ contains $E_B(u)$, then we call the resulting refinement a
  \emph{blob-preserving refinement} at $u$.
  If $E_1=E_B(u)$ and $E_2=E(u)\setminus E_B(u)$
  then we call the refinement the 
  \emph{canonical refinement} at $u$.
\end{definition}

\begin{figure}[h]
  \centering
  \includegraphics[scale=1.4]{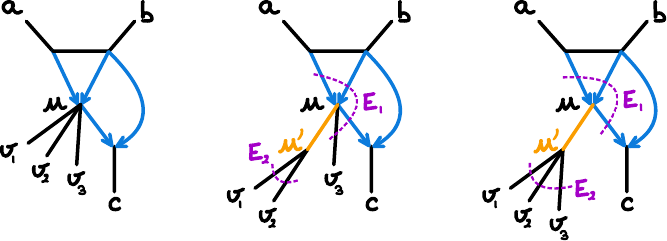}
  \caption{Left: network with a non-binary node $u$.
  Middle and right: refinements at $u$ for two choices of $E_1 \mid E_2$.
  In both, $E_1$ contains the 3 hybrid edges incident to $u$ and
  $E_2$ contains the edges incident to $v_1$ and $v_2$.
  The edge incident to $v_3$ is either in $E_1$ (middle) or in $E_2$ (right).
  The canonical refinement at $u$ is the rightmost network.}
  \label{net-refinement-fig}
\end{figure}

Since leaves must have degree one and refinements
  are defined at non-binary nodes, $u$ cannot be a leaf.
  Also, if either $E_1$ or $E_2$ has only one edge, then $u$ or $u'$ would become
  of degree 2, rendering the refinement uninteresting.
It is easy to see that $N'$ is still a valid
semidirected network, since for any rooted
network $N^+$ that induces $N$, one can keep the root and
direct the new edges consistently to get a rooted network ${N^+}'$
that induces $N'$. Namely, if $u$ is a hybrid node, then $uu'$ is directed
towards $u'$. Otherwise, we can direct $uu'$ depending
on whether the single parent of $u$ is in $E_1$ or in $E_2$.
In both cases, $uu'$ is a tree edge.

A blob-preserving refinement does
not change the non-trivial blobs, but adds a new trivial blob
$\{u'\}$:
suppose $B$ is a non-trivial blob in $N$, then
$E_2\subseteq E(u) \setminus E_B(u)$ contains cut edges only.
Therefore, the new tree edge $uu'$ is also a cut edge.
  The following lemma is a result of this property.

\begin{lemma}
  \label{lem:refinement-properties}
  Let $B$ be a blob
  in a metric semidirected network $N$ on $X$,
  $b$ the corresponding node in the tree of blobs $T$ of $N$,
  and let $u$ be a non-binary node in $B$.
  If $N'$ is a refinement of $N$ at
  $u$, then the pairwise distance on $X$ is the same on $N'$ and $N$.
  Furthermore, if $N'$ is a blob-preserving refinement at $u$ by the
  edge bipartition $E_1 \mid E_2$, then $\bt(N')$ is a refinement of
  $T$ at $b$ by
  $\widetilde{E_1} \mid \widetilde{E_2}$, where
  $\widetilde{E_2} = E_2$ 
  (which are cut-edges and appear in $T$)
  and $\widetilde{E_1} = E(b)\setminus\widetilde{E_2}$.
\end{lemma}

\begin{proof}
  For the first claim, it suffices to show that for any
  tips $x, y \in X$, there is a bijection between the sets of up-down
  paths $x \longleftrightarrow y$ in $N$ and in $N'$ that preserves the lengths
  and hybridization parameters.  Assume the refinement is by $E_1 \mid E_2$.
  Let $p$ be an up-down path between $x$ and $y$ in $N$.  Consider the
  following map $f$. 
  If $p$ does not include $u$, or if $p$ includes $u$ but the two edges
  incident to $u$ in $p$ are both in $E_1$, then $p$ is also an up-down path in
  $N'$, and we let $f(p) = p$.
  If $p$ includes $u$ and the two edges incident to $u$ in $p$ are both in $E_2$,
  then we may change $u$ to $u'$ in $p$ to obtain up-down path $p'$ in $N'$
  and set $f(p) = p'$.
  Finally, if $p$ includes $u$ with one incident edge in $E_1$ and one in
  $E_2$, then we may assume, without loss of generality, that
  $p = x \dots a u b \dots y$ with $au \in E_1$ and $bu \in E_2$.  Then let
  $p' = x \dots a u u' b \dots y$.  Since $uu'$ is a tree edge, $p'$
  has no v-structure and is an up-down path in $N'$.  We set
  $f(p) = p'$.
  It is easy to see that $f$ is injective, that $\gamma(f(p)) = \gamma(p)$,
  and $\ell(f(p)) = \ell(p)$.  As a result, the up-down paths
    in the image of $f$ 
  have hybridization parameters that sum up to one, so $f$ is surjective as well.

  For the second claim, let $u'$ be the new node introduced in the
  refinement.  As previously noted, $E_2$ contains cut edges only,
  so we are only deleting and adding
  cut edges during a blob-preserving refinement.  Hence all the
  operations correspond to the operations on the tree of blobs $T$.
  Consequently we can get $\bt(N')$ from $T$ by adding a node $b'$
  corresponding to the trivial blob $B' = \{u'\}$ in $N'$,
  cut edge $bb'$ corresponding to $uu'$, and replace edges $bc$ with
  $b'c$ for each cut edge $uv \in E_2$, where $c$ corresponds
    to the
  blob containing $v$.  This is the refinement of $T$ at $b$ by
  $\widetilde{E_1} \mid \widetilde{E_2}$.
\end{proof}

We now introduce definitions for splits that resolve a
polytomy in the tree of blobs without affecting the
blob itself in the network. Later, we prove that the distance
split tree resolves the tree of blobs with splits of this kind.

\begin{definition}[split along a blob; sibling groups]
  Let $N$ be a semidirected network on $X$, $T$ its tree of blobs,
  and $B$ a blob of $N$ with corresponding node $b$ in $T$.
  When $b$ is removed, suppose $T$ is disconnected into
  $k$ connected components, with taxa $Y_1, \ldots, Y_k$ in each.
  We call $\{Y_i; i\leq k\}$ \emph{the partition induced by $B$}.
  If a split $\sigma$ on $X$ has a set that is the union of two or more
  $Y_i$'s, then $\sigma$ is \emph{along the partition induced by
    $B$}, or \emph{along $B$}.
  Let $e_i$ be the cut edge in $T$ (or $N$) adjacent to $B$
  whose removal disconnects $Y_i$ from $B$.
  If $e_i$ and $e_j$ share a node $u\in B$ then
  $Y_i$ and $Y_j$ are called \emph{sibling sets of $B$ at $u$}.
  For a node $u\in B$, the \emph{sibling group} at $u$ is the
  union of all sibling sets of $B$ at $u$.
\end{definition}

\noindent
The following is a restatement of
Lemma~3.1.7 in \cite{charles03_phylog}, using our definitions.
\begin{lemma}
  \label{lem:extra-split}
  Let $T$ be a phylogenetic $X$-tree, and $T'$ a refinement of $T$.
  Then every split $\sigma \in \Sigma(T') \setminus \Sigma(T)$ is along
  some node $u$ of $T$.
\end{lemma}

\noindent
Next, we characterize the set of splits $\Sigma_D$
that satisfy the 4-point condition on $N$.
\begin{lemma}
  \label{lem:polytomy-split}
  Let $N$ be a level-$1$ network on $X$ with no degree-$2$ blobs, and
  with internal tree edges of positive lengths. Let $B$ be a blob of $N$
  of degree $4$ or more.
  If $B$ is trivial, then any split along $B$ satisfies the
  4-point condition. If $B$ is nontrivial, then a split
  along $B$ satisfies the 4-point condition if and only if it is of
  the form $S \mid \bar S$ where $S$ is a union of sibling sets of $B$.
  Furthermore, a split $\sigma$ along $B$ is in $\Sigma(T')$, $T'$
  being the distance split tree, if and only if $B$ is a nontrivial
  blob and $\sigma$ is of the form $S \mid \bar S$ where $S$ is a sibling group
  of $B$.
\end{lemma}

\begin{proof}
  If $B=\{u\}$ is trivial, then for any split
  $\sigma$ along $B$ we can find the corresponding refinement $N^R$ at
  $u$ with the extra edge inducing $\sigma$ in the tree of blobs.  Since
  $N^R$ has the same pairwise distances, $\sigma \in \Sigma(\bt(N^R))$
  satisfies the 4-point condition by
  lemma~\ref{ebt-edge-distance-compatible-lemma}.  
  Furthermore, we claim that for any split $\sigma$ along $B$, there
  is another split $\sigma'$ along $B$ that is incompatible with $\sigma$.
  This would imply that $\sigma \notin \Sigma(T')$.
  To show the claim, let
  $\{Y_i; i \leq d\}$ be the partition induced by $B$, with
  $d =\mathrm{deg}(B) \geq 4$.  Let
  $\sigma$ be of the form $\cup_{i \in I_1}Y_i  \mid  \cup_{i \in I_2} Y_i$, where
  $\{I_1, I_2\}$ is a bipartition of $\{1, \dots, d\}$.  Now we may choose
  $\sigma' = \cup_{i \in I_1'}Y_i  \mid  \cup_{i \in I_2'} Y_i$ where
  $\{I_1', I_2'\}$ is a bipartition of $\{1, \dots, d\}$ incompatible with
  $\{I_1, I_2\}$: such that $I_i\cap I'_j$ are all non-empty.
  Then $\sigma'$ is along
  $B$ and incompatible with $\sigma$.

  If $B$ is nontrivial, first consider $\sigma=S \mid \bar S$ where $S$ is a
  union of sibling sets of $B$, that is, $S$ contains the leaves
  corresponding to a set $E$ of cut edges adjacent to some node
  $u\in B$.  Then, in the blob-preserving
  refinement of $N^R$ at $u$ by $E(u)\setminus E \mid E$,
  the extra cut edge induces $\sigma$.  Hence
  $\sigma \in \Sigma(N^R)$ satisfies the 4-point condition.
  Conversely, consider a non-trivial split $\sigma=A \mid \bar A$ where
  both $A$ and $\bar A$ intersect at least two of the sibling groups of $B$.
  Let $Y_1,\ldots,Y_d$ be the sibling groups such that $Y_1$ is the
  sibling group at $B$'s hybrid node.
  Since $d=\mathrm{deg}(B) \geq 4$, it is easy to see that we can find distinct
  $\{i_1,\ldots,i_4\}$, $x_1,x_2\in A$ and $x_3,x_4\in\bar A$ such that
  $1\in\{i_1,\ldots,i_4\}$, 
  and $x_j\in Y_{i_j}$.
  Then the subnetwork $\tilde N = N_{\{x_1x_2x_3x_4\}}$ is
  equivalent to
  the leftmost network in Figure~\ref{fig:nonidentifiable}
  with positive branch lengths for both tree edges
  in the cycle.
  By Example~\ref{ex:4pc-4sunlet},
  the 4-point condition is not met for $x_1x_2 \mid x_3x_4$,
  which finishes the proof of the claim.

  Finally, consider a split $\sigma=S \mid \bar S$ along $B$ that satisfies
  the 4-point condition, but where $S$ is a proper subset
  of the sibling group at some node $u\in B$. Then $u$ must
  be adjacent to $k\geq 3$ cut edges, and $S$ must be the union
  of $l$ sibling sets at $u$, with $2\leq l<k$.
  Then similarly to the case when $B$ is trivial,
  we can find a nonempty union of sibling sets $S'$,
  such that $\sigma'=S' \mid \bar S'$ is incompatible
  with $\sigma$. Since $\sigma'$ satisfies the 4-point condition,
  $\sigma \notin \Sigma(T')$.
\end{proof}

\begin{proof}[Proof of Theorem~\ref{thm:level-1-blockcut-tree}]
  Note that the procedure described to obtain $N^R$ is a series of
  canonical refinements. So by Lemma~\ref{lem:refinement-properties},
  $N^R$ is a valid semidirected network with average distances
  identical to those in $N$.

  Let $T = \bt(N)$ and $T'$ the distance split tree of $N$.  By
  Lemma~\ref{lem:extra-split}, any split
  $\sigma \in \Sigma(T') \setminus \Sigma(T)$ is along some blob $B$
  of $N$.
  By Lemma~\ref{lem:polytomy-split}, there is no such extra split
  $\sigma$ when $B$ is trivial.  If $B$ is nontrivial, the extra
  splits must be of the form $S \mid \bar S$ where $S$ is a sibling group at
  some node $u$.
  Such a split corresponds to the split
  introduced in the canonical refinement at $u$.

  Finally, since $N^R$ can be obtained from a series of canonical refinements,
  by Lemma~\ref{lem:refinement-properties}, the tree of blobs of $N^R$ can
  also be obtained from the series of corresponding refinements, which
  introduces exactly all the extra splits described above.  As a
  result, $\bt(N^R) = T'$.
\end{proof}

\section{Identifying sunlets}
\label{sec:k-sunlet}

A $k$-sunlet is a semidirected network with a single $k$-cycle and
reticulation, and for each node on the cycle, one or more pendant
edge(s) (adjacent to a leaf).
The sunlet is binary if $k$ equals the number of leaves, $n$.
This section considers the problem of identifying the branch lengths
and hybridization parameters in a sunlet from the average
distances between the $n$ tips. We assume that we know the network is
a $k$-sunlet, but $k$ is unknown
and the ordering of the tips around the cycle is unknown.
In other words, we consider the problem of identifying
the exact network topology given that it is a sunlet.

A \emph{circular ordering} of the leaves $X=\{x_{1},\ldots,x_{k}\}$ is,
informally, the order of the leaves when placed around an undirected cycle.
Formally, it is the class of an ordering $(x_{i_1},\ldots,x_{i_k})$
up to the equivalence relations $(u_1,\ldots,u_k)\sim(u_k,\ldots,u_1)$
and $(u_1,u_2,\ldots,u_k)\sim(u_2,\ldots,u_k,u_1)$.

\subsection{4-sunlets}
\label{sec:4-sunlet}

First we consider the problem of identifying the lengths
and hybridization parameter in a binary $4$-sunlet,
assuming the labelled semidirected topology is known
(Fig.~\ref{4-sunlet-fig} left). Specifically, we
  assume that we know
 $h$ is of hybrid origin, $a$ and $b$ are its half-sisters and $g$
is opposite of the hybrid node.

\begin{figure}[h]
  \centering
  \includegraphics[scale=1.4]{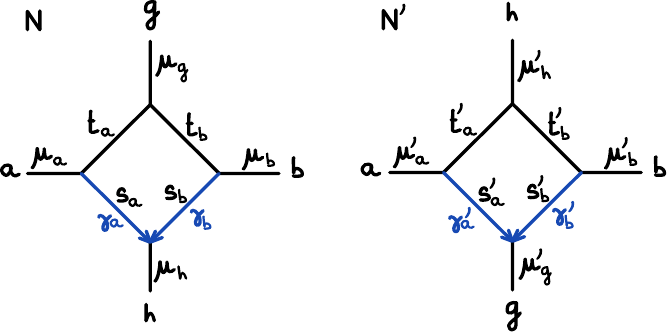}
  \caption{$4$-sunlets with the same undirected topology.
  Left: $h$ is of hybrid origin. Right: $g$ is of hybrid origin.
  By Theorem~\ref{thm:4-sunlet-switch-gh}, parameters can be chosen such that both networks
  have the same average distances between leaves.}
  \label{4-sunlet-fig}
\end{figure}

In this case, we have $6$ average distances, but $9$ parameters.
Zipping up the $4$-sunlet removes $2$ degrees of freedom,
but one free parameter still remains.
Specifically, we have
\begin{align}
\label{eq:4-sunlet-dist}
\begin{split}
  d_{ga} &= (\mu_a + t_a) + \mu_g \\
  d_{gb} &= (\mu_b + t_b) + \mu_g \\
  d_{ab} &= (\mu_a + t_a) + (\mu_b + t_b) \\
  d_{ah} &= (\mu_h + s_a\gamma_a + s_b\gamma_b) + \gamma_b(t_a + t_b) +
  \mu_a \\
  d_{bh} &= (\mu_h + s_a\gamma_a + s_b\gamma_b) + \gamma_a(t_a + t_b) +
  \mu_b \\
  d_{gh} &= (\mu_h + s_a\gamma_a + s_b\gamma_b) + \gamma_at_a +
  \gamma_bt_b + \mu_g
\end{split}
\end{align}

\begin{theorem}
  \label{thm:4-sunlet-metric}
  Let $d$ be a metric on four tips $\{a, b, g, h\}$.
  The $4$-sunlet $N$ with circular ordering $(a,g,b,h)$
  and in which $h$ is of hybrid origin (left of Fig.~\ref{4-sunlet-fig})
  has average distances $d$ for some set of parameters
  such that $t_a>0$ and $t_b>0$
  if and only if
  \begin{equation}\label{eq:4sunlet-4-point}
  d _{gh} + d_{ab} > \max \{d_{ah} + d_{bg}, d_{bh} + d_{ag}\}
  \end{equation}
  and
  \begin{equation}\label{eq:4sunlet-cond}
  \frac{d_{gh} + d_{ab} - d_{ah} - d_{bg}}{d_{ab}+d_{ag}-d_{bg}} +
    \frac{d_{gh} + d_{ab} - d_{bh} - d_{ag}}{d_{ab}+d_{bg}-d_{ag}} \leq 1.
  \end{equation}
  In this case, we can identify $\mu_g$ and the following composite parameters:
  \begin{align}
      \label{eq:4-sunlet-identify}
    \begin{split}
    \mu_g &= \frac{1}{2} (d_{ag}+d_{bg}-d_{ab}) \\
    \mu_a+t_a &= \frac{1}{2} (d_{ab}+d_{ag}-d_{bg}) \\
    \mu_b+t_b &= \frac{1}{2} (d_{ab}+d_{bg}-d_{ag}) \\
    l_h := \mu_h + s_a\gamma_a + s_b\gamma_b &= \frac{1}{2} (d_{ah}+d_{bh}-d_{ab}) \\
    \gamma_at_a &= \frac{1}{2} (d_{gh} + d_{ab} - d_{ah} - d_{bg}) \\
    \gamma_bt_b &= \frac{1}{2} (d_{gh} + d_{ab} - d_{bh} - d_{ag})\;.
  \end{split}
  \end{align}
  However, $\gamma$, $t_a$, $t_b$, $\mu_a$, $\mu_b$ are not identifiable.
  In particular, $\gamma_a$ can take any value in the following interval:
  \[\left[ \frac{d_{gh} + d_{ab} - d_{ah} - d_{bg}}{d_{ab}+d_{ag}-d_{bg}} \; , \;
    1 - \frac{d_{gh} + d_{ab} - d_{bh} - d_{ag}}{d_{ab}+d_{bg}-d_{ag}}\right]\,.\]
  Furthermore, \eqref{eq:4sunlet-4-point} is an equality if and only if one of the
  tree edges in the cycle has zero length: $t_a=0$ or $t_b=0$.
\end{theorem}

The proof below uses basic algebra.
Condition \eqref{eq:4sunlet-cond} ensures that 
\eqref{eq:4-sunlet-identify} can be solved
to give non-negative $\mu_a$ and $\mu_b$, and
\eqref{eq:4sunlet-4-point} ensures that
$\gamma_at_a>0$ and $\gamma_bt_b > 0$.
If \eqref{eq:4sunlet-4-point} is an equality, then the
4-point condition is satisfied and $d$ is a tree metric.
$\gamma_a=0$ or $1$ lead to a tree metric,
but hybrid edges are required to have $\gamma>0$ by definition.
Setting $t_a$ or $t_b$ to $0$ also leads to a tree metric.
Contracting the corresponding edge creates an unidentifiable degree-3 blob.

\begin{proof}[Proof of Theorem~\ref{thm:4-sunlet-metric}]
  It is easy to check with basic algebra that \eqref{eq:4-sunlet-identify}
  is equivalent to \eqref{eq:4-sunlet-dist}. Therefore, we simply need to
  show that additionally imposing \eqref{eq:4sunlet-4-point} and \eqref{eq:4sunlet-cond}
  is equivalent to requiring edge lengths be non-negative,
  $t_a,t_b>0$
  and hybridization parameters be in $(0,1)$.
  Suppose that $d$ comes from the $4$-sunlet $N$.
  Condition \eqref{eq:4sunlet-4-point} is
  equivalent to
  $d_{gh} + d_{ab} - d_{ah} - d_{bg} = 2 \gamma_at_a \geq 0$ and
  $d_{gh} + d_{ab} - d_{bh} - d_{ag} = 2 \gamma_bt_b \geq 0$.
  For condition \eqref{eq:4sunlet-cond}, we have
  \begin{equation*}
  \frac{d_{gh} + d_{ab} - d_{ah} - d_{bg}}{d_{ab}+d_{ag}-d_{bg}} +
  \frac{d_{gh} + d_{ab} - d_{bh} - d_{ag}}{d_{ab}+d_{bg}-d_{ag}}
  = \frac{\gamma_at_a}{\mu_a+t_a} + \frac{\gamma_bt_b}{\mu_b+t_b}
  \leq 1.
  \end{equation*}
  Conversely, suppose that a metric $d$ on $\{a,g,b,h\}$ satisfies
  \eqref{eq:4sunlet-4-point} and \eqref{eq:4sunlet-cond}.
  Then there exists $\tilde\gamma$ such that
  \[0 < \frac{d_{gh} + d_{ab} - d_{ah} - d_{bg}}{d_{ab}+d_{ag}-d_{bg}} \leq \tilde\gamma \leq
    1 - \frac{d_{gh} + d_{ab} - d_{bh} - d_{ag}}{d_{ab}+d_{bg}-d_{ag}} < 1\,.\]
  Then we can set $\gamma_a=1-\gamma_b=\tilde\gamma$ in
  \eqref{eq:4-sunlet-identify} to solve for $t_a$ first,
  getting $t_a>0$ from \eqref{eq:4sunlet-4-point}.
  Then solving for $\mu_a$, we get
  \[\mu_a = \frac{1}{2}\left((d_{ab}+d_{ag}-d_{bg}) - (d_{gh} + d_{ab} - d_{ah} - d_{bg}) / \tilde\gamma \right) \geq 0\]
  because of our condition on $\tilde\gamma$.
  Similarly, solving for $t_b$ and $\mu_b$ gives $t_b>0$ and $\mu_b\geq 0$.
\end{proof}

When the sunlet topology is unknown, we need to identify the circular ordering
of the tips around the cycle, and which of $a, b, g$ or $h$ is of hybrid origin.
Suppose the tips are labelled by $x, y, z, w$, then by Theorem~\ref{thm:4-sunlet-metric},
the opposing pairs $\{x, y\}$ and $\{z, w\}$ correspond to the largest
sum among $d_{xy} + d_{zw}$, $d_{xz} + d_{yw}$ and $d_{xw} + d_{yz}$.

Identifying the opposing pairs $\{a, b\}$ and $\{g, h\}$ is enough to identify the
undirected graph of the $4$-sunlet.  However, identifying
which tip is of hybrid origin is impossible, as we show below.

\begin{theorem}
  \label{thm:4-sunlet-switch-gh}
  Let $N$ be the $4$-sunlet with circular ordering $(a,g,b,h)$
  in which $h$ is of hybrid origin (Fig.~\ref{4-sunlet-fig}, left).
  Let $N'$ be the $4$-sunlet with the same circular ordering,
  but in which $g$ is of hybrid origin (Fig.~\ref{4-sunlet-fig}, right).
  For any parameters $(\ell,\gamma)$ on $N$, there exist
  parameters $(\ell',\gamma')$ on $N'$
  such that $N$ and $N'$ have the same average distances.
\end{theorem}

\begin{proof}
We apply Theorem~\ref{thm:4-sunlet-metric} to $N'$ and the distance $d$
obtained from $N$. We need to check that \eqref{eq:4sunlet-4-point} and
\eqref{eq:4sunlet-cond} are met.
Condition \eqref{eq:4sunlet-4-point} is met because it is symmetric in $g$ and $h$.
Condition \eqref{eq:4sunlet-cond} is not symmetric however.
To fit $d$ on $N'$, \eqref{eq:4sunlet-cond}
can be written as (after permuting $g$ and $h$):
\[\frac{d_{gh} + d_{ab} - d_{ag} - d_{bh}}{d_{ab}+d_{ah}-d_{bh}} + \frac{d_{gh} + d_{ab} - d_{bg} - d_{ah}}{d_{ab}+d_{bh}-d_{ah}} \leq 1\,.\]
Applying \eqref{eq:4-sunlet-identify} to $N$ from which $d$ is obtained,
we can rewrite the left-hand side as:
\begin{multline*}
\frac{\gamma_b t_b}{(\mu_a+t_a) -\gamma_a t_a + \gamma_b t_b} + \frac{\gamma_a t_a}{(\mu_b+t_b) -\gamma_b t_b + \gamma_a t_a} = \\ 
\frac{\gamma_b t_b}{\mu_a+ \gamma_b (t_a+t_b)} + \frac{\gamma_a t_a}{\mu_b + \gamma_a (t_a+t_b)} \leq
\frac{\gamma_b t_b}{\gamma_b (t_a+t_b)} + \frac{\gamma_a t_a}{\gamma_a (t_a+t_b)} = 1\,.
\end{multline*}
Hence \eqref{eq:4sunlet-cond} is met on $N'$, and parameters
can be set to match the average distances from $N$.
\end{proof}

Depending on the parameters in the $4$-sunlet,
it may be possible to switch $g$ with $a$ and $h$
with $b$ as well, if condition \eqref{eq:4sunlet-cond} holds for the network
in which $a$ (or $b$) is of hybrid origin.
Namely, this is possible if $\mu_h$ and $\mu_g$ are large enough to satisfy
\[\frac{\gamma_at_a}{\mu_g+\gamma_at_a} +
 \frac{\gamma_bt_b}{\mu_h + s_a\gamma_a + s_b\gamma_b+\gamma_bt_b} \leq 1.
 \]
Usually we do not have external information about which tip is of hybrid
origin, and even if we do, by Theorem~\ref{thm:4-sunlet-metric} we can only
identify $\mu_g$, the length of the branch ``across'' from the hybrid
node.  It is therefore not possible to identify the individual edge
lengths.
More generally, we can combine the swap lemma and
Theorem~\ref{thm:4-sunlet-metric} to
prove that almost any hybrid-closed subgraph with 4 boundary nodes can be swapped
with a 4-cycle without affecting distances.

\begin{proposition}[swap a subgraph with a 4-sunlet]
\label{prop:swap-4-cycle}
In a network $N$, let $A$ be a hybrid-closed connected subgraph
with 4 boundary nodes such that $\gamma_A\equiv 1$
and each node $u\in\partial A$ has degree 1 in $A$.
Let $t(u)$ denote the length of the edge incident to $u$ in $A$. Then there
exists $\eta\geq 0$ (which depends on $A\setminus \partial A$)
such that the following holds:
If $t(u)\geq\eta$ for each $u \in \partial A$, then we can swap $A$ with a tree
or with a 4-sunlet $A'$
on leaf set $\partial A$ to obtain a valid semidirected network $N'$ with $d_{N'} = d_N$.
\end{proposition}

\begin{proof}
To simplify notations, let $d$ denote $d(.,. \mid A)$.
If $d$ satisfies the 4-point condition, then there is a unique 4-taxon tree $A'$
on $\partial A$ such that $d_{A'}=d$ on $\partial A$ (and $\gamma_{A'} \equiv 1$).
Swapping $A$ with $A'$ leads to a valid semidirected network topology $N'$
because any valid root position in $N$ remains valid in $N'$.
$N'$ remains acyclic because $\gamma_A\equiv 1$ and $A$ is hybrid-closed:
when edges are directed away from the root, $A$ must have exactly
one ``entry" boundary node, whose incident edge in $A$ is outgoing.
Therefore, an undirected path between two nodes in $\partial A$
made of edges not in $A$ must have a v-structure, and then $N'$ cannot
contain directed cycles.
Finally, we can apply Lemma~\ref{lem:swap} to prove the claim with $\eta=0$.

If $d$ does not satisfy the 4-point condition, then $A$ must contain at least
one hybrid edge. We may label the nodes in $\partial A$ as $\{h,a,g,b\}$
such that \eqref{eq:4sunlet-4-point} holds for $d$ and
$h$ is below some hybrid edge in $A$.
Let $A'$ be a 4-sunlet on leaf set $\partial A$ with $h$ below the hybrid node
and circular ordering $(h,a,g,b)$.
Swapping $A$ with $A'$ leads to a valid semidirected network topology $N'$
because any valid root position in $N$ is not below $h$, and is again valid in $N'$.
We also have $\gamma_{A'} = \gamma_A\equiv 1$.
We now want to assign edge parameters in $A'$ such that $d_{A'} = d$ on $\partial A$.
By Theorem~\ref{thm:4-sunlet-metric}, this is possible provided that
\eqref{eq:4sunlet-cond} holds for $d$. 
Modifying $t(u)$ for $u\in\partial A$ does not modify the numerator of
either term in \eqref{eq:4sunlet-cond}. Let $aa_0$ and $b_0b$ be
the edges in $A$ incident to $a$ and $b$ respectively.
Then the denominators in \eqref{eq:4sunlet-cond} can be expressed as
$2t(a)+d_{a_0b_0}+d_{a_0g}-d_{b_0g}$ and
$2t(b)+d_{a_0b_0}+d_{b_0g}-d_{a_0g}$.
Therefore \eqref{eq:4sunlet-cond} holds if $t(a)>\eta$ and $t(b)>\eta$
where $\eta$ is the maximum of
$d_{gh} + d_{a_0b_0} - d_{a_0h} - d_{b_0g} - (d_{a_0b_0}+d_{a_0g}-d_{b_0g})/2$
and
$d_{gh} + d_{a_0b_0} - d_{b_0h} - d_{a_0g} - (d_{a_0b_0}+d_{b_0g}-d_{a_0g})/2$.
This concludes the proof by Lemma~\ref{lem:swap}.
\end{proof}

We can now prove Theorem~\ref{thm:nonidentifiable}
on networks of level up to $k$, $k\geq 2$.
\begin{proof}[Proof of Theorem~\ref{thm:nonidentifiable}]
It suffices to consider $k = 2$.
Consider the networks in Fig.~\ref{fig:nonidentifiable},
say $A_1$ on the left and $A_2$ on the right.
Let $n\geq 4$. If $n=4$, set $N_1=A_1$ and $N_2=A_2$.
If $n\geq 5$, we can form networks $N_i$ ($i=1,2$) with $n$ taxa
by replacing the leaves $a$, $b$, $c$ and/or $d$ in $A_i$ by subtrees
with enough taxa.
Given any values for the parameters labelled in Fig.~\ref{fig:nonidentifiable}
for $A_2$ such that $u_4>0$,
$d_{A_2}$ satisfies \eqref{eq:4sunlet-4-point} with the same
ordering as $d_{A_1}$.
By Proposition~\ref{prop:swap-4-cycle} and its proof,
we can swap $A_2$ with $A_1$ 
provided that the edges incident to $b$ and $c$ are long enough
in $A_2$.
It follows that $d_{N_1}=d_{N_2}$ for parameters
in subsets of positive Lebesgue measure.
\com{
To find choices of parameters, we can match the split weights
between networks.
We use the
notations in Fig.~\ref{fig:nonidentifiable} for internal edges
and subscripts with leaf labels for external edges
(e.g. length $t_a$ in $N_1$ for the length of the edge to $a$).
For any parameters on $N_3$, we first zip-up $N_3$ then assign the
following parameters in $N_1$ to get $d_{N_1}=d_{N_3}$:
$t_a=u_a$, $t_d=u_d$, $t_c=u_c$,
$\gamma  = (\gamma_1 u_4 + u_6)/(u_4+u_6)$,
$t_4 = (1-\gamma_2)(u_4+u_6)$,
$t_3 = (\gamma_2 + \frac{\gamma_1 u_3}{\gamma_1 u_4 + u_6}) (u_4+u_6)$
and
$t_b = (1-\gamma_1) u_3 u_6/(\gamma_1 u_4 + u_6)$.
This example holds with the full parameter space for $N_3$
and its image by the assignment above, which has positive
Lebesgue measure in the parameter space for $N_1$.
For any parameters on $N_2$, we may assume that $N_2$ is
zipped up with $v_1=v_2=v_3=v_5=0$, up to increasing $v_b$.
Then we define the following parameters on $N_1$:
$\gamma=\gamma_1 + (1-\gamma_1)\gamma_2\in (0,1)$,
$t_4=v_4$,
$t_3=v_6+\gamma_1v_7/\gamma$,
$t_a=v_a+\gamma v_4$ (this is $\hat\mu_a$),
$t_d= v_d + (1-\gamma)v_6 + (1-\gamma_1)v_7$ (this is $\hat\mu_d$),
$t_b+(1-\gamma)t_3=v_b$ (this is $\hat\mu_b$),
$t_c+\gamma t_4=v_c$ (this is $\hat\mu_c$),
provided that $v_b$ and $v_c$ are not too small.
}
\end{proof}

\begin{definition}[canonical 4-sunlet split network]\label{def:canonical-4sunlet}
Consider a 4-sunlet whose undirected topology has
circular ordering $(a,g,b,h)$ (e.g.\ Fig.~\ref{4-split-network} right)
and with cycle tree edges of positive lengths.
The underlying undirected graph (e.g.\ Fig.~\ref{4-split-network}
left) can be considered as a \emph{split network}, in which
each pair of parallel edges identifies a single split and a single split weight (edge length),
with \emph{canonical} edge lengths defined 
as follows.
\begin{align}
  \label{eqn:canonical-edge-lengths}
  \hat{\mu}_g &= \frac{1}{2}(d_{ga} + d_{gb} - d_{ab}) &
  \hat{\mu}_a &= \frac{1}{2}(d_{ag} + d_{ah} - d_{gh}) \nonumber \\
  \hat{\mu}_h &= \frac{1}{2}(d_{ha} + d_{hb} - d_{ab}) &
  \hat{\mu}_b &= \frac{1}{2}(d_{bg} + d_{bh} - d_{gh}) \\\nonumber
  \hat{t}_{hb \mid ga} &= \frac{1}{2}(d_{gh} + d_{ab} - d_{ga} - d_{hb}) &
  \hat{t}_{ha \mid gb} &= \frac{1}{2}(d_{gh} + d_{ba} - d_{gb} - d_{ha})\,.
\end{align}
Distances on this canonical split network, calculated between any two
tips as the length of the shortest path between them
\cite{huson_rupp_scornavacca_2010}, are identical to the average
distances on the original semi-directed 4-sunlet.
\end{definition}

\begin{figure}[h]
  \centering
  \includegraphics[scale=1.5]{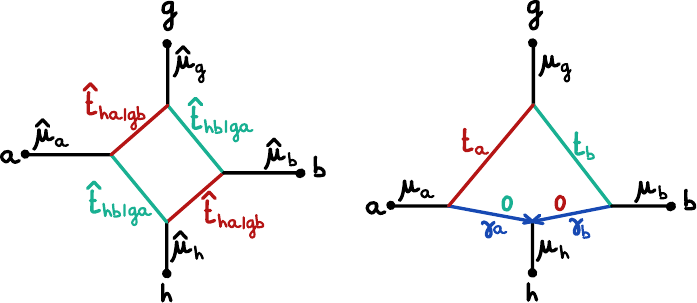}
  \com{\scalebox{0.7}{
  \begin{tikzpicture}
    \grsetup
    \Vertices[Math,unit=4]{circle}{b, g, a, h}
    \SetVertexNoLabel
    \Vertices[unit=2.5]{circle}{c, d, e, f}
    \Edge[label=$\hat{\mu}_g$](g)(d)
    \Edge[label=$\hat{\mu}_a$](a)(e)
    \Edge[label=$\hat{\mu}_b$](b)(c)
    \Edge[label=$\hat{t}_{hb \mid ga}$](c)(d)
    \Edge[label=$\hat{t}_{ha \mid gb}$](d)(e)
    \Edge[label=$\hat{\mu}_h$](f)(h)
    \Edge[label=$\hat{t}_{hb \mid ga}$](e)(f)
    \Edge[label=$\hat{t}_{ha \mid gb}$](c)(f)
  \end{tikzpicture}
  }}
  \caption{
  Left: split network, with canonical edge lengths from Definition~\ref{def:canonical-4sunlet}.
  A pair of parallel edges represent the same split and thus share the same length
  ($\hat{t}_{hb \mid ga}$ or $\hat{t}_{ha \mid gb}$ here).
  The distance between two nodes is defined as the length of the shortest path between them.
  Right: example zipped-up $4$-sunlet represented by the split network on the left,
  in which $h$ is of hybrid origin. Distances (using up-down paths) are identical to distances
  on the split network and satisfy \eqref{eqn:4-block-split-lengths}.
  }
  \label{4-split-network}
\end{figure}
This split network provides a unique representation of what can
be identified from pairwise distances: undirected topology and identifiable
composite parameters.
Ideally, we would have liked a semi-directed representation,
but since the location of the hybrid node is not identifiable, this was not an option.
\begin{theorem}[identifiability of $4$-sunlet split network]
  \label{thm:4-sunlet-identify}
  Let $N$ and $N_0$ be binary $4$-sunlets with identical leaf set
  and internal tree edges of positive lengths.
  If $N$ and $N_0$ have identical average distances, then
  the canonical $4$-sunlet split networks of $N$ and $N_0$ are identical.
\end{theorem}
\begin{proof}
  The positivity of cycle tree edge lengths ensures that
\begin{equation*}
d _{ab} + d_{gh} > \max\{d_{ag} + d_{bh}, d_{ah} + d_{bg}\}
\end{equation*}
is satisfied strictly, so
$N$ and $N_0$ must have the same circular ordering.
Finally, the definition of canonical edge lengths from average distances
in \eqref{eqn:canonical-edge-lengths} is symmetric with respect
to the hybrid node: canonical lengths depend on the circular ordering
only.
\end{proof}
There is a tight correspondence between edge lengths in the semi-directed
network and edge lengths in the split network, provided that the placement of the
hybrid node is known.
For example, if $h$ is of hybrid origin and if the network is zipped up
as in Fig.~\ref{4-split-network} (right), then by
  Theorem~\ref{thm:4-sunlet-metric} we have that:
\begin{align}
  \label{eqn:4-block-split-lengths}
  \hat{\mu}_g &= \mu_g, & \hat{\mu}_a &= \mu_a + \gamma_b t_a, & \hat{t}_{ha \mid gb} &= \gamma_a t_a,
  \nonumber \\
  \hat{\mu}_h &= \mu_h, & \hat{\mu}_b &= \mu_b + \gamma_a t_b, & \hat{t}_{hb \mid ga} &= \gamma_b t_b.
\end{align}

By
\eqref{eqn:4-block-split-lengths} we have that
$\hat{\mu}_s \geq \mu_s$ for each cut edge.
In fact, $\hat{\mu}_s = \mu_s$ is a correct length
estimate for the zipped-up child edge of the hybrid node and for the cut edge
opposite to the hybrid node.
For the other cut edges, $\hat{\mu}_s$ is an overestimate of $\mu_s$.
For example, for the network in Fig.~\ref{4-split-network} (right)
where
$h$ is the hybrid node, then $\hat{\mu}_a = \mu_a + \gamma_b t_a \geq \mu_a$.
We do not know which cut edge length is correctly represented, however.
Similarly, $\hat{t}_{hb \mid ga}$ and $\hat{t}_{ha \mid gb}$ are
underestimates of the length of \emph{tree} edges in the cycle,
although we do not know which edges in the cycle are tree or hybrid edges.

\smallskip
If a 4-sunlet has a polytomy, its canonical split network can be defined
(and Theorem~\ref{thm:4-sunlet-identify} can be applied)
after resolving the polytomy with an extra edge of length 0.
As Fig.~\ref{fig:4cycle} shows, networks with polytomies adjacent
to a 4-cycle may have the same average distances as networks without
polytomies.

\begin{figure}[h]
  \centering
  \includegraphics[scale=1.5]{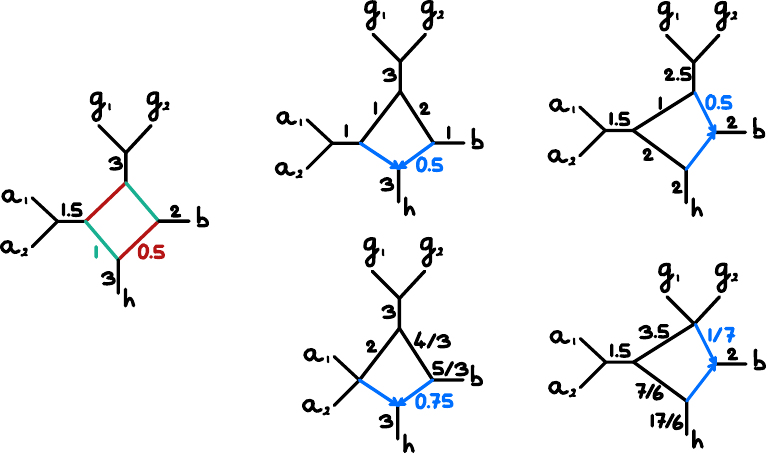}
  \caption{Canonical split network representation (left)
  of 4 zipped-up semidirected networks with identical average distances,
  in which $h$ (middle) or $b$ (right) is of hybrid origin.
  Two networks (bottom) have a polytomy adjacent to the 4-cycle.
  In the split network, pairs of parallel edges of identical color
  represent a single split and share the same length (split weight).
  Hybrid edges (arrows) have length $0$ and
  inheritance $\gamma$ shown in blue for one of them.
  Numbers in black indicate edge lengths.
  The split network shows the 6 composite parameters identifiable
  from distances, pertaining to the 4-cycle. Zipped-up semidirected networks
  have 7 associated parameters.
  }
  \label{fig:4cycle}
\end{figure}

\subsection{$k$-sunlet for $k \geq 5$}

With $5$ or more nodes in the cycle, we can identify the topology,
branch lengths, and hybridization parameters of the zipped-up version
of the sunlet.

\begin{theorem}[$k$-sunlet identifiability, $k\geq 5$]
  \label{thm:5-sunlet-identify}
  Let $N$ and $N_0$ be semidirected networks with identical leaf set
  $\{u_0,\ldots, u_{n-1}\}$ and internal tree edges having positive lengths, such that
  $N$ is a $k$-sunlet and $N_0$ is a $k_0$-sunlet with $k_0 \geq 5$.  If
  $N$ and $N_0$ have identical average distances, then the zipped-up
  versions of $N$ and $N_0$ are identical.
\end{theorem}

\begin{proof}
We first we show that $N$ and $N_0$
must have the same topology, and then
the same branch lengths and hybridization parameter.
In $N_0$, let the hybrid node be $v_0$,
and let the other internal nodes be $v_1, \ldots, v_{k-1}$ such that
$v_{i-1}$ and $v_{i}$ are neighbors (as in Fig.~\ref{fig:k-sunlet}).
Let $C_i$ be the set of leaves adjacent to $v_i$ in $N_0$.
If there are no polytomies, then each $C_i$ is reduced to a single leaf $u_i$.

\begin{figure}[h]
  \centering
  \includegraphics[scale=1.5]{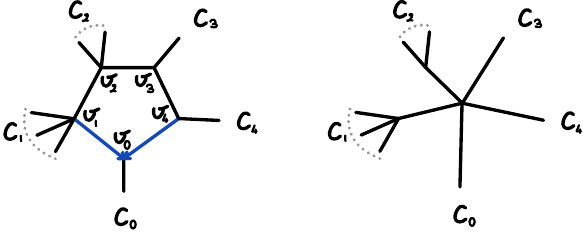}
  \caption{Left: $5$-sunlet $N$ on $8$ taxa. Its tree of blobs is a star.
  $C_i$ is the sibling group at $v_i$ ($i=0,\dots,4$).
  Right: distance split tree constructed from $d_N$.
  It is a refinement of the tree of blobs, but identifies the polytomies in the sunlet.}
  \label{fig:k-sunlet}
\end{figure}

By Lemma~\ref{lem:polytomy-split}, each non-trivial split in the
distance split tree of $N$ is of the form $C  \mid  \bar C$, where $C$ is a
set of all the sister leaves that are adjacent to the same cycle node
in $N$, and $\bar C = X\setminus C$.  The same holds for $N_0$.

Since $N$ and $N_0$ have identical pairwise distances, they have the
same distance split tree, and consequently identical sets of sister
leaves (polytomies).  In particular, we must have that $k=k_0$.  So
without loss of generality, we choose a single representative leaf
from each $C_i$ for the remainder of the proof.
In other words, we may
assume that both $N$ and $N_0$ have no polytomies
and $k=n$ (as in Fig.~\ref{fig:k-sunlet-tree-part}).

\paragraph{Identifying the circular ordering of leaves.}
First, we claim that the leaf of hybrid origin $u_0$ is the only leaf $u$ such that:
\begin{equation}\label{hybridproperty}
\forall x,y,z, \{u,x,y,z\} \mbox{ does \emph{not} satisfy the 4-point condition.}
\end{equation}
Indeed, if $u=u_0$, then $\{u,x,y,z\}$ induces in $N_0$ a 4-sunlet in which
both tree edges in the cycle have positive length, so $\{u,x,y,z\}$ does not
satisfy the 4-point condition by Theorem~\ref{thm:4-sunlet-metric}.
If $u=u_i$ for $i>0$, then we can
choose $3$ other leaves $u_j$, $u_k$, $u_l$ different from $u_0$, because $k_0\geq 5$.
The induced subnetwork is then a tree, so the 4-point condition holds,
such that $u_i$ does not satisfy \eqref{hybridproperty} for $i\geq 0$.
Therefore, the leaf of hybrid origin must be the same in $N$ as in $N_0$: $u_0$.

Next, we consider the subnetwork of $N_0$ induced by the leaves other
than $u_0$.  (For sake of brevity, in what follows in this subsection,
all the subnetworks have degree-$2$ nodes suppressed.)  This
subnetwork has no reticulation, it is binary and its internal branch
lengths are positive, so it is equal to its tree of blobs and
its distance split tree. Therefore, the subnetwork of $N$ induced by the
leaves other than $u_0$ has the same tree topology. This tree must be
a caterpillar (Fig.~\ref{fig:k-sunlet-tree-part}, middle) with two
cherries: $\{u_1,u_2\}$ and $\{u_{k-2},u_{k-1}\}$ and internal nodes
that correspond to $v_2,\ldots,v_{k-2}$.  Its topology determines the
ordering of the other leaves.  In other words, the ordering
of  $u_3,\ldots,u_{k-3}$ must be identical in $N$ and in
$N_0$. We can also match the internal nodes $v_3,\ldots,v_{k-3}$ in
$N_0$ to internal nodes in $N$.  What remains to be identified is
which of $\{u_1,u_2\}$ and which of $\{u_{k-2},u_{k-1}\}$ is adjacent
to either parent of the hybrid node in $N$.
For this, consider the subnetworks from $N$ and $N_0$ induced by
$\{u_2, u_1, u_0, u_{k-1}\}$.
By \eqref{eq:4sunlet-4-point} in Theorem~\ref{thm:4-sunlet-metric},
the average distances on $\{u_0, u_1,u_2,u_{k-1}\}$ determine the
circular ordering of these 4 taxa, such that $u_1$ must be adjacent
to a hybrid parent in $N$, as it is in $N_0$.
Similarly, the average distances on $\{u_0, u_1,u_{k-2},u_{k-1}\}$
determine the circular ordering of these 4 taxa such that
$u_{k-1}$ must be adjacent to a hybrid parent in $N$, like in $N_0$.
This finishes the proof that the circular ordering of leaves is identical
in $N$ and in $N_0$.

\begin{figure}[h]
  \centering
  \includegraphics[scale=1.5]{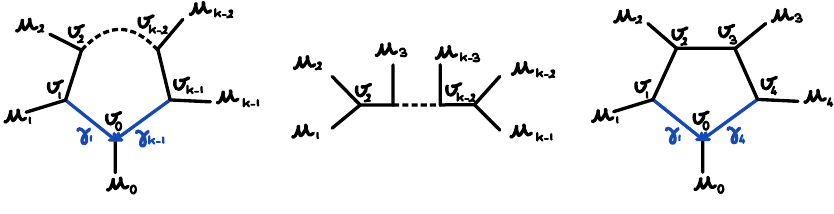}
  \caption{Left: binary $k$-sunlet.
  Middle: after excluding $u_0$, the subnetwork is a tree.
  Right: subnetwork on $\{u_0, u_1, u_2, u_{k-2}, u_{k-1}\}$, a sunlet with $k=5$.}
  \label{fig:k-sunlet-tree-part}
\end{figure}

\paragraph{Identifying branch lengths and hybridization parameters.}
By considering distances between $u_1, \ldots, u_{k-1}$, we can
determine the lengths of all the edges of the caterpillar tree $N_{\{u_1\ldots,u_{k-1}\}}$
(Fig.~\ref{fig:k-sunlet-tree-part}, middle). In particular, we get that the
following edges have the same length in $N$ as in $N_0$:
$v_i v_{i+1}$ for $2\leq i\leq k-3$ 
and $v_i u_i$ for $2\leq i\leq k-2$. 
The parameters that remain to be identified are in the subnetwork
induced by $\{u_0, u_1, u_2, u_{k-2}, u_{k-1}\}$.
Therefore, we may assume that $k=5$, as we do below
(Fig.~\ref{fig:k-sunlet-tree-part}, right).

For brevity in this paragraph, for an edge $uv$ we also write its length as $uv$.
From the tree
induced by $\{u_1,\ldots,u_4\}$, we have the lengths $u_2v_2$,
$u_3v_3$, $v_2v_3$, $v_3v_4 + u_4v_4 $, and
$u_1v_1 + v_1v_2$.
From the subnetwork induced by $\{u_0,u_1,u_2,u_4\}$, by
Theorem~\ref{thm:4-sunlet-metric} we can identify
$l_h = \gamma_1\cdot v_1v_0 + \gamma_4\cdot v_4v_0 + u_0v_0$,
which is the length of $v_0u_0$ after unzipping.
We can also identify
$\gamma_1\cdot v_2v_1$ and $\gamma_4(v_2v_3+v_3v_4)$.
From the subnetwork on $\{u_0,u_1,u_3,u_4\}$, we also get
$\gamma_1(v_1v_2+v_2v_3)$.
Hence we can identify $\gamma_1$ as
$\big(\gamma_1(v_1v_2+v_2v_3) - \gamma_1\cdot v_1v_2\big) / v_2v_3$
from pairwise distances.
All other parameters in the unzipped version of $N$ are also identifiable, using:
$v_1v_2 = \gamma_1 \cdot v_1v_2 / \gamma_1$,
$\gamma_4 = 1 - \gamma_1$,
$v_3v_4=\big(\gamma_4(v_2v_3+v_3v_4)\big)/\gamma_4 - v_2v_3$,
and as a result $u_1v_1$ and $u_4v_4$.

\end{proof}

\section{Identifying level-1 networks}
\label{sec:level-1}

While a degree-3 blob is not detectable, a $4$-cycle in a
level-$1$ network corresponds to a polytomy in the tree of blobs.
Its hybrid node and its zipped-up version is unidentifiable,
but the canonical split network of a 4-sunlet is identifiable,
by Theorem~\ref{thm:4-sunlet-identify}.
To prove Theorem~\ref{thm:level1-4up},
we first define mixed networks formally.

\subsection{Mixed network representation}

In the mixed representation of a semi-directed level-1 network,
the cycles of size $5$ or greater are unchanged. The $4$-cycles,
which are only partially identifiable, are replaced by split networks,
extending the split network representation of $4$-sunlets from
section~\ref{sec:4-sunlet} with canonical edge lengths given by
\eqref{eqn:canonical-edge-lengths}.

\begin{definition}[mixed network]
A \emph{mixed network} is a semidirected graph where undirected edges
are partitioned into two sets: tree edges $E_T$ and split edges $E_S$;
and where $E_S$ is itself partitioned into a set of classes.
When the graph is embedded in a Euclidean space, split edges
within the same class are represented as parallel segments.
A \emph{metric} $(\ell, \gamma)$ on a mixed network $M$ is such that
$\ell: E \to \RR_{\geq 0}$ assigns the same length to all edges in the same
class of split edges; and $\gamma: E \to [0, 1]$ assigns $\gamma(e) = 1$
if $e$ is undirected and $\gamma(e) \in (0, 1)$ if $e$ is directed.
\end{definition}

\begin{definition}[mixed network representation of a level-1 network]
  \label{def:mixedrep}
  Let $N$ be a level-$1$ semidirected network with no
  $2$- or $3$-cycles.  The \emph{mixed network representation}
  $N^*$ of $N$ is the mixed network obtained as follows:
  \begin{enumerate}
  \item In each $4$-cycle,
  \com{ the $4$ edges in the cycle are set to be undirected
  split edges categorized into $2$ classes.
  Each class is made of a pair of opposite edges:
  one tree edge and one hybrid edge in $N$.
  Adjacent tree edges are added and edge lengths are defined as shown in
  Fig.~\ref{fig:4-cycle-split-network}. Note that the hybrid edges in $N$
  that are part of a $4$-cycle become split edges with $\gamma=1$.
  }
  the subgraph on the left of Fig.~\ref{fig:4-cycle-split-network} is excised
  and replaced with that on the right.
  \item Suppress any degree-$2$ node.
  \end{enumerate}
  In $N^*$, $4$-cycles consist of split edges and are called \emph{split blobs}.
\end{definition}

\begin{figure}[h]
  \centering
  \includegraphics[scale=1.5]{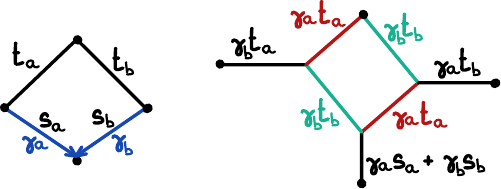}
  \caption{Mixed network: representation of a $4$-cycle in a
  semidirected network $N$ (left) to form a split blob (right).
  Edges in the cycle are converted to undirected split edges
  with $\gamma=1$, categorized in two classes depicted by colors.
  Adjacent tree edges are added.
  If $N$ is zipped up, then $s_a=s_b=0$ and the lower tree edge
  (adjacent to the hybrid node in $N$) is not needed.}
  \label{fig:4-cycle-split-network}
\end{figure}

\noindent
In a mixed network $M$, a \emph{mixed up-down path} between two nodes $a,b$
is a path $p = u_0 u_1\dots u_n$ between $u_0 = a$ and $u_n = b$ in $U(M)$
such that:
\begin{enumerate}
  \item $p$ has no v-structure, that is, no segment $u_{i-1}u_iu_{i+1}$
  such that $(u_{i-1}u_i)$ and $(u_{i+1}u_i)$ are directed edges in $M$;
  \item if a segment $u_i \dots u_j$ consists solely of split edges,
    then the segment is a shortest path between $u_i$ and $u_j$ in
    $U(M)$.
\end{enumerate}
Given a metric $(\ell, \gamma)$ on $M$, the length of $p$
is $\ell(p) = \sum_{e \in p} \ell(e)$, and the probability of $p$ is
$\gamma(p) = \prod_{e \in p} \gamma(e)$.
A \emph{split segment} is a path that consists solely of split edges.
Two split segments $S_1, S_2$ are \emph{equivalent} if
they have the same endpoints and $\ell(S_1) = \ell(S_2)$.
Note that a split segment $S$ must have $\gamma(S) = 1$.
Two mixed up-down paths $p$ and $q$ are \emph{equivalent} if one can obtain $q$ from
$p$ by replacing some split segments of $p$ by equivalent split segments.

\begin{definition}[average distance in a mixed network]
  The \emph{average distance} between two nodes $u, v$ in a mixed
  network $M$ is the weighted average length of mixed up-down paths,
  up to equivalence, between $u$ and $v$:
  \[d_M(u, v) = \sum_{p \in P_{uv}} \gamma(p)\ell(p)\] where
  $P_{uv}$ is a set of mixed up-down paths between $u$ and $v$,
  containing exactly one representative from each equivalence class.
  This distance is well-defined
  because equivalent paths have the same lengths and probabilities.
\end{definition}

\noindent
Importantly, average distances are preserved by the mixed representation
of a level-1 network:
\begin{theorem}\label{thm:mixedrepdistances}
  Let $N$ be a level-$1$ semidirected network on taxon set $X$,
  and $N^*$ the mixed network representation of $N$.  Then for any
  $x, y \in X$,
  \[d_N(x, y) = d_{N^*}(x, y).\]
\end{theorem}
\noindent
The proof, in appendix~\ref{sec:mixedrep-proofs}, first shows that average distances in
mixed networks can be interpreted as the expected shortest path length
over ``displayed split networks''.

\subsection{Identifying the mixed representation of level-1 networks}

We now have the tools to prove Theorem~\ref{thm:level1-4up}.

\begin{proof}[Proof of Theorem~\ref{thm:level1-4up}]
  For $i=1,2$, let $N_i$ be a zipped-up level-1 semidirected network on $X$ with
  internal tree edges of positive lengths, and $N_i^R$ the
  refinement described in Theorem~\ref{thm:level-1-blockcut-tree}.
  Also let $N_i^{R*}$ be the mixed representation of $N_i^R$, and $N_i^*$
  the mixed representation of $N_i$.
  Assume that $d_{N_1}=d_{N_2}=d$. By
  Theorem~\ref{thm:level-1-blockcut-tree}, $N_1^R$ and $N_2^R$ have the same
  tree of blobs $T$, so their blobs and cut edges are in bijection.

  Let $b$ be a node of degree $k\geq 4$ in $T$, and
  $\{e_1,\ldots,e_k\}$ be the cut edges incident to $b$ in $T$.
  Let $B_i$ be the corresponding cycle in $N_i^R$ ($i=1,2$), of length $k\geq 4$.
  Removing $e_j$ disconnects $T$ into two components.
  We select a leaf $x_j$ from the component that does not contain
  $b$ and use distances on $\{x_1,\ldots,x_k\}$.
  If $k\geq 5$, $B_1$ and $B_2$ have the same topology and edge parameters
  by Theorem~\ref{thm:5-sunlet-identify}.
  If $k=4$, the split cycle representing $B_1$ (in $N_1^{R*}$) and
  $B_2$ (in $N_2^{R*}$) have the same topology and canonical length of split
  edges, by Theorem~\ref{thm:4-sunlet-identify}.
  So $N_1^{R*}$ and $N_2^{R*}$ have the same topology (referred
  below as $N^{R*}$) and same parameters
  for edges within a blob.

  \medskip
  Next, we need to prove that cut edges have the same length in
  $N_1^{R*}$ and $N_2^{R*}$. These edges are also cut
  edges in $N_i^R$ ($i=1,2$) and $T$.
  Let $e=uv$ be a cut edge and $B(u), B(v)$ the blobs that
  $u$ and $v$ belong to.
  We can select leaves as follows.
  If $B(u)=\{u\}$ (resp. $B(v)=\{v\}$) is an internal tree node that is not in
  any cycle, let $e_1$ and $e_2$ (resp. $e_3$ and $e_4$) be the two
  cut edges incident to $u$ (resp. $v$) besides $e$.
  If $u$ (resp. $v$) is part of a cycle, let $e_1$ and $e_2$
  (resp. $e_3$ and $e_4$) be the cut edges incident to the two nodes
  adjacent to $u$ (resp. $v$) in that cycle.
  We then choose a leaf $x_j$, $j=1,2$ (resp. $j=3,4$)
  in the connected component disjoint from $B(u)$ (resp. $B(v)$)
  when $e_j$ is removed from $T$
  (see Fig.~\ref{fig:mixed-subnetwork}).
  Finally, if $u$ (resp. $v$) is a leaf, we set $x_1=x_2=u$ (resp. $x_3=x_4=v$).
  Note that $\{x_1,x_2,x_3,x_4\}$ contains at least 3 distinct leaves because
  $N^{R*}$ has no degree-2 blob.
  If $u$ and $v$ do not have an outgoing hybrid edge in $N^{R*}$,
  then the length of $e$ must be the same in $N_1^{R*}$ and $N_2^{R*}$
  and equal to
  \[\ell(e)= t :=\frac{1}{2}\left(d(x_1,x_3)+d(x_2,x_4) - d(x_1,x_2)-d(x_3,x_4)\right)\]
  because $N_i^{R*}$ is unzipped and because this equality holds
  on a tree and
  for external edges of 
  canonical degree-$4$ split sunlets.
  If $u$ (resp. $v$) has an outgoing edge in $N^{R*}$, then
  $B(u)$ (resp. $B(v)$) is a $k$-cycle with $k\geq 5$, $u$ (resp. $v$)
  is adjacent to its hybrid node,
  and $B(u)$ (resp. $B(v)$) is reduced to a 3-cycle in the subnetwork
  induced by $\{x_1,\dots,x_4\}$. By Proposition~\ref{prop:3-cycle}, shrinking
  this 3-cycle makes $u$ (resp. $v$) of degree 2, incident to $e$ and to a new edge
  of length $t_u$ (resp. $t_v$). This new edge length is known because
  parameters are known for all cycle edges in $B(u)$ (resp. $B(v)$).
  We can then identify $\ell(e)$ by subtracting $t_u$ (and/or $t_v$) from $t$.

  At this point we have that $N_1^{R*}=N_2^{R*}$.  By contracting the tree edges
  of length $0$ we get that $N_1^* = N_2^*$. This
  finishes the proof of Theorem~\ref{thm:level1-4up}.
\end{proof}

\begin{figure}[h]
  \centering
  \includegraphics[scale=1.5]{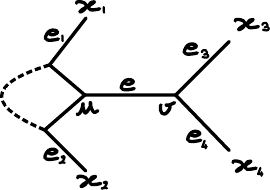} 
  \caption{Taxon sampling to cover a cut edge $e = uv$.
  In this example, the blob containing $v$ is $\{v\}$ and the blob containing $u$
  is non-trivial.
  The subnetwork on $\{x_1,x_2,x_3,x_4\}$ includes $e$ and
  both $u$ and $v$ as degree-3 nodes.}
  \label{fig:mixed-subnetwork}
\end{figure}

\backmatter
\bmhead{Acknowledgments}
We are grateful to Sebastien Roch for insightful discussions on this work
and John Fogg for detailed feedback.
We thank 3 reviewers, whose comments greatly improved the paper.
This work was supported in part by the National Science Foundation,
(DMS 1902892 
and DMS 2023239) 
and by a H. I. Romnes faculty fellowship (to C.A.) provided by the
University of Wisconsin-Madison Office of the Vice Chancellor for Research and Graduate Education
with funding from the Wisconsin Alumni Research Foundation.

\begin{appendices}

\section{Proofs of technical results}

\subsection{Up-down paths and average distances}
\label{sec:path-dist-proofs}

\begin{proof}[Proof of Proposition~\ref{prop:up-down-path}]
Proposition~\ref{prop:up-down-path} states the equivalence
of up-down-path Definition~\ref{def:up-down-path-rooted} from \cite{Bordewich_2018},
and Definition~\ref{def:up-down-path} for rooted networks.
  Let $N^+$ be a rooted network, and let $p = u_0 \dots u_n$ in $N^+$
  satisfy Definition~\ref{def:up-down-path-rooted}.
  Then clearly there is no v-structure in $p$, that is, there is no segment
  $u_{i-1}u_iu_{i+1}$ such that $u_{i-1}u_i$ and $u_{i+1}u_i$ are both
  edges in $N^+$, and $p$ satisfies Definition~\ref{def:up-down-path}.
  Next, let $p = u_0 \dots u_n$ satisfy
  Definition~\ref{def:up-down-path}. There are two cases:
  \begin{itemize}
  \item If $u_0u_1$ is an edge in $N^+$, then since the direction cannot
    reverse during the path, we have that for all $i$, $u_iu_{i+1}$ is an
    edge in $N^+$, and $p$ satisfies Definition~\ref{def:up-down-path-rooted}.
  \item If $u_1u_0$ is an edge in $N^+$, then we can look for the smallest index
    $j$ such that $u_ju_{j+1}$ is an edge in $N^+$.  If there is none,
    then $p$ is a directed path from $u_n$ to $u_0$.
    If there is such $j$, then by the same argument as before, for all $i\geq j$,
    $u_iu_{i+1}$ is an edge in $N^+$.
    Either way, $p$ satisfies Definition~\ref{def:up-down-path-rooted}.
  \end{itemize}
\end{proof}

\subsection{Proof that the zipped-up network is unique}
\label{sec:zip-proofs}

In this section we prove that a metric semidirected network $N$ has a
unique zipped-up version $N^*$ where all the hybrid edges have length
$0$, that can be obtained from $N$ by a series of ``zipping
operations''.

First we shall restrict ourselves to networks with all hybrid nodes
having a single child edge (i.e.\ tree edge or outgoing hybrid
edge). For a network that does not satisfy the requirement, we can
apply step 1 of Definition~\ref{def:zip} (shown in Fig.~\ref{fig:hybrid-zipper})
and work with the resulting network instead.
By Proposition~\ref{prop:zipped-up},
such modifications do not change the distances between existing
nodes.

\begin{definition}
  Let $N$ be a metric semidirected network in which hybrid nodes have
  a single child edge.  Let $h$ be a hybrid node, $t$ be the length of
  its child edge and $a_i, i = 1, \dots, n$ be the lengths of
  its $n$ incoming hybrid edges.
  A \emph{zipping operation} at $h$ is a modification of
  $t, a_1, \dots, a_n$ such that
  $t$ and all $a_i$ remain non-negative, and
  $t+\sum_{i = 1}^n\gamma_i a_i$
  stays constant. Two networks are \emph{zipping-equivalent} if one
  can be obtained from the other through a series of zipping operations.

  \com{
  A \emph{hybrid tree} at hybrid node $h$ is a connected subgraph of
  $N$ that consists of hybrid nodes $u$ such that there is a directed
  path from $u$ to $h$ made of hybrid edges only, and the hybrid edges
  in the aforementioned paths.  We call $h$ the \emph{base} of the
  hybrid tree.
  }
  The \emph{hybrid funnel} $F(h)$ based at a hybrid node $h$ is the
  maximal connected subgraph that consists of directed paths
  into $h$ made of hybrid edges only,
  plus $h$'s child edge (see Fig.~\ref{fig:hybridfunnel}).
  The \emph{height} of a hybrid
  funnel is the number of edges on its longest path ending at $h$.
\end{definition}

\begin{figure}[h]
  \centering
  \includegraphics[scale=1.3]{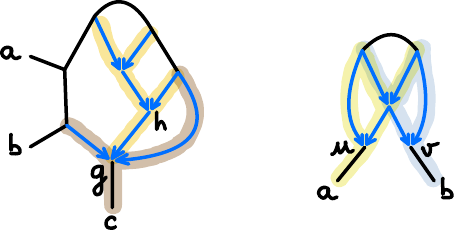}
  \caption{
  Left: the hybrid funnel $F(h)$, highlighted in orange, is not maximal.
  $F(g)$ is maximal. It contains $F(h)$ and all highlighted edges.
  Right: $F(u)$ and $F(v)$ are both maximal
  (highlighted in yellow and blue respectively). They are not disjoint, as may
  occur when a hybrid node has more than one child.
  Zipping operations are performed after refining the network
  so that all hybrid nodes have a unique child.
  }
  \label{fig:hybridfunnel}
\end{figure}

\begin{lemma}
  In a network in which hybrid nodes have
  a single child,
  distinct maximal hybrid funnels have disjoint edge sets.
\end{lemma}

\begin{proof}
  First, for a maximal hybrid funnel, the child of its base $h$
  must be a tree edge, because if it is a hybrid edge, we can get a
  larger hybrid funnel at its child node.
  Let $A$ and $B$ be different maximal hybrid funnels.
  Then their bases $u$ and $v$ (respectively) must be distinct: $u \neq v$.
  If $A$ and $B$ have the same tree edge,
  then $u$ and $v$ must be incident by
  their common child edge,
  and after rooting the network, $u$
  or $v$ would have three incoming edges but no outgoing edges, which is
  impossible.  Therefore, $A$ and $B$ cannot share a tree edge.

  Suppose $A, B$ share a hybrid edge $ab$.  Then there is a path from
  $b$ to $u$ consisting of hybrid edges only.  The same holds for $v$.
  Because each hybrid node has a unique child edge, one of the two
  paths must be contained in the other.  Consequently $u$ or $v$
  has a hybrid edge as a child edge, which is a contradiction.
\end{proof}

\begin{lemma}
  \label{lem:zip-invariant}
  Let $A$ be a maximal hybrid funnel in a metric semidirected
  network $N$.  Then zipping operations in $N$ do not change
  \begin{equation}\label{eq:proof-zip-unique}
    L(h) = t_h + \sum_{p \in P(h)} \gamma(p)\ell(p)
  \end{equation}
  where $h$ is the base of the funnel, $t_h$ is
  the length of its child edge and $P(h)$ is the set of maximal
  directed paths in $A$ that end at $h$.
\end{lemma}

\begin{proof}
  First, a zipping operation at a node outside of a \emph{maximal}
  hybrid funnel $F(h)$ does not modify any of the edges of the funnel,
  and so does not change $L(h)$.
  To finish the proof, we show the following claim:
  for any hybrid node $h$, zipping
  operations at nodes in the funnel $F(h)$ (which may not be
  maximal) do not change $L(h)$.

  To show this claim, we use induction on the height of $F(h)$.  When
  the height is $1$, the claim reduces to Proposition~\ref{prop:zipped-up}.
  Suppose the claim is true for heights up to $k$, and $F(h)$ is of
  height $k + 1$.  Let $u_1, \dots, u_n$ be the parent nodes of $h$.
  We may assume that $u_i$ is a hybrid node for $i\leq m$ and a tree node for $i>m$,
  for some $m\leq n$.
  Then $F(u_i), i = 1, \dots, m$
  are of height $\leq k$. Note that
  \[ L(h) = t_h + \sum_{i=1}^{m} \gamma_i L(u_i)
   + \sum_{i=m+1}^n\gamma_i t_i \]
  where $\gamma_i = \gamma(u_i h)$ and $t_i = \ell(u_i h)$.
  Let $u$ be a hybrid node in $F(h)$. If $u\neq h$, then $u$ must be
  in some $F(u_j)$, and a zipping operation at $u$ does not affect
  $L(h)$ because $t_h$ and all $L(u_i)$ stay unchanged (by induction).
  To consider a zipping operation at $u=h$, we rewrite
  \[ L(h) = t_h + \sum_{i=1}^n\gamma_it_i +
    \sum_{i=1}^{m} \gamma_i (L(u_i) - t_i).\]
  The
  last sum is a function of edges in the hybrid funnels at $u_1,\ldots,u_{m}$,
  and is hence unchanged by the zipping operation at $h$.  The term
  $t_h + \sum_{i=1}^n\gamma_it_i$ is constant by
  Proposition~\ref{prop:zipped-up}, which completes the proof.
\end{proof}

\begin{theorem}
  \label{thm:zip-unique}
  The zipped-up version of a network, as defined in Definition~\ref{def:zip},
  exists and is unique.
\end{theorem}

\begin{proof}
  Let $N$ be a metric semidirected network.  First we show that it is
  possible to obtain a zipped-up version of $N$.
  Consider a topological ordering of the nodes in $N$ based
  on some (arbitrary) rooting of $N$.
  This is an ordering from the root to the leaves,
  such that $u$ is listed before $v$ whenever there exists a directed edge $uv$.
  Perform zipping-up operations as in Definition~\ref{def:zip}
  according to this ordering (restricted to hybrid nodes).
  We claim that the resulting network $N^*$ is zipped-up.
  Indeed, by virtue of the topological ordering,
  zipping operations at either parent of a hybrid node $h$
  must be performed before zipping up $h$,
    so the length
  of $h$'s parent edges remain $0$ after being set to $0$,
  and $N^*$ is zipped-up.

  To show that the zipped-up version of $N$ is unique, we consider
  a zipped-up version $N'$ and show that $N'$ and $N^*$
  are identical. We consider two cases.
  If all hybrid nodes of $N$ have unique child edges, then $N'$, $N^*$
  and $N$ have the same topology 
  and the same lengths for edges outside of any funnel.
  $N'$ and $N^*$ have hybrid edges of length $0$.
  To show uniqueness, it suffices to show that they 
  have identical lengths for tree child edges below hybrid nodes.
  Let $h$ be a hybrid node with a tree child edge.
  Then the funnel $F(h)$ is maximal.
  By Lemma~\ref{lem:zip-invariant},
  $L(h)$ is identical in $N'$ and in $N^*$ (and in $N$),
  because $N'$ and $N^*$ are both zipping-equivalent to $N$.
  For $N'$ and $N^*$,
  all the paths have length $0$ in the last sum of \eqref{eq:proof-zip-unique},
  such that the length $t_h$ of the child edge must equal $L(h)$
  and therefore be identical in $N'$ and $N^*$; and $N'=N^*$.

  Now we consider the case when some hybrid nodes in $N$ have
  two or more children.
  For a network $M$, let $\widetilde{M}$ denote the network obtained
  by 
  performing step 1 of Def.~\ref{def:zip} below every hybrid
  node with multiple children, even those not zipped-up.
  \com{
  so that every
  hybrid node has a unique child in $\widetilde M$.
  Then $N'$ has an extra edge below a hybrid node $h$
  if and only if $h$ has multiple children in $N$
  and the length of child edge of $h$ in $\widetilde{N'}$ is
  positive. Indeed, the edges in $\widetilde{N'}$ that
  are not in $N'$ have length $0$, 
  and any new edge introduced in $N'$ is given a positive length.
  This is because zipping up is only performed at nodes that
  are not already zipped-up, and the new edge is a tree edge
  so its length can only increase when zipping up subsequent nodes.
  }
  Then a zipped-up version $N'$ of $N$ may have a different
  topology than $\widetilde{N}$.
  Step 1 of Def.~\ref{def:zip} was performed in $N'$ only at nodes that
  were not zipped-up, so the edges in $\widetilde{N}$ missing
  from $N'$ are below hybrid nodes that are zipped-up in $N'$.
  Conversely, the edges introduced in $N'$ during zipping-up have
  positive lengths from \eqref{eq:zip-up} in step 2.
  Now let $N_1'$ and $N_2'$ be two zipped-up versions of $N$.  Then
  $\widetilde{N_i'}$ is a zipped-up version of $\widetilde{N}$, for $i=1,2$
  (zipping up at the same series of hybrid nodes as to get $N_i'$ from
  $N$). By the previous case, $\widetilde{N_1'}=\widetilde{N_2'}$.
  Since the positivity of edges in $\widetilde{N_i'}$
  determines which tree edges have been introduced in $N_i'$,
  $N_1'$ and $N_2'$ have the same topology, and then the same metric.
\end{proof}

\subsection{Mapping the tree of blobs on the network}
\label{sec:blob-tree-proofs}

\begin{proof}[Proof of Proposition~\ref{prop:net-blobtree-same-cutedges-tips}]
  Let $T = \bt(N)$.  For a node $u$ in $N$, write $B(u)$ for the blob
  that contains $u$.
  For the first bijection: let $g$ be the map from cut edges of $N$ to
  the edges of $T$ such that $g(uv)$ is the edge $(B(u), B(v))$ in $T$.
  It is injective because if two edges $e_1$ and $e_2$ get mapped to
  the same edge $(B(u), B(v))$ in $T$, then one could find a cycle in
  $U(N)$ that contains both $e_1$ and $e_2$, contradicting that they
  are cut edges.  It is surjective because for an edge $e$ that
  connects two $2$-edge-connected components $U$ and $V$, the removal
  of $e$ must disconnect $U$ and $V$: Otherwise $U$ and $V$ would be
  the same $2$-edge-connected component.  Therefore $e$ must be a cut
  edge, and we have $g(e) = (U, V)$.

  For the second bijection: recall that we require all tips of a rooted
  network to have in-degree $1$, so the same holds for the tips of a
  semidirected network $N$.
  The blob that contains a tip $x$ is therefore a trivial blob, $\{x\}$.
  Let $f: V_L(N) \to V(T)$ be the
  map such that $f(x)$ is the trivial blob $\{x\}$.  Clearly
  $f$ is injective.  We only need to show that these trivial blobs are
  all the leaves of $T$, i.e.\ we did not introduce new leaves in $T$
  that do not correspond to tips of $N$.
  Suppose that there is a leaf blob $B$ in $T$ that is
  not of the form $\{x\}$ for some tip $x$ in $N$.
  $B$ may not contain any leaf $x$, because $\{x\}$ is a blob itself,
  so we would have $B=\{x\}$, a contradiction.
  Let $N^+$ be a rooted LSA network that induces $N$,
  obtained from rerooting $N$ at a node.
  Let $e = (u, v)$ be the cut edge incident to $B$ in $N^+$.
  Then we must have that $u\in B$, because otherwise $v\in B$ and
  $B$ would contain any descendant leaf of $e$ (found by following
  any directed path starting at $e$
  until we reach a node of out-degree 0).  Consequently, the root
  $\rho(N^+)$ must be in $B$ (a leaf in $T$),
  for otherwise the edge $e$ would be
  directed the other way.
  Now any path from the root to a leaf must go through
  $(u,v)$. Therefore $v$ lies on every path from the root to any leaf.
  Since $v\neq \rho(N^+),$ $N^+$ is not an LSA network, a
  contradiction.
\end{proof}

\subsection{Proof that mixed representations preserve distances}
\label{sec:mixedrep-proofs}

To prove Theorem~\ref{thm:mixedrepdistances}, we extend
Definition~\ref{def:displayedtree}.

\begin{definition}[displayed split network]
  Let $N^*$ be the mixed network representation of a
  level-$1$ semidirected network $N$.
  For hybrid node $h\in N^*$, let $E_H(h)$ be its parent hybrid edges.
  Let $G$ be the graph obtained by keeping one hybrid edge
  $e\in E_H(h)$ and deleting the remaining edge(s) in $E_H(h)$,
  for each hybrid node $h\in N^*$.
  Then $G$ is a split network \cite[p.240]{steel16_phylogeny} 
  and is called a \emph{displayed split network}.
  The distribution on displayed split networks generated
  by $N^*$ is the distribution obtained by keeping $e\in E_H(h)$
  with probability $\gamma(e)$, independently across $h$.
\end{definition}
\noindent
Note that $G$ is a split network because it has no directed (hybrid) edges;
its topology is of level 1;
and its blobs are degree-4 cycles, each with 2 pairs of split edges.

\begin{proposition}
\label{prop:mixed-udp}
  Let $N^*$ be the mixed network representation of a level-$1$
  semidirected network $N$.
  For two nodes $u, v$ and split network $G$, let $Q_{uv}(G)$
  be the set of paths between $u$ and $v$ in $G$ that have shortest length.
  Then
  each equivalence class of mixed up-down paths between
  $u$ and $v$ in $N^*$ is equal to $Q_{uv}(G)$ for some
  split network $G$ displayed in $N^*$.
  Furthermore,
  for a given mixed up-down path $p$ between $u$ and $v$ in $N^*$,
  \[\PP(p \in Q_{uv}(G)) = \gamma(p)\]
  where $G$ is a random split network displayed in $N^*$.
  Consequently,
  \begin{equation}\label{eq:mixedrep_averagedistance}
  d_{N^*}(u, v) = \EE\ell(Q_{uv}(G))
  \end{equation}
  where the expectation is taken over a
  random displayed split network $G$ in $N^*$.
\end{proposition}

\begin{proof}
  For the first claim, note that any shortest path between $u$ and $v$ in
  $G$ is a mixed up-down path.
  \com{
  : for the definition of mixed up-down
  paths, condition 1 is satisfied because of in $M$ for each hybrid
  node only one incoming hybrid edge is kept; condition 2 is satisfied
  because any segment of the shortest path is a shortest path as well.}
  We only need to show that all shortest paths between $u$ and $v$ in $G$
  are equivalent.
  Let $q$ and $q'$ be two of them.
  Note that any tree edge or hybrid edge from $N^*$ that was
  retained in $G$ is a cut edge in $G$.  Since $N$ is of level 1,
  each blob of $G$ corresponds to a split blob (4-cycle) from $N^*$.
  Since the unique path from $u$ to $v$ in the tree of blobs of
  $G$ contains the paths in $Q_{uv}(G)$, we get that
  $q$ and $q'$ must pass through the same
  tree edges, hybrid edges and split blobs; and in the same order.
  They may differ in the edges that they contain from
  each split blob. Any such difference corresponds to replacing
  one split segment by an equivalent split segment (of shortest length).
  Therefore $q$ and $q'$ are equivalent.

  For the second claim, let $p$ be a mixed up-down path from
  $u$ to $v$ in $N^*$.
  It suffices to show that $p \in Q_{uv}(G)$ if
  and only if all the hybrid edges present in $p$ are kept in $G$.
  The ``only if'' part is trivial.
  For the ``if'' part, assume that the hybrid edges from $p$ are
  retained in $G$, such that $p$ is in $G$. We want to
  show that it is of shortest length.
  Let $q \in Q_{uv}(G)$. Both $p$ and $q$ go from $u$ to $v$,
  so by the same argument as above, $p$ and $q$ must pass through
  the same tree edges, hybrid edges, and split blobs;
  and in the same order. Since $p$ is a mixed up-down path,
  its split segments must be of shortest length. Therefore
  $p$ must traverse each split blob through a split segment of length
  no larger than that of $q$.
  Consequently $p\in Q_{uv}(G)$.
\end{proof}

\begin{proof}[Proof of Theorem~\ref{thm:mixedrepdistances}]
  By induction, it suffices to show that in a mixed network,
  replacing a single $4$-cycle by its corresponding split cycle does not
  change the average distances.
  Let $M$ be a mixed network on taxon set $X$ with one or more
  semidirected $4$-cycles, and let $M'$ be the mixed network obtained from
  $M$ by replacing one
  $4$-sunlet subgraph $C$ in $M$ by the corresponding split subgraph $C'$.
  Since $U(M)=U(M')$, $M$ and $M'$ have the same tree of blobs $T$
  and $C$ and $C'$ correspond to the same node $b$ in $T$.
  Let $x, y \in X$ be two tips of $M$ and $p$ be the path from $x$ to $y$ in $T$.

  If $p$ does not go through $b$, then
  the sets of mixed up-down paths between $x, y$ are identical in $M$ and in $M'$
  (they do not intersect $C$ or $C'$ respectively), therefore $d_M(x,y) = d_{M'}(x,y)$.

  If instead $p$ goes through $b$,
  then all mixed up-down paths from $x$ to $y$ in $M$ (resp. $M'$) intersect
  $C$ (resp. $C'$), and they must all go through the same cut edges
  $e_1$ and $e_2$ adjacent to $C$ (resp. $C'$).
  Let $u$ and $v$ be the nodes in $C$ adjacent to $e_1$ and $e_2$.
  We can identify $u$ and $v$ with nodes in $C'$ if we omit
  step 2 in Def.~\ref{def:mixedrep}.
  We can do so without loss of generality because the suppression of degree-2 nodes
  does not affect distances.
  Then, all mixed up-down paths in $M$ (resp. $M'$) from $x$ to $y$
  go from $x$ to $u$ along edges that do not belong in $C$ (resp. $C'$),
  then from $u$ to $v$ within $C$ (resp. $C'$), and then from $v$ to $y$
  through edges not in $C$ (resp. $C'$).
  The same applies to each split network displayed in $M$ (resp. $M'$).
  Since, in addition, $d_G$ is defined as the length of the shortest path
  on a split network, the following holds when the graph $G$ is
  any split network displayed in $M$ or $M'$:
  \begin{equation}\label{eq:proofmixedrepdistances}
  d_G(x, y) = d_G(x, u) + d_G(u, v) + d_G(v, y)\,.
  \end{equation}
  By \eqref{eq:mixedrep_averagedistance} in
  Proposition~\ref{prop:mixed-udp}, \eqref{eq:proofmixedrepdistances}
  also holds when $G=M$ and $G=M'$.
  Because $M'$ differs from $M$ only in $C$, which is replaced by
  $C'$, we have $d_M(x, u) = d_{M'}(x, u)$ and $d_M(v, y) = d_{M'}(v, y)$.
  We also have $d_M(u, v) = d_{M'}(u, v)$ because $C'$ is the mixed
  representation of the $4$-sunlet $C$.
  Therefore $d_M(x, y) = d_{M'}(x, y)$.
\end{proof}

\end{appendices}

\bibliography{lib}


\begin{thebibliography}{44}
\ifx \bisbn   \undefined \def \bisbn  #1{ISBN #1}\fi
\ifx \binits  \undefined \def \binits#1{#1}\fi
\ifx \bauthor  \undefined \def \bauthor#1{#1}\fi
\ifx \batitle  \undefined \def \batitle#1{#1}\fi
\ifx \bjtitle  \undefined \def \bjtitle#1{#1}\fi
\ifx \bvolume  \undefined \def \bvolume#1{\textbf{#1}}\fi
\ifx \byear  \undefined \def \byear#1{#1}\fi
\ifx \bissue  \undefined \def \bissue#1{#1}\fi
\ifx \bfpage  \undefined \def \bfpage#1{#1}\fi
\ifx \blpage  \undefined \def \blpage #1{#1}\fi
\ifx \burl  \undefined \def \burl#1{\textsf{#1}}\fi
\ifx \doiurl  \undefined \def \doiurl#1{\url{https://doi.org/#1}}\fi
\ifx \betal  \undefined \def \betal{\textit{et al.}}\fi
\ifx \binstitute  \undefined \def \binstitute#1{#1}\fi
\ifx \binstitutionaled  \undefined \def \binstitutionaled#1{#1}\fi
\ifx \bctitle  \undefined \def \bctitle#1{#1}\fi
\ifx \beditor  \undefined \def \beditor#1{#1}\fi
\ifx \bpublisher  \undefined \def \bpublisher#1{#1}\fi
\ifx \bbtitle  \undefined \def \bbtitle#1{#1}\fi
\ifx \bedition  \undefined \def \bedition#1{#1}\fi
\ifx \bseriesno  \undefined \def \bseriesno#1{#1}\fi
\ifx \blocation  \undefined \def \blocation#1{#1}\fi
\ifx \bsertitle  \undefined \def \bsertitle#1{#1}\fi
\ifx \bsnm \undefined \def \bsnm#1{#1}\fi
\ifx \bsuffix \undefined \def \bsuffix#1{#1}\fi
\ifx \bparticle \undefined \def \bparticle#1{#1}\fi
\ifx \barticle \undefined \def \barticle#1{#1}\fi
\bibcommenthead
\ifx \bconfdate \undefined \def \bconfdate #1{#1}\fi
\ifx \botherref \undefined \def \botherref #1{#1}\fi
\ifx \url \undefined \def \url#1{\textsf{#1}}\fi
\ifx \bchapter \undefined \def \bchapter#1{#1}\fi
\ifx \bbook \undefined \def \bbook#1{#1}\fi
\ifx \bcomment \undefined \def \bcomment#1{#1}\fi
\ifx \oauthor \undefined \def \oauthor#1{#1}\fi
\ifx \citeauthoryear \undefined \def \citeauthoryear#1{#1}\fi
\ifx \endbibitem  \undefined \def \endbibitem {}\fi
\ifx \bconflocation  \undefined \def \bconflocation#1{#1}\fi
\ifx \arxivurl  \undefined \def \arxivurl#1{\textsf{#1}}\fi
\csname PreBibitemsHook\endcsname

\bibitem{2018folk-hybridization}
\begin{barticle}
\bauthor{\bsnm{Folk}, \binits{R.A.}},
\bauthor{\bsnm{Soltis}, \binits{P.S.}},
\bauthor{\bsnm{Soltis}, \binits{D.E.}},
\bauthor{\bsnm{Guralnick}, \binits{R.}}:
\batitle{New prospects in the detection and comparative analysis of
  hybridization in the tree of life}.
\bjtitle{American Journal of Botany}
\bvolume{105}(\bissue{3}),
\bfpage{364}--\blpage{375}
(\byear{2018}).
\doiurl{10.1002/ajb2.1018}
\end{barticle}
\endbibitem

\bibitem{2020blair-review}
\begin{barticle}
\bauthor{\bsnm{Blair}, \binits{C.}},
\bauthor{\bsnm{An\'e}, \binits{C.}}:
\batitle{Phylogenetic trees and networks can serve as powerful and
  complementary approaches for analysis of genomic data}.
\bjtitle{Systematic Biology}
\bvolume{69}(\bissue{3}),
\bfpage{593}--\blpage{601}
(\byear{2020}).
\doiurl{10.1093/sysbio/syz056}
\end{barticle}
\endbibitem

\bibitem{Sol_s_Lemus_2016_infer}
\begin{barticle}
\bauthor{\bsnm{Sol\'is-Lemus}, \binits{C.}},
\bauthor{\bsnm{An\'e}, \binits{C.}}:
\batitle{Inferring phylogenetic networks with maximum pseudolikelihood under
  incomplete lineage sorting}.
\bjtitle{PLOS Genetics}
\bvolume{12}(\bissue{3}),
\bfpage{1005896}
(\byear{2016}).
\doiurl{10.1371/journal.pgen.1005896}
\end{barticle}
\endbibitem

\bibitem{cao2019-phylonet-practical}
\begin{barticle}
\bauthor{\bsnm{Cao}, \binits{Z.}},
\bauthor{\bsnm{Liu}, \binits{X.}},
\bauthor{\bsnm{Ogilvie}, \binits{H.A.}},
\bauthor{\bsnm{Yan}, \binits{Z.}},
\bauthor{\bsnm{Nakhleh}, \binits{L.}}:
\batitle{Practical aspects of phylogenetic network analysis using phylonet}.
\bjtitle{bioRxiv}
(\byear{2019}).
\doiurl{10.1101/746362}
\end{barticle}
\endbibitem

\bibitem{rabier2021-snappnet}
\begin{barticle}
\bauthor{\bsnm{Rabier}, \binits{C.-E.}},
\bauthor{\bsnm{Berry}, \binits{V.}},
\bauthor{\bsnm{Stoltz}, \binits{M.}},
\bauthor{\bsnm{Santos}, \binits{J.D.}},
\bauthor{\bsnm{Wang}, \binits{W.}},
\bauthor{\bsnm{Glaszmann}, \binits{J.-C.}},
\bauthor{\bsnm{Pardi}, \binits{F.}},
\bauthor{\bsnm{Scornavacca}, \binits{C.}}:
\batitle{On the inference of complex phylogenetic networks by {M}arkov chain
  {M}onte-{C}arlo}.
\bjtitle{PLOS Computational Biology}
\bvolume{17},
\bfpage{1}--\blpage{39}
(\byear{2021}).
\doiurl{10.1371/journal.pcbi.1008380}
\end{barticle}
\endbibitem

\bibitem{2004bryant-neighbonet}
\begin{barticle}
\bauthor{\bsnm{Bryant}, \binits{D.}},
\bauthor{\bsnm{Moulton}, \binits{V.}}:
\batitle{Neighbor-net: An agglomerative method for the construction of
  phylogenetic networks}.
\bjtitle{Molecular Biology and Evolution}
\bvolume{21}(\bissue{2}),
\bfpage{255}--\blpage{265}
(\byear{2004}).
\doiurl{10.1093/molbev/msh018}
\end{barticle}
\endbibitem

\bibitem{1987satounei-NJ}
\begin{barticle}
\bauthor{\bsnm{Saitou}, \binits{N.}},
\bauthor{\bsnm{Nei}, \binits{M.}}:
\batitle{The neighbor-joining method: a new method for reconstructing
  phylogenetic trees}.
\bjtitle{Molecular Biology and Evolution}
\bvolume{4}(\bissue{4}),
\bfpage{406}--\blpage{425}
(\byear{1987}).
\doiurl{10.1093/oxfordjournals.molbev.a040454}
\end{barticle}
\endbibitem

\bibitem{2004despergascuel-fastme}
\begin{barticle}
\bauthor{\bsnm{Desper}, \binits{R.}},
\bauthor{\bsnm{Gascuel}, \binits{O.}}:
\batitle{Theoretical foundation of the balanced minimum evolution method of
  phylogenetic inference and its relationship to weighted least-squares tree
  fitting}.
\bjtitle{Molecular Biology and Evolution}
\bvolume{21}(\bissue{3}),
\bfpage{587}--\blpage{598}
(\byear{2004}).
\doiurl{10.1093/molbev/msh049}
\end{barticle}
\endbibitem

\bibitem{2017rusinko}
\begin{barticle}
\bauthor{\bsnm{Rusinko}, \binits{J.}},
\bauthor{\bsnm{McPartlon}, \binits{M.}}:
\batitle{Species tree estimation using neighbor joining}.
\bjtitle{Journal of Theoretical Biology}
\bvolume{414},
\bfpage{5}--\blpage{7}
(\byear{2017}).
\doiurl{10.1016/j.jtbi.2016.11.005}
\end{barticle}
\endbibitem

\bibitem{charles03_phylog}
\begin{bbook}
\bauthor{\bsnm{Semple}, \binits{C.}},
\bauthor{\bsnm{Steel}, \binits{M.}}:
\bbtitle{Phylogenetics}.
\bsertitle{Oxford lecture series in mathematics and its applications 24}.
\bpublisher{Oxford University Press},
\blocation{Oxford}
(\byear{2003})
\end{bbook}
\endbibitem

\bibitem{Bordewich_2018}
\begin{barticle}
\bauthor{\bsnm{Bordewich}, \binits{M.}},
\bauthor{\bsnm{Huber}, \binits{K.T.}},
\bauthor{\bsnm{Moulton}, \binits{V.}},
\bauthor{\bsnm{Semple}, \binits{C.}}:
\batitle{Recovering normal networks from shortest inter-taxa distance
  information}.
\bjtitle{Journal of Mathematical Biology}
\bvolume{77}(\bissue{3}),
\bfpage{571}--\blpage{594}
(\byear{2018}).
\doiurl{10.1007/s00285-018-1218-x}
\end{barticle}
\endbibitem

\bibitem{2017chang-shortestdistance-ultrametric}
\begin{bchapter}
\bauthor{\bsnm{Chang}, \binits{K.-Y.}},
\bauthor{\bsnm{Cui}, \binits{Y.}},
\bauthor{\bsnm{Yiu}, \binits{S.-M.}},
\bauthor{\bsnm{Hon}, \binits{W.-K.}}:
\bctitle{Reconstructing one-articulated networks with distance matrices}.
In: \beditor{\bsnm{Cai}, \binits{Z.}},
\beditor{\bsnm{Daescu}, \binits{O.}},
\beditor{\bsnm{Li}, \binits{M.}} (eds.)
\bbtitle{Bioinformatics Research and Applications},
pp. \bfpage{34}--\blpage{45}.
\bpublisher{Springer},
\blocation{Cham}
(\byear{2017})
\end{bchapter}
\endbibitem

\bibitem{2016bordewichsemple-intertaxadistances}
\begin{barticle}
\bauthor{\bsnm{Bordewich}, \binits{M.}},
\bauthor{\bsnm{Semple}, \binits{C.}}:
\batitle{Determining phylogenetic networks from inter-taxa distances}.
\bjtitle{Journal of Mathematical Biology}
\bvolume{73}(\bissue{2}),
\bfpage{283}--\blpage{303}
(\byear{2016}).
\doiurl{10.1007/s00285-015-0950-8}
\end{barticle}
\endbibitem

\bibitem{bordewich16_algor_recon_ultram_tree_child}
\begin{barticle}
\bauthor{\bsnm{Bordewich}, \binits{M.}},
\bauthor{\bsnm{Tokac}, \binits{N.}}:
\batitle{An algorithm for reconstructing ultrametric tree-child networks from
  inter-taxa distances}.
\bjtitle{Discrete Applied Mathematics}
\bvolume{213},
\bfpage{47}--\blpage{59}
(\byear{2016}).
\doiurl{10.1016/j.dam.2016.05.011}
\end{barticle}
\endbibitem

\bibitem{2018bordewich-treechild-multisetdistances}
\begin{barticle}
\bauthor{\bsnm{Bordewich}, \binits{M.}},
\bauthor{\bsnm{Semple}, \binits{C.}},
\bauthor{\bsnm{Tokac}, \binits{N.}}:
\batitle{Constructing tree-child networks from distance matrices}.
\bjtitle{Algorithmica}
\bvolume{80}(\bissue{8}),
\bfpage{2240}--\blpage{2259}
(\byear{2018}).
\doiurl{10.1007/s00453-017-0320-6}
\end{barticle}
\endbibitem

\bibitem{allman2022-identifiability-logdet-net}
\begin{barticle}
\bauthor{\bsnm{Allman}, \binits{E.S.}},
\bauthor{\bsnm{Ba\~{n}os}, \binits{H.}},
\bauthor{\bsnm{Rhodes}, \binits{J.A.}}:
\batitle{Identifiability of species network topologies from genomic sequences
  using the {logDet} distance}.
\bjtitle{Journal of Mathematical Biology}
\bvolume{84}(\bissue{5}),
\bfpage{35}
(\byear{2022}).
\doiurl{10.1007/s00285-022-01734-2}
\end{barticle}
\endbibitem

\bibitem{willems14_new_effic_algor_infer_explic}
\begin{barticle}
\bauthor{\bsnm{Willems}, \binits{M.}},
\bauthor{\bsnm{Tahiri}, \binits{N.}},
\bauthor{\bsnm{Makarenkov}, \binits{V.}}:
\batitle{A new efficient algorithm for inferring explicit hybridization
  networks following the neighbor-joining principle}.
\bjtitle{Journal of Bioinformatics and Computational Biology}
\bvolume{12}(\bissue{05}),
\bfpage{1450024}
(\byear{2014}).
\doiurl{10.1142/s0219720014500243}
\end{barticle}
\endbibitem

\bibitem{willson12_tree_averag_distan_certain_phylog}
\begin{barticle}
\bauthor{\bsnm{Willson}, \binits{S.J.}}:
\batitle{Tree-average distances on certain phylogenetic networks have their
  weights uniquely determined}.
\bjtitle{Algorithms for Molecular Biology}
\bvolume{7}(\bissue{1}),
\bfpage{13}
(\byear{2012}).
\doiurl{10.1186/1748-7188-7-13}
\end{barticle}
\endbibitem

\bibitem{willson13_recon_certain_phylog_networ_from}
\begin{barticle}
\bauthor{\bsnm{Willson}, \binits{S.J.}}:
\batitle{Reconstruction of certain phylogenetic networks from their
  tree-average distances}.
\bjtitle{Bulletin of Mathematical Biology}
\bvolume{75}(\bissue{10}),
\bfpage{1840}--\blpage{1878}
(\byear{2013}).
\doiurl{10.1007/s11538-013-9872-z}
\end{barticle}
\endbibitem

\bibitem{francis-steel-2015}
\begin{barticle}
\bauthor{\bsnm{Francis}, \binits{A.R.}},
\bauthor{\bsnm{Steel}, \binits{M.}}:
\batitle{Tree-like reticulation networks--when do tree-like distances also
  support reticulate evolution?}
\bjtitle{Mathematical Biosciences}
\bvolume{259},
\bfpage{12}--\blpage{19}
(\byear{2015}).
\doiurl{10.1016/j.mbs.2014.10.008}
\end{barticle}
\endbibitem

\bibitem{2005chan-ultrametric-galled-distancematrix}
\begin{bchapter}
\bauthor{\bsnm{Chan}, \binits{H.-L.}},
\bauthor{\bsnm{Jansson}, \binits{J.}},
\bauthor{\bsnm{Lam}, \binits{T.-W.}},
\bauthor{\bsnm{Yiu}, \binits{S.-M.}}:
\bctitle{Reconstructing an ultrametric galled phylogenetic network from a
  distance matrix}.
In: \beditor{\bsnm{Jedrzejowicz}, \binits{J.}},
\beditor{\bsnm{Szepietowski}, \binits{A.}} (eds.)
\bbtitle{Mathematical Foundations of Computer Science 2005},
pp. \bfpage{224}--\blpage{235}.
\bpublisher{Springer},
\blocation{Berlin, Heidelberg}
(\byear{2005})
\end{bchapter}
\endbibitem

\bibitem{Pardi_2015}
\begin{barticle}
\bauthor{\bsnm{Pardi}, \binits{F.}},
\bauthor{\bsnm{Scornavacca}, \binits{C.}}:
\batitle{Reconstructible phylogenetic networks: Do not distinguish the
  indistinguishable}.
\bjtitle{PLOS Computational Biology}
\bvolume{11}(\bissue{4}),
\bfpage{1004135}
(\byear{2015}).
\doiurl{10.1371/journal.pcbi.1004135}
\end{barticle}
\endbibitem

\bibitem{2019banos}
\begin{barticle}
\bauthor{\bsnm{Ba\~{n}os}, \binits{H.}}:
\batitle{Identifying species network features from gene tree quartets under the
  coalescent model}.
\bjtitle{Bulletin of Mathematical Biology}
\bvolume{81}(\bissue{2}),
\bfpage{494}--\blpage{534}
(\byear{2019}).
\doiurl{10.1007/s11538-018-0485-4}
\end{barticle}
\endbibitem

\bibitem{gross20_distin_level_phylog_networ_basis}
\begin{botherref}
\oauthor{\bsnm{Gross}, \binits{E.}},
\oauthor{\bsnm{Iersel}, \binits{L.v.}},
\oauthor{\bsnm{Janssen}, \binits{R.}},
\oauthor{\bsnm{Jones}, \binits{M.}},
\oauthor{\bsnm{Long}, \binits{C.}},
\oauthor{\bsnm{Murakami}, \binits{Y.}}:
Distinguishing level-1 phylogenetic networks on the basis of data generated by
  markov processes.
CoRR
(2020)
{\href{https://arxiv.org/abs/2007.08782}{{arXiv:2007.08782}}}
{[q-bio.PE]}
\end{botherref}
\endbibitem

\bibitem{Allman_2019}
\begin{botherref}
\oauthor{\bsnm{Allman}, \binits{E.S.}},
\oauthor{\bsnm{Ba\~{n}os}, \binits{H.}},
\oauthor{\bsnm{Rhodes}, \binits{J.A.}}:
{NANUQ}: a method for inferring species networks from gene trees under the
  coalescent model.
Algorithms for Molecular Biology
\textbf{14}(1)
(2019).
\doiurl{10.1186/s13015-019-0159-2}
\end{botherref}
\endbibitem

\bibitem{steel16_phylogeny}
\begin{bbook}
\bauthor{\bsnm{Steel}, \binits{M.}}:
\bbtitle{Phylogeny: Discrete and Random Processes in Evolution},
p. \bfpage{302}.
\bpublisher{Society for Industrial and Applied Mathematics},
\blocation{Philadelphia, PA}
(\byear{2016}).
\doiurl{10.1137/1.9781611974485}
\end{bbook}
\endbibitem

\bibitem{francis-moulton-2018}
\begin{barticle}
\bauthor{\bsnm{Francis}, \binits{A.}},
\bauthor{\bsnm{Moulton}, \binits{V.}}:
\batitle{Identifiability of tree-child phylogenetic networks under a
  probabilistic recombination-mutation model of evolution}.
\bjtitle{Journal of Theoretical Biology}
\bvolume{446},
\bfpage{160}--\blpage{167}
(\byear{2018}).
\doiurl{10.1016/j.jtbi.2018.03.011}
\end{barticle}
\endbibitem

\bibitem{fischer18_unroot_non_binar_tree_based_phylog_networ}
\begin{barticle}
\bauthor{\bsnm{Fischer}, \binits{M.}},
\bauthor{\bsnm{Herbst}, \binits{L.}},
\bauthor{\bsnm{Galla}, \binits{M.}},
\bauthor{\bsnm{Long}, \binits{Y.}},
\bauthor{\bsnm{Wicke}, \binits{K.}}:
\batitle{Unrooted non-binary tree-based phylogenetic networks}.
\bjtitle{Discrete Applied Mathematics}
\bvolume{294},
\bfpage{10}--\blpage{30}
(\byear{2021}).
\doiurl{10.1016/j.dam.2021.01.005}
\end{barticle}
\endbibitem

\bibitem{huber14_how_much_infor_is_needed}
\begin{barticle}
\bauthor{\bsnm{Huber}, \binits{K.T.}},
\bauthor{\bsnm{Iersel}, \binits{L.V.}},
\bauthor{\bsnm{Moulton}, \binits{V.}},
\bauthor{\bsnm{Wu}, \binits{T.}}:
\batitle{How much information is needed to infer reticulate evolutionary
  histories?}
\bjtitle{Systematic Biology}
\bvolume{64}(\bissue{1}),
\bfpage{102}--\blpage{111}
(\byear{2014}).
\doiurl{10.1093/sysbio/syu076}
\end{barticle}
\endbibitem

\bibitem{gambette2012_quartets_unrooted}
\begin{barticle}
\bauthor{\bsnm{Gambette}, \binits{P.}},
\bauthor{\bsnm{Berry}, \binits{V.}},
\bauthor{\bsnm{Paul}, \binits{C.}}:
\batitle{Quartets and unrooted phylogenetic networks}.
\bjtitle{Journal of Bioinformatics and Computational Biology}
\bvolume{10}(\bissue{04}),
\bfpage{1250004}
(\byear{2012}).
\doiurl{10.1142/S0219720012500047}
\end{barticle}
\endbibitem

\bibitem{gusfield07_decom_theor_phylog_networ_incom_charac}
\begin{barticle}
\bauthor{\bsnm{Gusfield}, \binits{D.}},
\bauthor{\bsnm{Bansal}, \binits{V.}},
\bauthor{\bsnm{Bafna}, \binits{V.}},
\bauthor{\bsnm{Song}, \binits{Y.S.}}:
\batitle{A decomposition theory for phylogenetic networks and incompatible
  characters}.
\bjtitle{Journal of Computational Biology}
\bvolume{14}(\bissue{10}),
\bfpage{1247}--\blpage{1272}
(\byear{2007}).
\doiurl{10.1089/cmb.2006.0137}
\end{barticle}
\endbibitem

\bibitem{murakami19_recon_tree_child_networ_from}
\begin{barticle}
\bauthor{\bsnm{Murakami}, \binits{Y.}},
\bauthor{\bparticle{van} \bsnm{Iersel}, \binits{L.}},
\bauthor{\bsnm{Janssen}, \binits{R.}},
\bauthor{\bsnm{Jones}, \binits{M.}},
\bauthor{\bsnm{Moulton}, \binits{V.}}:
\batitle{Reconstructing tree-child networks from reticulate-edge-deleted
  subnetworks}.
\bjtitle{Bulletin of Mathematical Biology}
\bvolume{81}(\bissue{10}),
\bfpage{3823}--\blpage{3863}
(\byear{2019}).
\doiurl{10.1007/s11538-019-00641-w}
\end{barticle}
\endbibitem

\bibitem{diestel-graphtheory}
\begin{bbook}
\bauthor{\bsnm{Diestel}, \binits{R.}}:
\bbtitle{Graph Theory},
\bedition{5th edition} edn.
\bsertitle{Graduate Texts in Mathematics},
vol. \bseriesno{173},
p. \bfpage{447}.
\bpublisher{Springer},
\blocation{Heidelberg}
(\byear{2017}).
\doiurl{10.1007/978-3-662-53622-3}
\end{bbook}
\endbibitem

\bibitem{degnan2018}
\begin{barticle}
\bauthor{\bsnm{Degnan}, \binits{J.H.}}:
\batitle{Modeling hybridization under the network multispecies coalescent}.
\bjtitle{Systematic Biology}
\bvolume{67}(\bissue{5}),
\bfpage{786}--\blpage{799}
(\byear{2018}).
\doiurl{10.1093/sysbio/syy040}
\end{barticle}
\endbibitem

\bibitem{harary1971graph}
\begin{bbook}
\bauthor{\bsnm{Harary}, \binits{F.}}:
\bbtitle{Graph Theory}.
\bsertitle{Addison Wesley series in mathematics}.
\bpublisher{Addison-Wesley},
\blocation{Reading, MA}
(\byear{1971})
\end{bbook}
\endbibitem

\bibitem{gross18_distin_phylog_networ}
\begin{barticle}
\bauthor{\bsnm{Gross}, \binits{E.}},
\bauthor{\bsnm{Long}, \binits{C.}}:
\batitle{Distinguishing phylogenetic networks}.
\bjtitle{SIAM Journal on Applied Algebra and Geometry}
\bvolume{2}(\bissue{1}),
\bfpage{72}--\blpage{93}
(\byear{2018}).
\doiurl{10.1137/17m1134238}
\end{barticle}
\endbibitem

\bibitem{liu2009_STAR_STEAC}
\begin{barticle}
\bauthor{\bsnm{Liu}, \binits{L.}},
\bauthor{\bsnm{Yu}, \binits{L.}},
\bauthor{\bsnm{Pearl}, \binits{D.K.}},
\bauthor{\bsnm{Edwards}, \binits{S.V.}}:
\batitle{Estimating species phylogenies using coalescence times among
  sequences}.
\bjtitle{Systematic Biology}
\bvolume{58}(\bissue{5}),
\bfpage{468}--\blpage{477}
(\byear{2009}).
\doiurl{10.1093/sysbio/syp031}
\end{barticle}
\endbibitem

\bibitem{liu2011_NJst}
\begin{barticle}
\bauthor{\bsnm{Liu}, \binits{L.}},
\bauthor{\bsnm{Yu}, \binits{L.}}:
\batitle{Estimating species trees from unrooted gene trees}.
\bjtitle{Systematic Biology}
\bvolume{60}(\bissue{5}),
\bfpage{661}--\blpage{667}
(\byear{2011}).
\doiurl{10.1093/sysbio/syr027}
\end{barticle}
\endbibitem

\bibitem{peter2016}
\begin{barticle}
\bauthor{\bsnm{Peter}, \binits{B.M.}}:
\batitle{Admixture, population structure, and {F}-statistics}.
\bjtitle{Genetics}
\bvolume{202}(\bissue{4}),
\bfpage{1485}--\blpage{1501}
(\byear{2016}).
\doiurl{10.1534/genetics.115.183913}
\end{barticle}
\endbibitem

\bibitem{conover2019_malvaceae}
\begin{barticle}
\bauthor{\bsnm{Conover}, \binits{J.L.}},
\bauthor{\bsnm{Karimi}, \binits{N.}},
\bauthor{\bsnm{Stenz}, \binits{N.}},
\bauthor{\bsnm{An\'e}, \binits{C.}},
\bauthor{\bsnm{Grover}, \binits{C.E.}},
\bauthor{\bsnm{Skema}, \binits{C.}},
\bauthor{\bsnm{Tate}, \binits{J.A.}},
\bauthor{\bsnm{Wolff}, \binits{K.}},
\bauthor{\bsnm{Logan}, \binits{S.A.}},
\bauthor{\bsnm{Wendel}, \binits{J.F.}},
\bauthor{\bsnm{Baum}, \binits{D.A.}}:
\batitle{A {Malvaceae} mystery: A mallow maelstrom of genome multiplications
  and maybe misleading methods?}
\bjtitle{Journal of Integrative Plant Biology}
\bvolume{61}(\bissue{1}),
\bfpage{12}--\blpage{31}
(\byear{2019}).
\doiurl{10.1111/jipb.12746}
\end{barticle}
\endbibitem

\bibitem{karimi2020_baobab}
\begin{barticle}
\bauthor{\bsnm{Karimi}, \binits{N.}},
\bauthor{\bsnm{Grover}, \binits{C.E.}},
\bauthor{\bsnm{Gallagher}, \binits{J.P.}},
\bauthor{\bsnm{Wendel}, \binits{J.F.}},
\bauthor{\bsnm{An\'e}, \binits{C.}},
\bauthor{\bsnm{Baum}, \binits{D.A.}}:
\batitle{Reticulate evolution helps explain apparent homoplasy in floral
  biology and pollination in baobabs ({\it adansonia}; {Bombacoideae};
  {Malvaceae})}.
\bjtitle{Systematic Biology}
\bvolume{69}(\bissue{3}),
\bfpage{462}--\blpage{478}
(\byear{2020}).
\doiurl{10.1093/sysbio/syz073}
\end{barticle}
\endbibitem

\bibitem{Sol_s_Lemus_2016}
\begin{barticle}
\bauthor{\bsnm{Sol\'is-Lemus}, \binits{C.}},
\bauthor{\bsnm{Yang}, \binits{M.}},
\bauthor{\bsnm{An\'e}, \binits{C.}}:
\batitle{Inconsistency of species tree methods under gene flow}.
\bjtitle{Systematic Biology}
\bvolume{65}(\bissue{5}),
\bfpage{843}--\blpage{851}
(\byear{2016}).
\doiurl{10.1093/sysbio/syw030}
\end{barticle}
\endbibitem

\bibitem{solislemus2020_identifiability}
\begin{botherref}
\oauthor{\bsnm{Sol{\'\i}s-Lemus}, \binits{C.}},
\oauthor{\bsnm{Coen}, \binits{A.}},
\oauthor{\bsnm{An\'e}, \binits{C.}}:
On the identifiability of phylogenetic networks under a pseudolikelihood model.
arXiv
(2020).
\doiurl{10.48550/arxiv.2010.01758}
\end{botherref}
\endbibitem

\bibitem{huson_rupp_scornavacca_2010}
\begin{bbook}
\bauthor{\bsnm{Huson}, \binits{D.H.}},
\bauthor{\bsnm{Rupp}, \binits{R.}},
\bauthor{\bsnm{Scornavacca}, \binits{C.}}:
\bbtitle{Phylogenetic Networks: Concepts, Algorithms and Applications}.
\bpublisher{Cambridge University Press},
\blocation{Cambridge}
(\byear{2010}).
\doiurl{10.1017/CBO9780511974076}
\end{bbook}
\endbibitem

\end{thebibliography}

\end{document}